\documentclass[keeplastbox]{style/vldb}
\usepackage{amssymb}
\usepackage{amsmath}
\usepackage{blkarray}
\usepackage{color}
\usepackage{times}
\usepackage{cite}
\usepackage{enumitem}
\usepackage{arydshln}
\usepackage{url}
\usepackage{algorithm}
\usepackage{algorithmic}

\usepackage{flushend}
\usepackage{subcaption}
\usepackage{balance}
\usepackage{enumitem}

\newcommand{\eat}[1]{}

\newtheorem{theorem}{Theorem} 

\newtheorem{definition}{Definition}
\newtheorem{example}{Example}
\newtheorem{problem}{Problem}

\usepackage{xspace}
\newcommand{\getnext}{{\sc get-next}\xspace}
\newcommand{\getnexttd}{{\sc get-next$_{2D}$}\xspace}
\newcommand{\getnextr}{{\sc get-next$_r$}\xspace}

\newcommand{\raysweeping}{{\sc raysweeping}\xspace}
\newcommand{\getnexttb}{{\sc get-next$_{md}$}\xspace}
\newcommand{\svtd}{{\sc sv$_{2D}$}\xspace}
\newcommand{\sv}{{\sc sv}\xspace}

\newcommand{\stitle}[1]{\vspace{1ex}\noindent{\bf #1}}

\usepackage{graphicx}
\usepackage{comment}
\usepackage[colorinlistoftodos]{todonotes}
\usepackage[colorlinks=true, allcolors=blue]{hyperref}

\newcommand\jag[1]{\textcolor{red}{(Jag) #1}}
\newcommand\julia[1]{\textcolor{magenta}{(Julia) #1}} 
\newcommand\abol[1]{\textcolor{orange}{(Abol) #1}}

\newcommand\new[1]{#1}
\newcommand\techrep[1]{#1}%\textcolor{green}{#1}}%\textcolor{orange}{(techrep) #1}}
\newcommand\submit[1]{}%\textcolor{red}{(submit) #1}}
\newcommand\revision[1]{#1}
\techrep{
\makeatletter
\def\@copyrightspace{\relax}
\makeatother
}

\vldbTitle{On Obtaining Stable Rankings}
\vldbAuthors{Abolfazl Asudeh, H. V. Jagadish, Gerome Miklau, Julia Stoyanovich}
\vldbDOI{https://doi.org/10.14778/3291264.3291269}
\vldbVolume{12}
\vldbNumber{3}
\vldbYear{2018}

\title{On Obtaining Stable Rankings\titlenote{\small{This work was supported in part by NSF Grants No. 1741022, 1741254, 1741047, and 1250880.
\techrep{\\ This is the technical report.}
}}}
\author{
	\alignauthor
	Abolfazl Asudeh$^\dag$, H. V. Jagadish$^\dag$, Gerome Miklau$^{\dag\dag}$, Julia Stoyanovich$^\ddag$ \\
	\affaddr {$^\dag$University of Michigan, $^{\dag\dag}$University of Massachusetts Amherst, $^\ddag$New York University}
	{\email{
		$^\dag\{$asudeh, jag$\}$@umich.edu, $^{\dag\dag}$miklau@cs.umass.edu, $^\ddag$stoyanovich@nyu.edu
        }
	}
}

\hypersetup{draft}

\begin{document}
\maketitle
\begin{abstract}
Decision making is challenging when there is more than one criterion to consider.
In such cases, it is common to assign a goodness score to each item as a weighted sum of its attribute values and rank them accordingly.
% It is often necessary to rank items with multiple attributes.
% A typical method to achieve this is to compute a goodness score for each item as a weighted sum of its attribute values, and then to rank by sorting on this score.
Clearly, the ranking obtained depends on the weights used for this summation.
Ideally, one would want the ranked order not to change if the weights are changed slightly. We call this property {\em stability} of the ranking.
A consumer of a ranked list may trust the ranking more if it has high stability.
A producer of a ranked list prefers to choose weights that result in a stable ranking, both to earn the trust of potential consumers and because a stable ranking is intrinsically likely to be more meaningful.

In this paper, we develop a framework that can be used to assess the stability of a provided ranking and to obtain a stable ranking within an ``acceptable'' range of weight values (called ``the region of interest'').
We address the case where the user cares about the rank order of the entire set of items, and also the case where the user cares only about the top-$k$ items.
Using a geometric interpretation, we propose algorithms that produce stable rankings.
% We also develop a randomized algorithm that uses Monte-Carlo estimation. To do so, we first propose an unbiased sampler that draws rankings (or top-$k$ results) uniformly at random from the region of interest.
In addition to theoretical analyses, we conduct extensive experiments on real datasets that validate our proposal.
\end{abstract}

\section{Introduction}
\label{sec:intro}
%\julia{We use, ``ranking'', ``rank order'' and ``ranked order'' to refer to the same concept.  Should we stick to one or two of these?}
It is often useful to rank items in a dataset. It is straightforward to sort on a single attribute, but that is often not enough. When the items have more than one attribute on which they can be compared, it is challenging to place them in ranked order.
Consider, for example, the problem of ranking computer science departments.  Various entities, such as U.S. News and World Report, Times Higher Education, and the National Research Council, produce such rankings.  
These rankings are impactful, yet heavily criticized.
Several of these rankings have deficiencies in the attributes they choose to measure and in their data collection
%\jsref{data collection}{is this what you mean, Jag?} \jag{I meant more, but fixed it now by adding a few words before the jsref} 
methodology, not of relevance to our paper now.
Our concern is that even if these deficiencies were addressed, we are compelled to obtain a single score/rank for a department by combining multiple objective measures, such as publications, citations, funding, and awards.  Different ways of combining values for these attributes can lead to very different rankings.  There are similar problems when we want to rank/seed sports teams, rank order cars or other products, as Malcolm Gladwell has nicely described \cite{Gladwell11Order}.

Differences in rank order can have significant consequences. 
For example, a company may promote high-ranked employees and fire low-ranked employees. %Even in less critical scenarios, such as I
In university rankings, it is well-documented that the ranking formula has a significant effect on policies adopted by universities~\cite{bowman2011anchoring, monks1999impact}.
In other words, it matters how we choose to combine values of multiple attributes into a scoring formula.
%
%\julia{The connection between assigning letter grades and positions in a ranking is very strong: assigning students to letter grades generates a partial order, of which a total order (a ranking) is a special case.}
Even when there is lack of consensus on a specific way to combine attributes, we should make sure that the method we use is robust: it should not be the case that small perturbations, such as small changes in parameter values, can change the rank order. %Imagine you are a professor grading on a curve, with flexibility to choose the threshold scores separating letter grades. Having sorted the students based on score, you will prefer to place each separating threshold where there is a larger gap between successive scores in the approximate score range you consider appropriate.\abol{not sure if this example directly translates to our formulation?}  \julia{I suggest that we remove this example, it exposes us to criticism that we don't handle the problem that this example raises, which, arguably, is more natural.}\abol{agreed} The motivation here is robustness of the letter grade: not to have a student's grade change based on a small difference in score.  While the problem we address in this paper concerns the stability of a ranking and not of a grade, this example illustrates the importance of robustness in cases where decisions are made based on ordered data.

%Our problem in this paper is slightly different--stability of a rank order and not stability of letter grade.  We present the professor grading example just to illustrate the value of stability.

In this paper we address the following problem: Assume that a linear combination of the attribute values is used for assigning a score to each item; then items are sorted to produce a final ranked order.  We want this ranking to be {\em stable} with respect to changes in the weights used in scoring.  Given a particular ranked list of items, one question a consumer will ask is: how robust is the ranking?   If small changes in weights can change the ranked order, then there cannot be much confidence in the correctness of the ranking.  

A given ranking of a set of items can be generated by many weight functions.  Because this set of functions is continuous, we can think of it as forming a region in the space of all possible weight functions.  We call a ranking of items {\em stable} if it is generated by a weight function that corresponds to a large region of this space.  

%\julia{I'd like to change the order of the following two paragraphs: "Note that..." goes after "A given ranking..."}

Note that if some items are very close in score, it is possible that small changes to attribute values can change their relative ordering.  Such effects tend to be local, indicating that the affected items are effectively ``tied'' so that the change in ranking is merely a breaking of the tie.  Past work~\cite{NumberOne} has considered the implications of data uncertainty and sensitivity of rankings to imprecision; it is not our focus here.
Instead, we address a much bigger problem, that of changes in the ranking even without any change to the attribute values, but due to a small change in the weighting function used to compute item scores.  Such global changes can drastically affect the ranked order, with far-reaching economic and societal effects~\cite{Gladwell11Order}. %, as eloquently described in~\cite{Gladwell11Order}.

%It is therefore supported by a large set of functions and \gmref{not affected by small changes to the weights}{this part of the explanation may raise some questions: recall our discussions about weight functions near the boundary.}.

Stability is a natural concern for consumers of a ranked list (i.e. those who use the ranking to prioritize items and make decisions), who should be able to assess the magnitude of the region in the weight space that produces the observed ranking.  If this region is large, then the same ranked order would be obtained for many choices of weights, and the ranking is stable.  But if this region is small, then we know that only a few weight choices can produce the observed ranking.  This may suggest that the ranking was engineered or ``cherry-picked'' by the producer to obtain a specific outcome. 

Data scientists often act as producers of ranked lists (i.e. they design weight functions that result in ranked lists), and desire to produce stable results.  We argued in~\cite{sigmod2018demo} that stability in a ranked output is an important aspect of algorithmic transparency, because it allows the producer to justify their ranking methodology, and to gain the trust of consumers.  % Our goal is to help producers of ranked lists to devise stable rankings.  
Of course, stability cannot be the only criterion in the choice of a ranking function: the result may be weights that seem ``unreasonable'' to the ranking producer.  To support the producer, we allow them to specify a range of reasonable weights, or an {\em acceptable region} in the space of functions, and help the producer find stable rankings within this region. 

%\julia{I'd like to bring up transparency here, which concerns (a) making explicit and (b) justifying the ranking methodology.  Stability is an important component of this justification.}

Our work is motivated by the lack of formal understanding of ranking stability and the consequent lack of tools for consumers and producers to assess this critical property of rankings.  We will show that stability hinges on complex geometric properties of rankings and weight functions.  We will provide a novel technique to identify stable rankings efficiently.

Our technique 
%the \gmref{system}{is system the right word? I'm not sure we're proposing a system. Perhaps instead techniques? tools?}
%\jag{Let's call it "technique"}
does not stop at proposing just the single most stable choice, or even the $h$ most stable choices for some pre-determined fixed value of $h$.  Rather, it will continue to propose weight choices,
%\gmref{weight choices}{this is not really true is it? we don't propose particular weight functions do we?} \jag{Well, we do propose regions in weight space, and each region comprises an infinite number of weight choices. This is technically not a specific weight choice, but we probably don't want to get into weight regions yet??}, 
and the corresponding rank ordering of items, beginning with the most stable in the specified region of interest, and continuing in decreasing order of stability, until the user finds one that is satisfactory.

Alternatively, our technique can provide an overview of all the rankings that occupy a large portion in the acceptable region, and hence are stable, along with an indication of the fraction of the acceptable region occupied by each.  Thereby, the user can see at a glance what the stable options are, and also how dominant these are within the acceptable region.

We now motivate our techniques with an example.
%\gm{The example below should explain clearly what our technologies can do for consumers/producers.}
% Cornell and Toronto swap 10 vs 11 (among universities): Cornell is at 10 in the stable ranking, Toronto is at 11.  Stability is 0.020555004 vs. 0.003229 (close to uniform 0.00297619).
% Consumer is Toronto, asks about the ranking.

\begin{example}
\label{ex:1}
CSMetrics~\cite{csmetrics} ranks computer science research institutions based on the measured ($M$) and predicted ($P$) number of citations. These values are appropriately scaled and used in a weighted scoring formula, with parameter $\alpha \in [0,1]$ that sets their relative importance (see \S~\ref{sec:exp:setup} for details). 
CSMetrics includes a handful of companies in its ranking, but we focus on academic departments in this example.

%The score of each department is computed using the scoring formula $\alpha\log(M + \epsilon) + (1-\alpha)\log(P + \epsilon)$, where $\alpha$ is chosen by the user and  defaults to 0.3.  

As $\alpha$ ranges from 0 to 1, CSMetrics generates 336 distinct rankings of the top-$100$ departments.  Assuming (as a baseline) that each ranking is equally likely, we would expect an arbitrarily chosen ranking to occur 0.3\% of the time, which we take to mean that it occupies 0.3\% of the volume in the space of attributes and weights. We formalize this in \S~\ref{sec:stability} and call it {\em stability of a ranking}.

%We will formalize our notion of stability based on this intuitive notion, and will relate it to the volume that a ranking occupies in the space of attributes and weights.

Suppose that the ranking with %the default 
$\alpha=0.3$ is released, placing Cornell (a consumer) at rank 11, just missing the top-$10$. Cornell then checks the stability of the ranking (see \S~\ref{sec:stability:cons}), and learns that it's value is 0.3\%, matching that of the uniform baseline.  With this finding, Cornell asks CSMetrics to justify its choice of $\alpha$.

CSMetrics (the producer) can respond to Cornell by {\new further interrogating, and potentially} revising the published ranking.  It first enumerates stable regions (see \S~\ref{sec:stability:prod}) and finds that the most stable ranking indeed places Cornell at rank 10 (switching with the University of Toronto), and represents 2\% of the volume --- an order of magnitude more than the reference ranking. However, this stable ranking is very far from the default, placing more emphasis on measured citations with $\alpha=0.608$.  If this is unsatisfactory, CSRankings can propose another ranking closer to the reference ranking, but with better stability (see \S~\ref{sec:stability:reg}). Interestingly, Cornell also appears at the top-$10$ in the most stable ranking that is within 0.998 cosine similarity from the original scoring function. 
%making the choice of $\alpha=0.3$ difficult to justify.
\submit{\vspace{-2mm}}
%{\bf Draft} Producer take 1: Stability is the primary objective of the producer, he executes Get-Next, doesn't like its region, executes Get-Next again, is happy after a couple of calls.  This is stability-guided exploration of the space of rankings.

%{\bf Draft} Producer take 2: Region of interest is the starting point, take the most stable function within that region. I can dream up a region once we decide what attributes to use.
\end{example}

\eat{
For example, for a manager of a flight agency that wants to rank its employee, the choices of weights on the ``customer satisfaction'' and ``travel sales'' as long as those are relatively similar are equally good. In another example, a customer who wants to buy a call phone may accept any weights that consider the battery life-time to be more important that the camera quality.
On the other hand, small changes in the weight vector may dramatically change the ranked list -- making it very dependent on the choice of weights.
}

\eat{
Therefore, this paper initiates the research in addressing this need.
Especially, rather than limiting the view of the user to a fixed number of rankings, we propose a \getnext operator that upon calling it returns the next most stable ranking along with its measure of stability to the user.
As a result, the user may issue a \getnext to get the next most stable ranking, and if it does not satisfy her preference, she may continue finding the next rankings until she is satisfied with the results.
Carefully providing the problem setting, we discuss the geometric interpretations that help us developing the algorithms for the problem.
}

%Our contributions in this paper are summarized as following:
%\jag{Abol, we need to redo the list of contributions so that (1) they do not use terminoogy not yet defined above, and (2) reflect the 4 problem structure that we develop in the next section.  COuld you please rewrite this list, beginning with contribution 4?  (I rewrote the first three)}
%\stitle{Contributions:}
\vspace{1ex}
Our contributions include the following:
%\submit{\begin{itemize}[leftmargin=*]}
%\techrep{\begin{itemize}}
\begin{itemize}[leftmargin=*,itemsep=0pt]
\item We formalize a novel notion of the {\em stability} of a ranking, for rankings that result from a linear weighting of item attribute values. Stability captures the tolerance to changes in the weights. 
\item We propose algorithms that enable the efficient testing of ranking stability as well as the enumeration of the most-stable rankings, optionally constrained by a set of acceptable scoring functions.  We propose both exact algorithms and approximation algorithms that are based on novel sampling methods.

\item Our empirical evaluation demonstrates the efficiency of our techniques on real and synthetic datasets, and investigates the stability of real published rankings of computer science departments, soccer teams, and diamond retailers. \new{We show that existing rankings in these domains are often unstable and that favoring stability can sometimes have a significant impact on the rank of some items.
For instance, our findings cast doubt on the validity of the FIFA rankings which are used in making important decisions such as seeding the World Cup final draws.}
%\gm{update to reflect findings.}
\end{itemize}

%\stitle{Paper organization:} In Section ... \gm{TO DO}

\section{Problem Setup}
\subsection{Preliminaries} \label{sec:pre}
%In this section we describe our data and rankings models, and the geometric interpretations that are the foundation for our algorithms.
\begin{figure*}[t!]
\centering
\subcaptionbox{\label{fig:toy1} A sample database, $\mathcal{D},$ of items with scoring attributes $x_1$ and $x_2$; and the result of scoring function $f=x_1+x_2$.}
{	
\begin{tabular}{|l|c|c||@{}c@{}|}
	\hline
	\multicolumn{3}{|c||}{$\mathcal{D}$} & $f$ \\ \hline
	id   & $x_1$ & $x_2$ & $ \; x_1+x_2 \;$ \\ \hline \hline
	$t_1$& 0.63 & 0.71&1.34 \\ \hline
	$t_2$& 0.83 & 0.65&1.48 \\ \hline
	$t_3$& 0.58 & 0.78&1.36 \\ \hline
    $t_4$& 0.7 & 0.68&1.38 \\ \hline
	$t_5$& 0.53 & 0.82&1.35 \\ \hline
	\end{tabular}
	\vspace{7ex}
	}
\hfill 
\subcaptionbox{\label{fig:toy1-2} The {\em original space:} each item is a point. A scoring function is a ray ($f = x_1+x_2$ is shown) which induces a ranking of the items by their projection.}{\includegraphics[width=0.32 \textwidth]{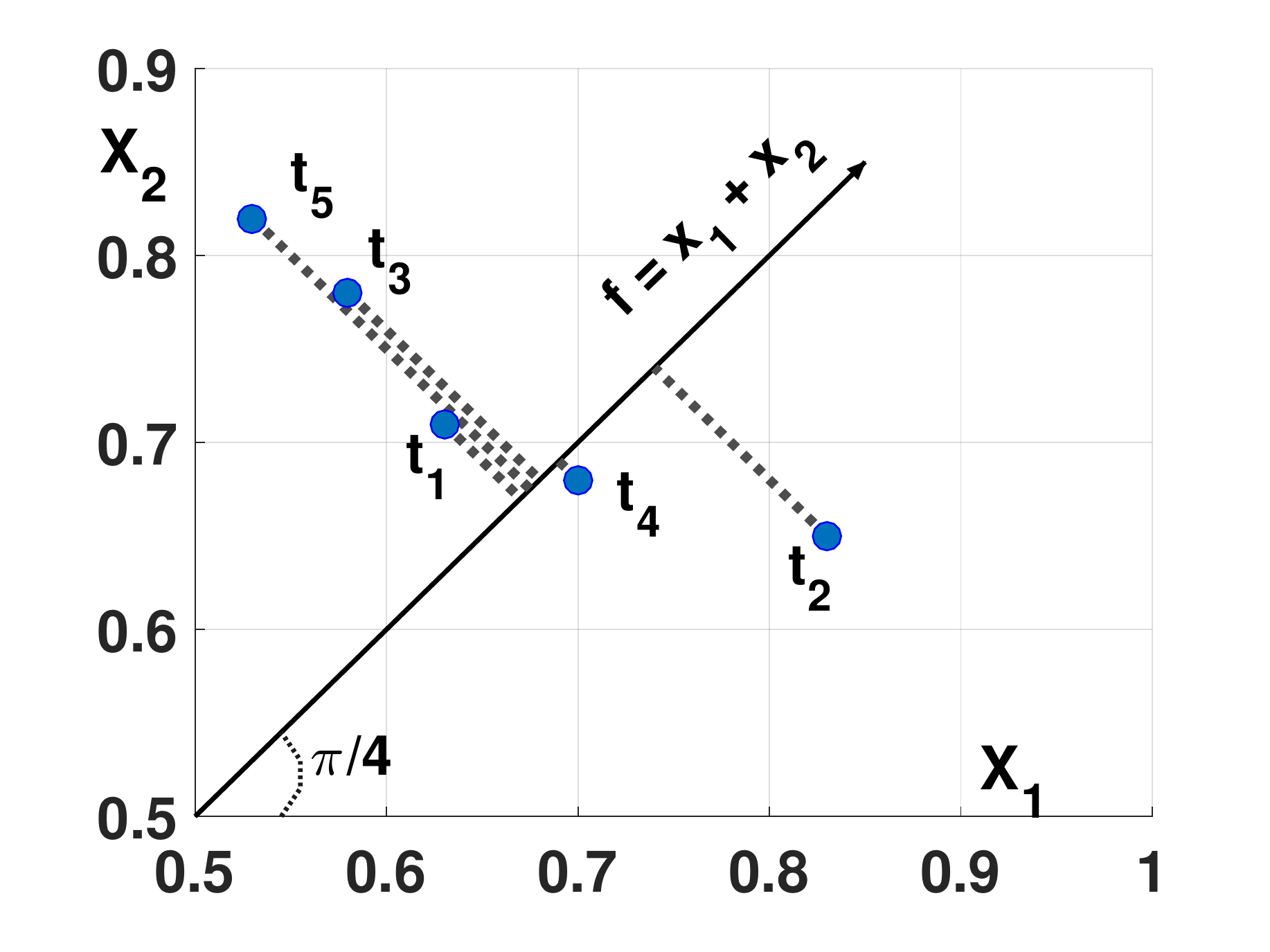}}
\hfill 
\subcaptionbox{
\label{fig:toy3} 
\revision
The {\em dual space:} items are the hyperplanes (lines here). Each scoring function is a ray; within a region bounded by the intersections of the item hyperplanes, all scoring functions induce the same ranking.}
{\includegraphics[width=0.33 \textwidth]{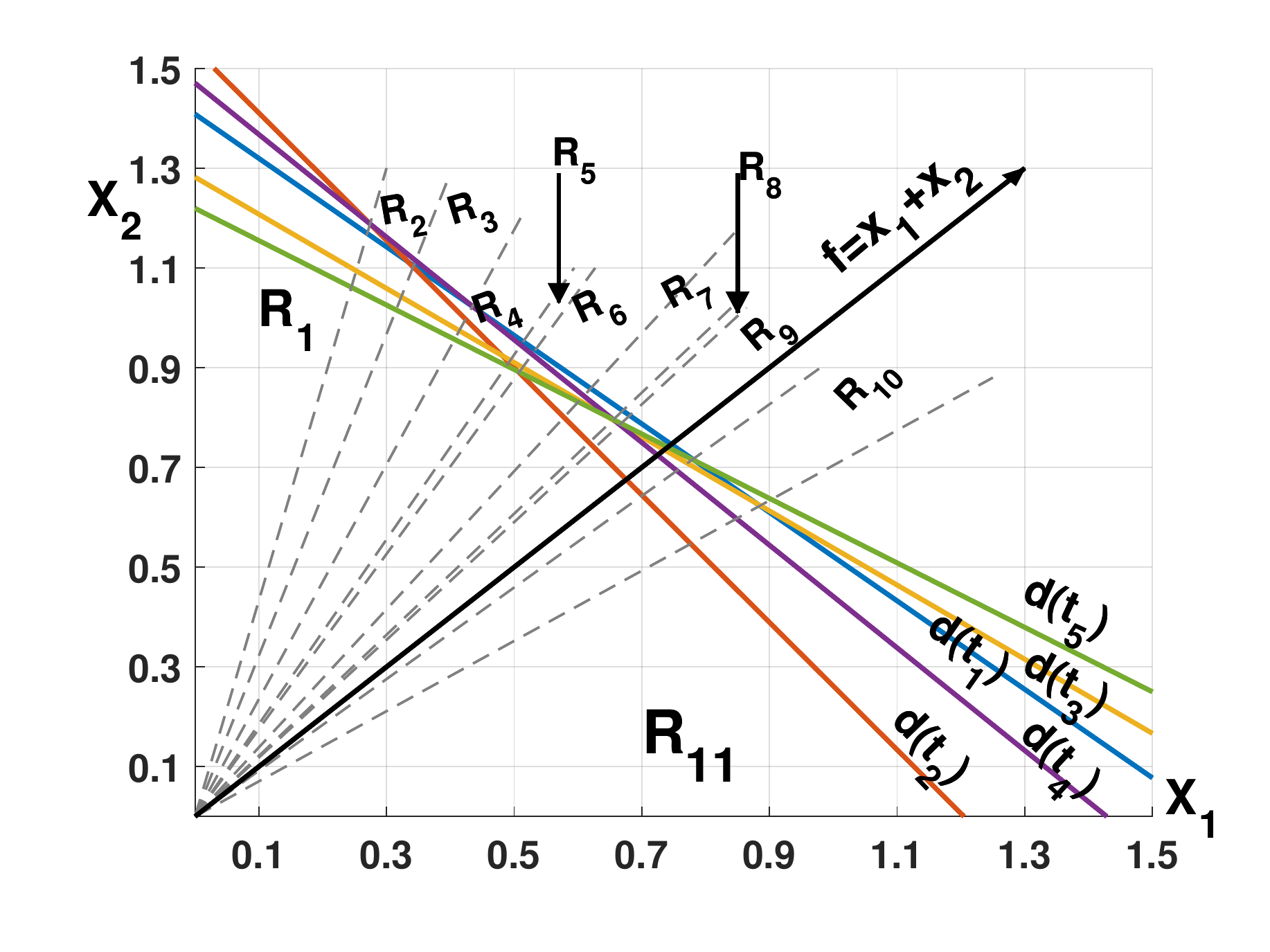}
}
\submit{\vspace{-3mm}}
\caption{
\label{fig:prelim}
A sample database and its geometric interpretation in the original space and dual space.
}
\submit{\vspace{-5mm}}
\end{figure*}

\subsubsection{Data model and rankings}

We consider a fixed database $\mathcal{D}$ consisting of $n$ items, each with $d$ scalar scoring attributes. In addition to the scoring attributes, the dataset may contain non-scoring attributes that are used for filtering, but they are not our concern here.  Thus we represent an item $t \in \mathcal{D}$ as a $d$-length vector of scoring attributes, $\langle t[1], t[2], \ldots, t[d] \rangle$.  Without loss of generality, we assume that the scoring attributes have been appropriately transformed: normalized to non-negative values between 0 and 1, standardized to have equivalent variance, and adjusted so that larger values are preferred. 
{\revision
Note that this normalization is not strictly required for the techniques we propose in the paper but they make the comparison between attribute value weights fair and our stability measure more meaningful.
}

%we assume that each scoring attribute is a , and that larger values are preferred.  Ideally, each scoring attribute also has the same variance.  These assumptions are usually straightforward to meet by applying an appropriate simple transformation function to values of an attribute.  
%They can be relaxed if need be, with some additional notation and bookkeeping. 
%

% GM: I removed this discussion about identical scores.  
%
%Further, we assume that no two items are identical on every scoring attribute.
%\abol{do we need to mention this? In fact the items that have equal ranking scores fall in the same ``bucket'' and they are ranked equally good. So, if we want to keep this, we can say something like: this assumption can easily be relaxed by putting the items with the same ranking scores in the same bucket.}
%This assumption is also easy to relax, by adding a check for such equality, collapsing any such identical scoring items into one, and then separating them out at the very end when displaying the ranking. \julia{It's not as straight-forward as that, because distance measures between rankings would need to be adjusted, somewhat non-trivially, to accommodate ties.  See for example~\cite{DBLP:journals/pvldb/BrancotteYBBDH15}.}

We consider rankings of items that are induced by first applying a linear weight function to each item, then sorting the items by the resulting score to form a ranking. 
\begin{definition}[Scoring function]
A {\em scoring function} $f_{\vec{w}}:\mathbb{R}^d\rightarrow\mathbb{R}$, with weight vector $\vec{w} ~=~\langle w_1, w_2,\ldots,w_d \rangle$, assigns the score $f_{\vec{w}}(t) = \Sigma_{j=1}^d w_j t[j]$ to any item $t\in\mathcal{D}$.	
\end{definition}
{\revision Without loss of generality, we assume that $w_j\in\vec{w} \geq 0$.
This assumption is straightforward to relax with some additional notation and bookkeeping.}
When $\vec{w}$ is clear, we denote $f_{\vec{w}}(t)$ by $f(t)$.

We use $\mathcal{U}$ to refer to the set of all possible scoring functions.
Given a score for each item, the ranking of items induced by  $f$ is the permutation of items in $\mathcal{D}$ defined by  sorting them by their scores under $f$, in descending order, and breaking ties consistently by an item identifier.  We use the notation $\mathcal{r}_f(\mathcal{D})$ to denote the ranking of items in $\mathcal{D}$ based on $f$. 

% Julia: I edited, see above
% \gmref{If $f$ results in equal scores for two or more items, we assume they are ordered arbitrarily by an item id field, so that a ranking $\mathcal{r}_f(\mathcal{D})$ does not have ties.}{please check.}

%\gm{there is no definition of "score-based ranking" beyond this language. Should there be?} \julia{I think this is sufficiently precise, but what about ties then?  We do need to say something, both regarding what happens when two items have identical values for all scoring attributes (and so would always be tied), and for cases when a particular scoring function induces a tie between a pair of items.  Even something like ``ties are broken by item id'', and add id to the data model. We don't want to say that ties are broken arbitrarily --- that would make a ranking unstable in a way that we cannot address effectively.}

\begin{example}\label{example1}
The human resources (HR) department of a company wants to prioritize hiring candidates based on two criteria: an aptitude measure $x_1$ (e.g. a score on a qualifying exam) and an experience measure $x_2$ (e.g. the number of years of relevant experience).   
Figure~\ref{fig:toy1} shows the candidates as well as their (normalized) values for $x_1$ and $x_2$. The score for each candidate is also shown, for weight vector $\vec{w}=\langle 1,1\rangle$, computed as $f(t)= x_1+x_2$.
\end{example}

{\revision
Although we restrict our attention to linear scoring functions, our techniques can be used with more general scoring functions by applying non-linear transformations to $\mathcal{D}$ as a preprocessing step.
Consider Example~\ref{example1}, and let the scoring function be $f(t)= x_1+x_2+0.5x_1^2$. The quadratic term $x_1^2$ can be added as  $x_3=x_1^2$. %For example, $t_1[3] = x_1^2 = 0.3969$.
%After adding $x_3$, $f$ can be rewritten as the linear function $x_1+x_2+0.5x_3$.
}

\subsubsection{Geometry of ranked items}
Our algorithms are based on a geometric interpretation of scored items and induced rankings. We now present two geometric views of the database, to which we respectively refer  as (i) the {\em original space}, where every item corresponds to a point, and (ii) the {\em dual space}, where every item corresponds to a hyperplane.

%\paragraph*{The original space}
{\bf The original space.}
The {\em original space} 
consists of $\mathbb{R}^d$ with each item in $\mathcal{D}$ represented as a point, and a linear scoring function $f_{\vec{w}}$ viewed as a ray starting from the origin and passing through the point $\vec{w} = (w_1, w_2,...,w_d)$.  The ranking $\mathcal{r}_{f_{\vec{w}}}(\mathcal{D})$ %score-based ordering of the points induced by $f_{\vec{w}}$ 
corresponds to the ordering of their projections onto this ray. 

Continuing our example, Figure~\ref{fig:toy1-2} shows the items of our sample database in the original space, as points in $\mathbb{R}^2$. The function $f=x_1+x_2$ is shown as a ray passing through $\langle 1,1 \rangle$. 
The projection of the points onto the vector of $f$ specifies the ranking: the further a point is from the origin, the higher its rank.
{\revision The reason is that, for every score $f(t)$, $\sum w_j x_j = f(t)$ is the perpendicular hyperplane to the ray of $f$ that passes through the point $t$.} Hence, looking at Figure~\ref{fig:toy1-2}, the candidates in Example~\ref{example1} are ranked as $\langle t_2, t_4, t_3, t_5, t_1\rangle$ based on $f$.  One can also easily imagine the ranking of items that would result from an extreme scoring function that ranks only by attribute $x_1$ (i.e. $f=x_1$) by considering the projections onto the $x_1$-axis (or respectively for the $x_2$-axis).

% Removed
%\gmref{Note that, for any constant value $c>0$, the vector corresponding to functions $f$ and $f'$ with weights $\vec{w} = \langle w_1, w_2,\ldots, w_d\rangle$ and $\vec{w'} = \langle c.w_1, c.w_2,\ldots,c.w_d\rangle$ are the same.  This is because the scores generated by $f'$ are a linear scaling of the scores generated by $f$, and so both functions induce the same ordering of the items.}{Do we really need this? It should be clear since these two are the same vector.}

Viewing the items from $\mathcal{D}$ in the original space provides clarity about the range of rankings that can be induced by the scoring functions: all scoring functions are defined by rays in the first quadrant of $\mathbb{R}^d$ that is determined by the weight vector. It is sometimes convenient to use polar coordinates to represent a scoring function:  a ray in $\mathbb{R}^d$ starting from the origin (corresponding to function $f_{\vec{w}}$) can be identified by $(d-1)$ angles $\langle \theta_1, \theta_2, \cdots, \theta_{d-1} \rangle$, each in the range $[0,\pi/2]$. Thus, given a function $f_{\vec{w}}$, its angle vector can be computed using the polar coordinates of $w$. 
%Let $w$ have the polar coordinates $(r,\Theta)$. Then the $(d-1)$-dimensional vector $\Theta$ defines the ray of $w$.
For example, function $f$ with weights $\langle 1,1 \rangle$ in Figure~\ref{fig:toy1-2} is identified by a single angle $\langle \pi/4 \rangle$.
There is a one-to-one mapping between these rays and the points on the surface of the origin-centered unit $d$-sphere (the unit hypersphere in $\mathbb{R}^d$), {\revision or to the surface of any origin-centered $d$-sphere}.
Thus, the first quadrant of the unit $d$-sphere represents the universe of functions $\mathcal{U}$.

%\paragraph*{The dual space}
{\bf The dual space.} We are particularly interested in reasoning about the transition points of the weight vector, where we move from one ranking to a different ranking.  The {\em dual space}~\cite{edelsbrunner} consists of $\mathbb{R}^d$, but we represent an item $t$ by a hyperplane $\mathsf{d}(t)$ given by the following equation of $d$ variables $x_1 \dots x_d$:
{\revision
\begin{align}\label{eq:dual}
\mathsf{d}(t):~ t[1]\times x_1 + \dots + t[d]\times x_d = 1
\end{align}
}
Continuing our example, Figure~\ref{fig:toy3} shows the items in the dual space. In $\mathbb{R}^2$, every item $t$ is a 2-dimensional hyperplane (i.e. simply a line) given by $\mathsf{d}(t): t[1] x_1 + t[2] x_2=1$.
In the dual space,  functions are represented by the same ray as in the original space{\revision, passing through the point $\vec{w}$.
Consider the intersection of a dual hyperplane $\mathsf{d}(t)$ with this ray.
This intersection is in the form of $a\times\vec{w}$, because every point on the ray of $f$ is a linear scaling of $\vec{w}$.
Since this point is also on the hyperplane $\mathsf{d}(t)$,
$t[1]\times a\times w_1 + \dots + t[d]\times a\times w_d = 1$. Hence, $\sum t[j] w_j = 1/a$. This means that the dual hyperplane of any item with the score $f(t)=1/a$ intersects the ray of $f$ at point $a\times\vec{w}$.
Following this, the}
ordering of the items based on a function $f$ is determined by the ordering of the intersection of the hyperplanes with the vector of $f$. The closer an intersection is to the origin, the higher its rank.
{\revision
For example, in Figure~\ref{fig:toy3}, the intersection of the line $t_2$ with the ray of $f=x_1 + x_2$ is closest to the origin, and so $t_2$ has the highest rank for $f$.
}
We will show in the next section that the intersections of hyperplanes in the dual space define {\em regions}, within which rankings {\em do not} change under small changes of the weight vector.

\subsection{Stability of a ranking} \label{sec:stability}

We now present our definition of stability and identify the key algorithmic problems for consumers and producers of rankings.

\subsubsection{Definition of stability}

Every scoring function in the universe $\mathcal{U}$ induces a single ranking of the items. But each ranking is generated by many functions.  For database $\mathcal{D}$, let $\mathfrak{R}_\mathcal{D}$ be the set of rankings over the items in $\mathcal{D}$ that are generated by at least one scoring function $f\in\mathcal{U}$, that is, by at least one choice of weight vector. For a ranking $\mathfrak{r}\in\mathfrak{R}_\mathcal{D}$, we define its {\em region}, $R_\mathcal{D}(\mathfrak{r})$,  as the set of functions that generate $\mathfrak{r}$: 
\begin{align}
R_\mathcal{D}(\mathfrak{r}) = \{f ~|~  \mathcal{r}_{f}(\mathcal{D}) = \mathfrak{r}\}
\end{align}
Figure~\ref{fig:toy3} shows the boundaries (as dotted lines) of the regions for our sample database, one for each of the 11 feasible rankings. 

We use the region associated with a ranking to define the ranking's stability. The intuition is that a ranking is stable if it can be induced by a large set of functions.
\techrep{Recall from Section~\ref{sec:pre} that each function is a ray, in both the original and the dual space.}
If the region of a ranking is large, then small changes in the weight vector are not likely to cross the boundary of a region and therefore the ranked order will not change. 
{\revision
For every region $R$, we define its volume, $vol(R)$, to measure the bulk of the region.
Specifically, we use the one-to-one mapping between the surface of the unit $d$-sphere and $\mathcal{U}$ for this purpose.
The volume of a region is the area of the space carved out in the unit $d$-sphere by the set of functions in the region.
}

\begin{definition}[Stability of $\mathfrak{r}$ at $\mathcal{D}$]\label{def:1}
Given a ranking $\mathfrak{r}\in\mathfrak{R}_\mathcal{D}$, the stability of $\mathfrak{r}$ is the proportion of ranking functions in $\mathcal{U}$ that generate $\mathfrak{r}$.  That is, stability is the ratio of the volume of the ranking region of $\mathfrak{r}$ to the volume of $\mathcal{U}$. Formally:
\begin{align}
S_\mathcal{D}(\mathfrak{r}) = \frac{\mbox{vol}(R_\mathcal{D}(\mathfrak{r}))}{\mbox{vol}(\mathcal{U})}
\end{align}
\end{definition}

We emphasize that stability is a property of a ranking (not of a scoring function) and it holds for a particular database, as indicated by the notation $S_\mathcal{D}(\mathfrak{r})$.
{\revision
For ease of notation, we denote $S_\mathcal{D}(\mathfrak{r})$ and $R_\mathcal{D}(\mathfrak{r})$ with $S(\mathfrak{r})$ and $R(\mathfrak{r})$, respectively, in the rest of this paper.
}

{\revision
In the following, we  % in \S~\ref{sec:stability:reg}, we 
define the scope for studying the stability of rankings% (\S~\ref{sec:stability:reg})
, and
we develop three alternative problems that build on the notion of stability in Definition~\ref{def:1} and correspond to the views of two different stakeholders: consumers and producers of rankings.
}

\subsubsection{Acceptable scoring functions}\label{sec:stability:reg}
{\revision
When generating a ranking, the producer will often need to consider trade-offs between the choice of an acceptable scoring function and the stability of the generated ranking.  Stable rankings are preferable because they are robust to small changes in scoring function weights.
%producers prefer to present robust rankings rather than those that might change significantly under small changes of the scoring function.  
Furthermore, to the extent that consumers trust more stable rankings, producers are interested in earning this trust. Still, stability is not the only concern for the producer. We return to our running example to motivate this point.
}
%The producer of a ranking may also care about stability because they prefer to present robust rankings rather than those that might change significantly under small changes of the scoring function.  Furthermore, to the extent that consumers trust more stable rankings, producers must try to earn this trust by creating stable rankings.  Of course, stability is not the only concern for the producer, and will often be traded off with the choice of a scoring function that is acceptable to the producer a priori. 

\begin{example}\submit{\vspace{-2mm}}
In producing a ranking of employees, an HR officer believes that aptitude ($x_1$) should be twice as important as experience ($x_2$), but this is only a rough guideline. Any weight with a ratio within 20\% of $2$ is acceptable. By testing different weights within this acceptable range, the officer observes different rankings of candidates and selects one that maximizes stability.
\submit{\vspace{-1mm}}
\end{example}

\def\ar{\mathcal{U}^*}

We allow producers to constrain the scoring function by specifying an {\em acceptable region}, denoted $\ar \subseteq \mathcal{U}$, in one of two ways:
%\submit{\begin{itemize}[leftmargin=*]}
%\techrep{\begin{itemize}}
\begin{itemize}[leftmargin=*,itemsep=0pt]
\item {\em A vector and angle distance}\footnote{\small Note that this can be expressed by cosine similarity.}: the acceptable region is identified by a hypercone around the central ray defined by the weight vector. For example, a user may equally prefer any function that has at most $\pi/10^\circ$ angle distance (at least 95.1\% cosine similarity) with the function $f$ with weight vector $\langle 1,1 \rangle$.
\item {\em A set of constraints}: the acceptable region is a convex region identified by a set of inequalities. For example, if the user is interested in the functions that weigh $x_2$ no greater than $x_1$, then the acceptable region is constrained by $w_2\leq w_1$.
\end{itemize}

%Therefore, considering the fact that small changes in the weights may change the ranking drastically, she wants to find out the ranking(s) that appear most of the times.

%While $\mathcal{U}$ identifies an infinite number of possible functions, the number of rankings of $n$ items bounded by $n!$.  

%\julia{Why did we drop $\mathcal{D}$ from the subscripts?  Let's bring it back, so that what's below doesn't differ from Definition 2 in a non-obvious way.  Also equations are typeset differently.}\abol{I suggest adding a footnote in definition 1 and simplify the notation to not having subscript}

We incorporate the notion of an acceptable region into the definition of stability in a natural way.  
%\julia{I think this notation is unnecessary.  I'd just say: Let $\mathfrak{R}_i$ be the set of rankings  that are generated by at least one function $f\in\mathcal{U}_i$.  Then you don't need to specialize the definitions of stability explicitly.  Otherwise it's confusing that $\mathfrak{r}$ is already drawn from $\mathcal{U}_i$, but still we state $R(\mathfrak{r}, \mathcal{U}_i)=\ldots$.  So, I suggest having just one definition of stability (Def. 1).  This will also save space.}
Let $\mathfrak{R}^*$ be the set of rankings that are generated by at least one function $f\in\ar$.  The ranking region in $\ar$ of a ranking $\mathfrak{r}\in \mathfrak{R}^*$ is: $R^*(\mathfrak{r}) = \{f \in \ar~|~  \mathcal{r}_{f}(\mathcal{D}) = \mathfrak{r}\}$.
%For the ease of notation, when $\mathcal{U}_i$ is clear, we use  $R(\mathfrak{r})$ to refer to $R(\mathfrak{r}, \mathcal{U}_i)$.
%\footnote{\abol{do not forget to prove that each ranking regions is a convex polygon while every ranking is generated by at most one region.}}
Accordingly, we modify the definition of stability of a ranking $\mathfrak{r}\in\mathfrak{R}^*$ to be:
\submit{
$S(\mathfrak{r}) =\mbox{vol}({R^*}(\mathfrak{r}))/\mbox{vol}(\ar).$
}
\techrep{
$$
S(\mathfrak{r}) = \frac{\mbox{vol}({R^*}(\mathfrak{r}))}{\mbox{vol}(\ar)}
$$
}

%\begin{definition}[Stability]\label{def:3}
%Given a ranking $\mathfrak{r}\in\mathfrak{R}$, the stability of $\mathfrak{r}$ is the volume ratio of its region in $\mathcal{U}_i$ to the one of the region of interest. Formally:
%\end{definition}
%Please note that global stability is a special case of stability in which $\ar=\mathcal{U}$. For the ease of notations, in the rest of paper we use $S(\mathfrak{r})$ to refer to $S(\mathfrak{r},\mathcal{U}_i)$ when $\mathcal{U}_i$ is clear.

\subsubsection{Consumer's stability problem}\label{sec:stability:cons}

The basic problem for the consumer is {\em stability verification}, where the consumer seeks to validate the stability of a given ranking. A ranking with higher stability will be more robust and is less likely to be the result of an engineered scoring function.

\begin{problem}[Stability verification]
For dataset $\mathcal{D}$ with $n$ items over $d$ scoring attributes, and ranking $\mathfrak{r}\in\mathfrak{R}$ of the items in $\mathcal{D}$, compute the ranking region $R_\mathcal{D}(\mathfrak{r})$ and its stability $S_\mathcal{D}(\mathfrak{r})$.
\label{pr:p1}
\end{problem}

\subsubsection{Producer's stability problems}\label{sec:stability:prod}
With all of the above machinery in place, we can return to helping the producer of a ranking choose one that is stable.  To this end, we develop two related methods for a producer to explore stable rankings.  We state these problems with respect to an acceptable region $\ar$ and set $\ar=\mathcal{U}$ when all scoring functions are acceptable.

First, the producer may wish to enumerate rankings, prioritizing those that are more stable.  Below we consider an enumeration of rankings in order of stability, with stopping criteria based either on a stability threshold or on a bound on the number of desired rankings. 

\begin{problem}[batch stable-region enumeration]
For a dataset $\mathcal{D}$ with $n$ items over $d$ scoring attributes, a region of interest $\ar$ (specified either by a set of constraints or by a vector-angle), and a stability threshold $s$ (resp. a value $h$), find all rankings $\mathfrak{r}\in\mathfrak{R}^*$ such that $S(\mathfrak{r}) \geq s$ (resp. the top-$h$ stable rankings). 
\end{problem}

In many scenarios, rather than enumerating rankings, the producer may wish to incrementally generate stable regions, in the order of their stability,
{\revision 
using the \getnext primitive.
%Every call to \getnext returns the next stable region, enabling their incremental generation.  
So, the $h$-th call to \getnext will return the $h$-th most stable ranking in  $\ar$.
%Therefore, the $h$-th call of the primitive, after the top-$(h-1)$ stable rankings in $\ar$ being discovered by the previous calls of the primitive, returns the $h$-th stable ranking:
}

\begin{problem}[iterative stable-region enumeration]
\label{prob:get-next}
For a dataset $\mathcal{D}$ with $n$ items over $d$ scoring attributes, a region of interest $\ar$, specified either by a set of constraints or by a vector-angle, and the top-$(h-1)$ stable rankings in $\ar$, {\revision discovered in the previous iterations of the problem,} find the $h$-th stable ranking $\mathfrak{r}\in\mathfrak{R}$. That is, find: 
\begin{align}
\underset{\mathfrak{r}\in ~ \mathfrak{R}\backslash top-(h-1)}{\mbox{argmax}}\big( S(\mathfrak{r})\big)
\end{align}
\end{problem}

Of course, the two enumeration problems are closely related. In fact,  an algorithm for iterative ranking enumeration can be used directly for batch ranking enumeration, if it's called multiple times. In our algorithmic contributions we focus on efficiently evaluating an operator we call \getnext, which can be used to solve both enumeration problems. %, which is describe in \S~\ref{sec:2d}, for the case of two scoring attributes (2D), and \S~\ref{sec:md} for the general case.

{\revision
In the above, for convenience, we relate the stability enumeration to the producers and stability verification to the consumers of rankings. However, a producer can use verification for testing the stability of a ranking, while a consumer can use enumeration for identifying stable rankings.
}

\subsubsection{Stability \revision{of the top-$k$} items}
%\abol{I removed the phrase "partial ranking" all over the paper. Please double check}

So far we focused on complete rankings of $n$ items in $\mathcal{D}$.  However, when $n$ is large, one may be interested in only the highest-ranked $k$ items, for $k <\!\!< n$.  
%\julia{This is a weak point, no need to stress.}\abol{OK}
%{\revision
%When $n$ is large, the regions of all rankings may be small and unstable. As a result, for large settings, it is reasonable not to consider the complete ranking between the items.
%}
This motivates us to reformulate the problems above, focusing on the top-$k$ portion of the ranked list.

{\revision 
We consider two notions of stability of the top-$k$ items.  With the first, weight vectors $\vec{w}$ and $\vec{w}'$ are said to generate the same result if they produce the same {\em set} of top-$k$ items, while with the second, $\vec{w}$ and $\vec{w}'$ must both select the same set of top-$k$ items and return them in the same {\em order}.  
}
%{\revision
%We consider two definitions of the partial rankings: (i) top-$k$ partial rankings, in which two weight vectors $\vec{w}$ and $\vec{w}'$ generate the same result, if the {\em ordering} between their top-$k$ items is the same, and (ii) top-$k$ sets, where $\vec{w}$ and $\vec{w}'$ generate the same result, if the {\em set} of the top-$k$ items based on them is the same. The Stability of partial rankings is defined the same as Definition~\ref{def:1}, except that rather the complete rankings, the partial rankings are considered.
%}
% One challenge is that, while ranking regions are guaranteed to be convex for complete rankings, they may not be convex for top-$k$ partial rankings. 
We present sampling-based randomized algorithms that support top-$k$ partial rankings in \S~\ref{sec:randomized}.

{\revision We will discuss the relationship between our approach and the rich body of work on top-$k$ processing and skyline queries in Section~\ref{sec:related}.  
Here we note that the set of most-stable top-$k$ items is in general different from the skyline~\cite{skylineoperator}, or 
any of its subsets%, proposed by the works on selecting a representative subset of the skyline
~\cite{lin2007selecting,nanongkai2010,chester2014,su2010top}.
The key difference is that the stable top-$k$ items are not necessarily a subset of the skyline. Yet, these items are of high quality and so are potentially of interest to the user. Consider the toy example  $\mathcal{D}=\{t_1(1,0),t_2(.99,.99),t_3(.98,.98),t_4(.97,.97),t_5(0,1)\}$.\\ The skyline of this dataset is $\{t_1,t_2,t_5\}$, while the most stable top-$3$ items are $\{t_2,t_3,t_4\}$.   Of these, only $t_2$ is part of the skyline. %, and so $t_1$ and $t_5$ would not belong to any of the skyline representatives.
}

\section{Two dimensional (2D) Ranking}\label{sec:2d}
To develop our intuition, we start with the case of $d=2$ scoring attributes. %(the 2D case).
Using the geometric interpretation of items provided in \S~\ref{sec:pre} while considering the dual representation of the items, we propose exact algorithms for stability verification and enumeration.
%that runs in polynomial time in the size of the database ($n$) for 2D. 

\techrep{
Consider the items in Example~\ref{example1}, as shown in Figure~\ref{fig:toy1}.
While the number of ranking functions, rays between $0^\circ$ and $\pi /2^\circ$, is infinite, the number of possible orderings between the items is (combinatorially) limited to $n!$.
The number of possible orderings based on linear ranking functions is even less.
}

Consider a pair of items $t_i$ and $t_j$ presented in the dual space in $\mathbb{R}^2$.
Recall that in 2D, every item $t$ is transformed to the line:
{\revision
\begin{align}\label{eq:dual2d}
\mathsf{d}(t): t[1]\times x_1 + t[2]\times x_2 = 1
\end{align}
}
Also, recall that every function $f$ with the weight vector $w$ is represented with the origin-starting ray passing through the point $w$, and consider points $t_i$ and $t_j$.
$f$ ranks $t_i$ higher than $t_j$ if the intersection of $\mathsf{d}(t_i)$ with $f$ is closer to the origin than the intersection of
$\mathsf{d}(t_j)$ with $f$.
%Between $t_i$ and $t_j$, for any function $f$, the one that its dual line intersects the ray of $f$ closer to the origin is ranked higher.

Consider $f$ whose origin-starting ray passes through the intersection of $\mathsf{d}(t_i)$ and $\mathsf{d}(t_j)$.
Since both lines intersect with the ray of $f$ at the same point, $f$ assigns an equal score to $t_i$ and $t_j$. We refer to this function (and its ray) as the {\em ordering exchange} {\revision (first defined in~\cite{fairranking})} between $t_i$ and $t_j$, and denote it $\times_{t_i,t_j}$. The ordering between $t_i$ and $t_j$ changes on two sides of $\times_{t_i,t_j}$: %\julia{I'd say ``intersection'' rather than ``ray'', since it's not a ray in the general case, right?  generally this is a D-1 structure.}\abol{it is the origin-starting ray passing through that intersection}
$t_i$ is ranked higher than $t_j$ one side of the ray, and $t_j$ is ranked higher than $t_i$ on the other side.
For example, consider $t_1$ and $t_4$ in Example~\ref{example1}, shown in
Figure~\ref{fig:toy3} in the dual space: {\revision the closest line to the  origin on the $x_1$ axis represents $\mathsf{d}(t_2)$, and the next closest line is $\mathsf{d}(t_4)$.}
The left-most intersection in the figure is between $\mathsf{d}(t_1)$ and $\mathsf{d}(t_4)$. The top-left dashed line that starts from the origin and passes through this intersection shows $\times_{t_1,t_4}$: $t_1$ is preferred over $t_4$ on the left of $\times_{t_1,t_4}$, and $t_4$ is preferred over $t_1$ on the right.

{\revision 
An item $t$ dominates~\cite{asudeh2016discovering,pareto, skylineoperator} an item $t'$, if $\nexists x_i$ s.t. $t'[i]>t[i]$ and $\exists x_i$ s.t. $t[i]> t'[j]$. If $t$ dominates $t'$, then these items do not exchange order.}
Consider two items $t$ and $t'$ that do not dominate each other.
% \footnote{\small An item $t$ dominates an item $t'$, if $\nexists x_i$ s.t. $t'[i]>t[i]$ and $\exists x_i$ s.t. $t[i]> t'[j]$~\cite{asudeh2016discovering,pareto}. If $t$ dominates $t'$, then these items do not exchange order.}.
Equation~\ref{eq:dual2d} can be used for finding the intersection between the lines $\mathsf{d}(t)$ and $\mathsf{d}(t')$. Considering the polar coordinates of the intersection, $\times_{t,t'}$ is specified by the angle $\theta_{t,t'}$ (between the ordering exchange and the x-axis) as follows:
\begin{align}\label{eq:2dexhange2}
\theta_{t,t'} = \arctan \frac{t'[1]-t[1]}{t[2]-t'[2]}
\end{align}

The ordering exchanges between pairs of items of a database partition the space of scoring functions into a set of regions. Each region is identified by the two ordering exchanges that form its borders. Since there are no ordering exchanges within a region, all scoring functions inside a region induce the same ranking of the items. Thus, the number of regions is equal to $|\mathfrak{R}|$, as $\mathfrak{R}$ is the collection of rankings defined by these regions. For instance, Figure~\ref{fig:toy3} shows regions $R_1$ through $R_{11}$ that define the set of possible rankings of Example~\ref{example1} for $\mathcal{U}$.

\submit{\newpage }
\subsection{Stability Verification}\label{subsec:2dsv}
%\julia{I still don't understand this notation: ``let $t$ be $t_{\mathfrak{r}[i]}$ and $t'$ be $t_{\mathfrak{r}[i]}$.''  If we need a way to refer to an item at position $i$ in a ranking, the standard way to do this is to write $\tau(i)$, we can add this to preliminaries.  If we also need to look up the rank of some item $t$ in $\tau$, we can write that as $\tau^{-1}(t)$.} \abol{I changed the write up and removed the confusion.}

The ordering exchanges are the key to figuring out the stability of a ranking.
Consider a ranking $\mathfrak{r}$. For a value of $1\leq i<n$, let $t$ and $t'$ be the $i$-th and $(i+1)$-th items in $\mathfrak{r}$. %in which every item $t_{\mathfrak{r}[i]}$ is ranked higher than $t_{\mathfrak{r}[i+1]}$. For any value of $i\in[1,n]$, let $t$ be $t_{\mathfrak{r}[i]}$ and $t'$ be $t_{\mathfrak{r}[i]}$.
Following Equation~\ref{eq:2dexhange2}, $\theta_{t,t'}$ specifies the ordering exchange between $t$ and $t'$.
If $t[1]<t'[1]$ (resp. $t[1]>t'[1]$), all functions with  angles $\theta<\theta_{t,t'}$ (resp. $\theta>\theta_{t,t'}$) rank $t$ higher than $t'$.
{\revision
The reason is that if $t[1]>t'[1]$, $t[2]$ should be smaller than $t'[2]$, otherwise $t$ dominates $t'$. Hence $\frac{t[1]}{t[2]}>\frac{t'[1]}{t'[2]}$, i.e. the dual line $\mathsf{d}(t)$ has a larger slope than $\mathsf{d}(t')$, and intersects the rays in range $[0,\theta_{t,t'})$ closer to the origin.
}

\techrep{
\begin{algorithm}[h]
\caption{\svtd \\
         {\bf Input:} Two dimensional dataset $\mathcal{D}$ with $n$ items and the ranking $\mathfrak{r}$ \\
         {\bf Output:} The stability and the region of $\mathfrak{r}$
        }
\begin{algorithmic}[1]
\label{alg:2dsv}
	\STATE $(\theta_1,\theta_2) = (0,\pi/2)$
	\FOR{$i=1$ to $n-1$}
    	\STATE $t = \mathfrak{r}[i]$; $t' = \mathfrak{r}[i+1]$
    	\STATE {\bf if} $t$ dominates $t'$ {\bf then continue}
        \STATE {\bf if} $t'$ dominates $t$ {\bf then return} null %{\scriptsize \tt // $\mathfrak{r}$ is infeasible}
    	\STATE $\theta =\arctan\frac{t'[1]-t[1]}{t[2]-t'[2]}$
        %\STATE {\bf if} $\big(t[1]<t'[1]$ and $\theta<\theta_1 \big)$ or $\big( t[1]>t'[1]$ and $\theta>\theta_2\big)$ {\bf then return} null {\scriptsize \tt // $\mathfrak{r}$ is infeasible}
        \STATE {\bf if} $\big(t[1]<t'[1]$ and $\theta>\theta_1 \big)$ {\bf then} $\theta_1 = \theta$
        \STATE{\bf if} $\big(t[1]>t'[1]$ and $\theta<\theta_2 \big)$ {\bf then} $\theta_2 = \theta$
        \STATE{\bf if} $\theta_1>\theta_2$ {\bf then return} null %{\scriptsize \tt // $\mathfrak{r}$ is infeasible}
    \ENDFOR
    \STATE {\bf return} $\frac{\theta_2-\theta_1}{\pi/2}$, $(\theta_1,\theta_2)$
\end{algorithmic}
\end{algorithm}
Algorithm~\ref{alg:2dsv} uses
}
\submit{We use}
this idea for computing the stability (and the region) of a given ranking $\mathfrak{r}$.
The stability verification algorithm uses the angle range $(\theta_1,\theta_2)$ for specifying the region of $\mathfrak{r}$.
For each value of $i$ in range $[1,n)$, the algorithm considers the items $t$ and $t'$ to be the $i$-th and $(i+1)$-th items in $\mathfrak{r}$, respectively.
If $t'$ dominates $t$, the ranking is not valid.
Otherwise, if $t$ does not dominate $t'$,
the algorithm computes the ordering exchange $\times_{t,t'}$ and, based on the values of $t[1]$ and $t'[1]$, decides to use it for setting the upper bound or the lower bound of the ranking region.
After traversing the ranked list $\mathfrak{r}$, the algorithm returns $\frac{\theta_2-\theta_1}{\pi/2}$ as the stability value and $(\theta_1,\theta_2)$ as the region of $\mathfrak{r}$.
Since the algorithm scans the ranked list only once, stability verification in 2D is in $O(n)$.
\submit{The algorithm's pseudocode is provided in the technical report~\cite{techrep}.}

\subsection{Stability Enumeration}
In 2D, $\ar$ is identified by two angles demarcating the edges of the pie-slice.
For example, let $\ar_1$ be defined by the set of constraints $\{w_1 \leq w_2, 2w_1\geq w_2\}$. This defines the set of functions above the line $w_1 = w_2$ and below the line $2w_1 = w_2$, limiting the region of interest to the angles in the range $[\pi/4^\circ, \pi/3^\circ]$.
\submit{
Similarly, region $\ar_{2}$ defined around $f=x_1+x_2$ with the maximum angle $\pi/10^\circ$ corresponds to the angles in the range $[3\pi/20^\circ, 7\pi/20^\circ]$.}
\techrep{Similarly, let $\ar_{2}$ be defined around the function $f=x_1+x_2$ with the maximum angle $\pi/10^\circ$ (i.e., minimum cosine similarity of 95.1\% with $f$). The $\pi/10^\circ$ angle around $f$ identifies the range $[3\pi/20^\circ, 7\pi/20^\circ]$ as the region of interest.}
In what follows, we use $[~\ar[1],~ \ar[2]~]$ to denote the borders of $\ar$.
Based on Definition~\ref{def:1}, the stability of a ranking $\mathfrak{r}\in\mathfrak{R}$ in 2D is the span of its region -- the distance between its two borders.
\techrep{
For example, in Figure~\ref{fig:toy3}, the regions $R_{11}$ and $R_1$ are wide and provide stable rankings, while the rankings provided by $R_5$ and $R_8$ are hazy and may change by small changes in the weight vector.
}

We propose the algorithm \raysweeping (Algorithm~\ref{alg:raysweeping}) that starts from the angle $\ar[1]$ and, while sweeping a ray toward $\ar[2]$, uses the dual representation of the items for computing the ordering exchanges and finding the ranking regions. The algorithm stores the regions, along with the stability of their rankings, in a heap data structure that is later used by the \getnexttd primitive.

%As explained in \S~\ref{sec:pre}, the ordering exchanges in 2D are identified by one angle.

\begin{algorithm}[!h]
\caption{\raysweeping \\
         {\bf Input:} Two dimensional dataset $\mathcal{D}$ with $n$ items and the region of interest in the form of an angle range $[~\ar[1],~ \ar[2]~]$ \\
         {\bf Output:} A heap of ranking regions and their stability
        }
\begin{algorithmic}[1]
\label{alg:raysweeping}
	\STATE sweeper = {\it new} min-heap$([\ar[2]])$;
    \STATE {\revision $\vec{w} = (\cos \ar[1], \sin \ar[1])$}
    \STATE $L =$ $\mathcal{r}_{f}(\mathcal{D})$
    \FOR{$i=1$ to $n-1$}
    	\STATE $\theta =\arctan (L_{i+1}[1]-L_i[1])/(L_i[2]-L_{i+1}[2])$
        \STATE {\bf if} $\ar[1]< \theta < \ar[2]$ {\bf then} sweeper.push$((\theta,L_i,L_{i+1}))$
    \ENDFOR
    \STATE {\revision h} = {\it new} max-heap$()$; $\theta_p = \ar[1]$
    \WHILE{sweeper is not empty}
        \STATE ($\theta$, $t$,$t'$) = sweeper.pop()
        \STATE $i,j = $ index of $t,t'$ in $L$
        \STATE {\revision h}.push $\left( \frac{\theta - \theta_p}{\ar[2] - \ar[1]}, (\theta_p, \theta) \right)$
        \STATE swap $L_i$ and $L_{j}$ \submit{and add the ordering exchanges between the new adjacent items to the sweeper}
        \techrep{
        \IF{$i>1$}
        	\STATE $\theta' = \arctan (t'[1]-L_{i-1}[1])/(L_{i-1}[2]-t'[2])$
        	\STATE {\bf if} $\theta< \theta' < \ar[2]$ {\bf then} sweeper.push( ($\theta',L_{i-1}$))
        \ENDIF
        \IF{$j<n$}
        	\STATE $\theta' = \arctan (t[1]-L_{j+1}[1])/(L_{j+1}[2]-t[2])$
        	\STATE {\bf if} $\theta< \theta' < \ar[2]$ {\bf then} sweeper.push( ($\theta,L_{j+1}$))
        \ENDIF
        }
        \STATE $\theta_p = \theta$
    \ENDWHILE
    \STATE {\bf return} {\revision h}
\end{algorithmic}
\end{algorithm}

%Algorithm~\ref{alg:raysweeping} presents the pseudo-code of \raysweeping.
Algorithm~\ref{alg:raysweeping} starts by ordering the items based on $\ar[1]$.
It uses the fact that at any moment, an adjacent pair in the ordered list of items exchange ordering, and, therefore, computes the ordering exchanges between the adjacent items in the ordered list. The intersections that fall into the region of interest are added to the sweeper's min-heap.
Until there are intersections over which to sweep, the algorithm pops the intersection with the smallest angle, marks the region between it and the previous intersection in the output max-heap, and updates the ordered list accordingly. Upon updating the ordered list, the algorithm adds the intersections between the new adjacent items to the sweeper. Since the total number of intersections between the items is bounded by $O(n^2)$, and the heap operation is in $O(\log n)$, \raysweeping is in $O(n^2\log n)$.

After finding the ranking regions and theirs spans, every call to \getnexttd \techrep{(Algorithm~\ref{alg:getnext2d})} pops the most stable region from the heap and chooses a scoring function $f$ inside the region. The algorithm returns  the ranking $\mathcal{r}_{f}(\mathcal{D})$, along with the width of the region (its stability), to the user.
\submit{
{\revision Due to the space limitations, the pseudo code of \getnexttd is provided in the technical report~\cite{techrep}}.
}

\techrep{
\vspace{3mm}
\begin{algorithm}[!h]
\caption{\getnexttd \\
         {\bf Input:} The heap of regions, with the top-$h$ regions being removed from it \\
         {\bf Output:} the top-$h+1$ ranking, along with its stability
        }
\begin{algorithmic}[1]
\label{alg:getnext2d}
	\STATE $S,(\theta_1,\theta_2)$ = heap.pop()
    \STATE $w = (\cos \frac{\theta_1 + \theta_2}{2}, \sin \frac{\theta_1 + \theta_2}{2})$
    \STATE $L =$ $\mathcal{r}_{f}(\mathcal{D})$
    \STATE {\bf return} $L,S,(\theta_1,\theta_2)$
\end{algorithmic}
\end{algorithm}
}

Since there are no more than $O(n^2)$ regions in the heap, \getnexttd needs $O(\log n)$ to find the $(h+1)$-th stable region. Then, it takes $O(n\log n)$ to compute the {\new ranking} for the region. As a result, the first call to \getnexttd that creates the heap of regions takes $O(n^2\log n)$, while subsequent calls take $O(n\log n)$.
Note that subsequent \getnexttd calls can be done in the order of $O(\log n)$, with memory cost of $O(n^3)$, by storing the ordered list $L$ for every region in \raysweeping algorithm.

\section{Multi dimensional (MD) Ranking}\label{sec:md}
Building upon the intuitions developed from the 2D case, we now turn to the general setting where $d>2$. 
Again, we consider the items in dual space and use the ordering exchanges for specifying the borders of the ranking regions.
Recall from Equation~\ref{eq:dual} that an item $t$ is presented as the hyperplane $\mathsf{d}(t):~ \sum_{i=1}^d t[i].x_i=1$.
For a pair of items $t_i$ and $t_j$ the ordering exchange $h = \times_{t_i,t_j}$ is a hyperplane that contains the functions that assign the same score to both items. Therefore:
\submit{\vspace{-5mm}}
\begin{align}\label{eq:mdexchange}
\times_{t_i,t_j} = \sum\limits_{k=1}^{d} (t_i[k] - t_j[k])x_k = 0
\end{align}
Every hyperplane $h = \times_{t_i,t_j}$ {\em partitions} the function space $\mathcal{U}$ in two ``half-spaces''~\cite{edelsbrunner}:
\begin{itemize}[leftmargin=*,itemsep=0pt]
\item $h^-: \sum_{k=1}^{d} (t_i[k] - t_j[k])x_k < 0$: for the functions in $h^-$, $t_j$ outranks $t_i$.
\item $h^+: \sum_{k=1}^{d} (t_i[k] - t_j[k])x_k > 0$: for the functions in $h^+$, $t_i$ outranks $t_j$.
\end{itemize}

Similar to \S~\ref{sec:2d}, we first show how ordering exchanges can be used for verifying the stability of a ranking and then focus on designing the \getnext operator.

\subsection{Stability Verification}
Identifying the half-spaces defined by the ordering exchanges between adjacent items in a ranking $\mathfrak{r}$ is the key to figuring out its stability.
For each value of $i$ in range $[1,n)$, let $t$ and $t'$ be the $i$-th and $(i+1)$-th items in $\mathfrak{r}$.
Using Equation~\ref{eq:mdexchange}, every function in the positive half-space $h^+: \sum_{k=1}^{d} (t[k] - t'[k])x_k > 0$ ranks $t$ higher than $t'$.
The intersection of these half-spaces specifies an open-ended $d$-dimensional cone (referred to as $d$-cone) whose base is a $(d-1)$ dimensional convex polyhedron.
Every function in this cone generates the ranking $\mathfrak{r}$, while no function outside it generates $\mathfrak{r}$.
In other words, this $d$-cone is the ranking region of $\mathfrak{r}$.
\techrep{Algorithm~\ref{alg:mdsv}}\submit{The algorithm for verifying stability} %uses this, and 
finds the region of a ranking $\mathfrak{r}$ as the set of positive half-spaces defined by the ordering exchanges of the adjacent items in $\mathfrak{r}$. 

\techrep{
\begin{algorithm}[h]
\caption{\sv \\
         {\bf Input:} Dataset $\mathcal{D}$ with $n$ items and $d$ attributes, and the ranking $\mathfrak{r}$ \\
         {\bf Output:} The stability and the region of $\mathfrak{r}$
        }
\begin{algorithmic}[1]
\label{alg:mdsv}
    \STATE $R = \{\}$
    \FOR{$i=1$ to $n-1$}
        \STATE $t = \mathfrak{r}[i]$; $t'=\mathfrak{r}[i+1]$
        \STATE {\bf if} $t$ dominates $t'$ {\bf then continue}
        \STATE {\bf if} $t'$ dominates $t$ {\bf then return} null {\scriptsize \tt // infeasible}
        \STATE $R.add(\sum_{k=1}^{d} (t[k] - t'[k])x_k > 0)$
    \ENDFOR
    \STATE {\bf return} $S(R, \mathcal{U})$ {\tt \scriptsize /*c.f. \S~\ref{subsec:SOracle}*/}, $R$
\end{algorithmic}
\end{algorithm}
}

Based on Definition~\ref{def:1}, the volume ratio of the region of $\mathfrak{r}$ to the one of $\mathcal{U}$ (or, more generally, to $\ar$) is its stability.  However, since $\mathfrak{r}$ is a polyhedron, computing its volume is \#P-hard~\cite{dyer1988complexity}.
%This, however, is \#P-complete~\cite{dyer1988complexity}. 
Therefore, we use numeric methods {\revision and sampling} for estimating this quantity.
Throughout this section, we assume the existence of a stability oracle $S(R,\ar)$ that, given a convex region $R$ in the form of an intersection of half-spaces and a region of interest $\ar$, returns the stability of $R$ in $\ar$.
\techrep{We will describe in \S~\ref{subsec:SOracle} how to construct such an oracle.}
\submit{Due to the space limitations, We will describe in the technical report~\cite{techrep} how to construct such an oracle.}
Stability verification in MD is in $O(n+X_S)$ where $X_S$ is the complexity of constructing the stability oracle. 

\subsection{Stability Enumeration}\label{subsec:thbased}
Similarly to verifying the stability of a ranking, ordering exchanges can be used for finding possible rankings in a region of interest $\ar$.
The set of ordering exchanges intersecting $\ar$ define a dissection of $\ar$ into connected convex $d$-cones, each showing a ranking region as the intersection of ordering exchange half-spaces and $\ar$. This dissection is named the {\em arrangement} of ordering exchange hyperplanes~\cite{edelsbrunner}.
{\revision
For example, % where $d=3$
the ordering exchanges in $\mathbb{R}^3$ are the planes passing through the origin.
Each plane dissects the space in two half-spaces.
The intersection of the half-spaces forms an arrangement in the form of open-ended convex cones.
}
\techrep{
Algorithm~\ref{alg:hyperplanes} shows the pseudocode of the function {\em $\times$hps} that finds the set of ordering exchange hyperplanes passing through the region of interest $\ar$.

\begin{algorithm}[!h]
\caption{{\bf $\times$hps} \\
         {\bf Input:} database $\mathcal{D}$, $n$, $d$, region of interest $\ar$\\
         {\bf Output:} the set of ordering exchange hyperplanes intersecting $\ar$
        }
\begin{algorithmic}[1]
\label{alg:hyperplanes} 
    \STATE $H=\{\}$
    \FOR{$t_i$, $t_j$ in $\mathcal{D}$}
        \IF{$t_i$ and $t_j$ do not dominate each other} 
            \STATE Let $\times_{t_i,t_j}$ be the hyperplane $\sum\limits_{k=1}^{d} (t_i[k] - t_j[k])x_k = 0$
            \STATE {\bf if} $\times_{t_i,t_j}$ intersect with $\ar$ {\bf then} $H.add(\times_{t_i,t_j})$
        \ENDIF
    \ENDFOR
    \STATE {\bf return} $H$
\end{algorithmic}
\end{algorithm}
}

\begin{theorem}\label{th:uniqueRanking}
Every ranking $\mathfrak{r}\in\mathfrak{R}^*$ is provided by the functions in exactly one convex region in the arrangement of ordering exchange hyperplanes in $\ar$.
\end{theorem}

\submit{
\noindent{\it Proof sketch:} The proof follows the non-existence of ordering exchanges inside a region and the existence of at least one ordering exchange between two regions. Additional details are provided in the technical report~\cite{techrep}.

}
\techrep{
\begin{proof}
First, every ranking $\mathfrak{r}\in\mathfrak{R}$ is generated by at least one function in $\ar$. Since the arrangement partitions the $\ar$ into disjoint regions, this function belongs to one of the regions ($d$-cones) of the arrangement. Also, because there is no ordering exchange in a region, every function in it will generate the same ranking, i.e., for every ranking there exists at least one region that generates it.
Now, to prove the one to one mapping by contradiction, assume the existence of two regions that generate the same ranking. Since these two regions are disjoint, there should exist at least an ordering exchange hyperplane $h=\times_{t_i,t_j}$ between them. Hence, one of them belongs to the half-space $h^+$ while the other falls into $h^-$. Therefore, in one of the regions, $t_i$ is preferred over $t_j$ while in the other $t_j$ is ranked higher.
This contradicts with the assumption that both regions generate the same ranking.
\end{proof}
}

Theorem~\ref{th:uniqueRanking} shows the one-to-one mapping between the rankings $\mathfrak{r}\in\mathfrak{R}^*$ and the regions of the arrangement.
Following Theorem~\ref{th:uniqueRanking}, the baseline for finding the stable regions in $\ar$ is to first construct the arrangement and then, similarly to \S~\ref{sec:2d}, create a heap of regions and their stabilities. % that will be used in the \getnext$_b$ operator. \julia{I think just  \getnext, no "b".}
Then, each \getnext operation is as simple as removing the most stable ranking from the heap. The construction of arrangements is extensively studied in the literature~\cite{orlik2013arrangements,grunbaum2003arrangements, schechtman1991arrangements,fairranking, agarwal2000arrangements, edelsbrunner}.
The problem with this baseline is that it first needs to construct the complete arrangement of ordering exchange hyperplanes, and to compute the stability of each.
The number of ordering exchanges intersecting $\ar$ is bounded by $O(n^2)$. 
Therefore the arrangement can contain as many as $O(n^{2d})$ regions~\cite{edelsbrunner}. Yet, the baseline needs to compute the stability of each ranking associated with every region.
Considering that our objective is to find stable rankings, rather than to discover all possible rankings, and that the user will likely be satisfied after observing {\em a few} rankings, this construction if wasteful. 
%This gets more highlighted by paying attention to the fact that our objective is to find stable rankings (not discovering all possible rankings), and the user will probably get satisfied after observing {\em a few} rankings.  Therefore, the construction of all regions is a waste of computation.
{\revision
Instead, since the objective is to find the next stable ranking (not to discover all rankings), we propose an algorithm that targets the construction of only the next stable ranking and delays the construction of other rankings.
}

Arrangement construction is an iterative process that starts by partitioning the space into two half-spaces by adding the first hyperplane $H[1]$ (it partitions the space into $H[1]^-$ and $H[1]^+$).
The algorithm then iteratively adds the other hyperplanes; to add $H[i]$, it first identifies the set of regions in the arrangement of $H[1]$ to $H[i-1]$ that $H[i]$ intersects with, and then splits each such region into two regions (one including $H[i]^-$ and one $H[i]^+$).

\techrep{
\begin{figure}[!tb]
    \centering
    \begin{small}
    \begin{tabular}{|l|}
    \hline
    \\
    \texttt{struct Region $\{$ }\\
    \texttt{\hspace{5mm} C; // region half-spaces}\\
    \texttt{\hspace{5mm} S; // the stability of the region}\\
    \texttt{\hspace{5mm} pending; // the first pending hyperplane}\\
    \texttt{\hspace{5mm} sb; // samples beginning index (c.f. \S~\ref{subsec:par})}\\
    \texttt{\hspace{5mm} se; // samples end index (c.f. \S~\ref{subsec:par})}\\
    {\tt $\}$ } \\
    \\
    \hline
    \end{tabular}
    \end{small}
    \caption{The region data structure}
    \label{fig:region}
\end{figure}
}

The \getnexttb algorithm, however, only breaks down the largest region at every iteration, delaying the construction of the arrangement in all other regions.
The algorithm uses the ``region'' data structure \techrep{(Figure~\ref{fig:region})} to record each region in the arrangement of ordering exchanges.
This data structure contains the following fields: (a) {\tt C}: the set of half-spaces defining the region; (b) {\tt S}: the stability of the region, and (c) {\tt pending}: the index of the next hyperplane to be added to the region. % \julia{what's ``it'', the arrangement?}
In addition, every region contains two extra fields {\tt sb} and {\tt se} that are used for determining the intersection of next hyperplanes with it\techrep{ (as well as computing its stability)}. 
\submit{Due to the space limitations, we provide further details about these, as well as the pseudo code of the algorithm \getnexttb in the technical report~\cite{techrep}.}
\techrep{
We defer further details about these two fields to \S~\ref{subsec:par}.

\begin{algorithm}[!h]
\caption{\getnexttb \\
         {\bf Input:} database $\mathcal{D}$, $n$, $d$, \techrep{region of interest }$\ar$, top-$h$ stable rankings $\mathfrak{R}_h$\\
         {\bf Output:} $(h+1)$-th stable ranking and its stability measures
        }
\begin{algorithmic}[1]
\label{alg:tba}
    \IF{$\mathfrak{R}_h$ is $\emptyset$}
        \STATE heap = {\it new } max-heap()
        \submit{
        \STATE $H=$the ordering exchanges that intersect $\ar$
        }
        \techrep{\STATE $H=${\bf $\times$hps}($\mathcal{D}$, $n$, $d$, $\ar$)}
           \STATE heap.push(1,{\it new} Region($C=\{\}$, pending=$1$, $S=1$))
    \ENDIF
    \WHILE{heap is not empty}
        \STATE $s,r$ = heap.pop()
        \WHILE{true}
            \submit{
            \STATE {\bf if $r.$pending$>|H|$} {\bf then} set $w$ as a point in $r$ and {\bf return} $\mathcal{r}_{f}(\mathcal{D})$, $r.S$
            }
            \techrep{
            \IF{$r.$pending$>|H|$}
                \STATE $w$ = a point in $r$
                \STATE {\bf return} $\mathcal{r}_{f}(\mathcal{D})$, $r.S$
            \ENDIF
            }
            \STATE $k = $ {\bf passThrough}($H[r.$pending$]$, $r$) {\scriptsize \tt // c.f. \S~\ref{subsec:par}}
            \STATE {\bf if} $k>0$ {\tt\scriptsize /*intersects $r$*/} {\bf then break}
            \STATE $r.$pending $+=1$
        \ENDWHILE
        \STATE $C_1 = C_2 = r.C$; $p = r.$pending$+1$
        \STATE add $H[r.$pending$]^-$ to $C_1$ and $H[r.$pending$]^+$ to $C_2$
        \STATE $r_1 = $ {\it new} Region($C = C_1$, pending=$p$, $S=${\bf S}($C_1$) \big)
        \STATE $r_2 = $ {\it new} Region($C = C_2$, pending=$p$, $S=${\bf S}($C_2$) \big)
        \STATE heap.push($r_1.S$,$r_1$); heap.push($r_2.S$,$r_2$)
    \ENDWHILE
\end{algorithmic}
\end{algorithm}
} % end of techrep

\techrep{Algorithm~\ref{alg:tba} shows the pseudo-code of the multi-dimensional \getnext operator.
}
%At the high level, w
While constructing the arrangement of hyperplanes, the algorithm keeps track of the stability of the regions, as it adds hyperplanes to the largest one.
It uses a max-heap for this purpose.
For the first \getnext operation, the algorithm \techrep{calls Algorithm~\ref{alg:hyperplanes} to }find\submit{s} the set of ordering exchanges $H$ that intersect with $\ar$.
It also creates the root region that contains all functions in $\ar$ and adds it to the heap.
While the heap is not empty, the algorithm pops the most stable region $r$ from it. It then iterates over the pending hyperplanes that can be added to $r$, attempting to find one that intersects with that region. Testing whether a hyperplane intersects with a region is done by solving a linear program.  Specifically, we solve a quadratic program %for the region of interest defined by the ray and angle model)
that looks for a function in $\ar$ that satisfies both the inequality constraints defined by the half-spaces of the region, and the equality constraints defined by the hyperplane. (Alternatively, sampling can be used for this purpose\techrep{, as we will show in \S~\ref{subsec:par}.}\submit{. We provide further details about this in the technical report~\cite{techrep}.})

If no more hyperplanes can be added to region $r$, the algorithm returns $r$ as the next stable region.
Otherwise, if a hyperplane is found that intersects with $r$, then the algorithm uses it to break $r$ into two regions, and adds them to the heap.

Still, in the worst case (where all regions are equally stable) the algorithm may need to construct the arrangement before returning even the first stable region. Therefore, the worst case cost of the algorithm is still $O(n^{2d})$.

Throughout this section, we assumed the existence of an oracle that, given a region in the form of a set of half-spaces, returns its stability. In next section, we discuss unbiased sampling from the function space that plays a key role in the design of the oracle.
Such sampling will also enable the design of the randomized algorithm in \S~\ref{sec:randomized} that does not depend on the arrangement construction (and therefore, does not suffer its high complexity).

%\section{\getnextr : Randomized Get-Next}\label{sec:randomized}
%\submit{\vspace{-2mm}}
\subsection{Randomized Get-Next}\label{sec:randomized}
\begin{comment}
As explained in \S~\ref{sec:md}, the threshold-based algorithm aims to efficiently finding the most stable regions by focusing the construction to the candidate region.
However, it still needs to fully construct the regions before returning them which is computationally expensive and makes it intractable for large settings.

On the other hand, in large settings, users are usually interested in the head (the top-$k$) of the ranked list, rather than the complete ordering.
For example, a complete ordering between ten thousands items may be overwhelming and less interesting than the top few items.
Thus, in this section, we propose a randomized algorithm for settings with large number of items.
\end{comment}

In a setting with many items, users are usually not interested in the complete ranking between all of the items.
The top-$k$ items model~\cite{ihab,fagin2003} is a natural fit for such settings, and therefore, is used as the de-facto data retrieval model in the web~\cite{asudeh2016query,qr2}.
In this model, the user is interested in the head (the top-$k$) of the ranked list, rather than the complete ordering.
In the following, we propose a randomized algorithm that, in addition to being scalable for large settings, is applicable for enumerating the top-$k$ items.

%Every ranking in the region of interest is associated with a stability measure. The idea is that if we can sample the rankings based on their stabilities, we discover them and can devise a Monte-Carlo method~\cite{montecarlo} for estimating their stability.
{\revision
While every ranking is generated by continuous ranges of functions, every function $f$ generates only one ranking of items. Moreover, the larger the volume of a ranking region (i.e. the more stable it is), the higher the chance of choosing a random function from it.
Therefore, {\it uniform sampling of the function space allows sampling of rankings based on their stability distribution}.
%This suggests using such samples from the function space to design a Monte-Carlo method.
We delay the details of a sampler that generates functions uniformly at random from a region of interest $\ar$ to \S~\ref{sec:sampler}.
Assuming the existence of such sampler, in this section, we use the Monte-Carlo methods for Bernoulli random variables~\cite{blaker2000confidence, hickernell2013guaranteed, hoeffding1963probability, chernoff1952measure,montecarlo} and design the randomized \getnextr operator.}
% As discussed earlier in \S~\ref{sec:sampler}, %each function in the region of interest $\ar$ provides a ranking of items and 
% the uniform sampling of the functions of a region of interest $\ar$ provides a sampling of the rankings based on the distribution of their volumes.
% Therefore, in this section, we 
%use the unbiased sampler provided in \S~\ref{sec:sampler} and use the Monte-Carlo methods for Bernoulli random variables~\cite{blaker2000confidence, hickernell2013guaranteed, hoeffding1963probability, chernoff1952measure,montecarlo} and design the randomized \getnextr operator.
This operator works both for finding the stable rankings in a region of interest, as well as the top-$k$ results. In the following, we use ranking for the explanations but all the algorithms and explanations are also valid for top-$k$.

Monte-Carlo methods work based on repeated sampling and the central limit theorem in order to solve deterministic problems. We consider using these algorithms for numeric integration.
Based on the law of large numbers~\cite{durrett2010probability},
the mean of independent random variables can be used for approximating the integrals.
That is possible, as the expected number of occurrences of each observation is proportional to its probability.
At a high-level, the Monte-Carlo methods work as follows:
first they generate a sufficiently large set of inputs based on a probability distribution over a domain; then they use these inputs to do  estimation and aggregate the results.

We use sampling both for discovering the rankings as well as for estimating their stability.
We design the \getnextr operator to allow the user to either (i) specify the sampling budget, or (ii) the confidence interval.
Each of these two approaches has their pros and cons.
The running time in (i) is fixed but the error is variable. In (ii), on the other hand, the operator guarantees the output quality while the running time is not deterministic.
In the following, we \techrep{first} explain the details for (i).
\techrep{Then}\submit{Due to space limitations}, we show how this can be adopted for (ii)\submit{ in the technical report~\cite{techrep}}.

\techrep{\subsection{Fixed budget}}
The sampler explained in \S~\ref{subsec:sampleui} draws functions uniformly at random from the function space. Each function is associated with a ranking. The uniform samples on the function space provide ranking samples based on their stabilities (the portion of functions in $\ar$ generating them).
For each ranking $\mathfrak{r}\in\mathfrak{R}$, consider the distribution of drawing a function that generate it. The probability mass function of this distribution is:
\begin{align}\label{eq:functiondist}
p(\Theta;S(\mathfrak{r}))=\begin{cases}
S(\mathfrak{r}) &\mathcal{r}_{f(\Theta)}(\mathcal{D}) = \mathfrak{r} \\
1-S(\mathfrak{r}) &\mathcal{r}_{f(\Theta)}(\mathcal{D}) \neq \mathfrak{r}
\end{cases}
\end{align}

Let the random Bernoulli variable $x_\mathfrak{r}$ be 1 if $\mathcal{r}_{f(\Theta)}(\mathcal{D}) = \mathfrak{r}$ and 0 otherwise.
Recall that the mean and standard deviation of a Bernoulli distribution with the success probability of $S(\mathfrak{r})$ are $\mu_\mathfrak{r} = S(\mathfrak{r})$ and $\sigma_\mathfrak{r} = S(\mathfrak{r})(1-S(\mathfrak{r}))$.
Let $m_\mathfrak{r}$ be the average of a set of $N$ samples of the random variable $x_\mathfrak{r}$.
Then, $E[m_\mathfrak{r}] = S(\mathfrak{r})$ and the standard deviation of samples are $s_\mathfrak{r} = m_\mathfrak{r}(1-m_\mathfrak{r})$.
Based on the central limit theorem, we also know that the distribution of the sample average is N$\big(\mu_\mathfrak{r}, \frac{\sigma_\mathfrak{r}}{\sqrt{N}}\big)$ -- the Normal distribution with the mean $\mu_\mathfrak{r}$ and standard deviation $\frac{\sigma_\mathfrak{r}}{\sqrt{N}}$.
For a large value of $N$, we can estimate $\sigma_\mathfrak{r}$ by $s_\mathfrak{r}$.

For a confidence level $\alpha$, the confidence error $e$ identifies the range $[m_\mathfrak{r}-e, m_\mathfrak{r}+e]$ such that:
\begin{align}
p(m_\mathfrak{r}-e\leq \mu_\mathfrak{r}\leq m_\mathfrak{r}+e)= 1-\alpha
\end{align}
Using the Z-table:
\techrep{
\begin{align} \label{eq:costerror}
\nonumber
e &= Z(1-\frac{\alpha}{2})\frac{s_\mathfrak{r}}{\sqrt{N}}\\
  &= Z(1-\frac{\alpha}{2})\sqrt{\frac{ m_\mathfrak{r}(1-m_\mathfrak{r})}{N}}
\end{align}
}
\submit{
\vspace{-3.5mm}
\begin{align} \label{eq:costerror}
e &= Z(1-\frac{\alpha}{2})\frac{s_\mathfrak{r}}{\sqrt{N}}
  = Z(1-\frac{\alpha}{2})\sqrt{\frac{ m_\mathfrak{r}(1-m_\mathfrak{r})}{N}}
\end{align}
}

Using this argument, we use a set of $N$ samples of the function space for the design of \getnextr with a fixed budget.
Every time  \getnextr is called, we collect a set of $N$ samples and use them for finding the next stable ranking and estimating its stability. In order to provide a more accurate estimation, in addition to the $N$ new samples, it uses the aggregates of its previous runs. Algorithm~\ref{alg:gnr1} shows the pseudocode of \getnextr with a fixed budget.

%\submit{\vspace{-2mm}}
\begin{algorithm}[!h]
\caption{\getnextr\\
{\bf Input:} \techrep{database} $\mathcal{D}$, \techrep{region of interest} $\ar$, previous stable rankings $\mathfrak{R}_{h-1}$, hash of previous aggregates $cnts$, total number of previous samples $N'$, confidence level $\alpha$, and sampling budget $N$\\
{\bf Output:} Next stable ranking and its stability measures
}
\begin{algorithmic}[1]
\label{alg:gnr1}
	\FOR {$k=1$ to $N$} 
    	\STATE $w =$ Sample$\ar(\ar)$
        \STATE $\mathfrak{r} = \mathcal{r}_{f}(\mathcal{D})$
        \techrep{
        \IF{$\mathfrak{r}\in cnts.keys$} 
        	\STATE $cnts[\mathfrak{r}]+=1$ 
        \ELSE
        	\STATE $cnts[\mathfrak{r}]=1$
        \ENDIF
        }
        \submit{
        \STATE {\bf if} $\mathfrak{r}$ is in $cnts.keys$ {\bf then} $cnts[\mathfrak{r}]$+=$1$ 
        {\bf else} $cnts[\mathfrak{r}]=1$
        }
    \ENDFOR
    \STATE{\bf if} $cnts.keys\backslash \mathfrak{R}_{h-1} = \emptyset$ {\bf then return} null 
    \STATE $\mathfrak{r}_{h} = \underset{\mathfrak{r}\in ~ cnts.keys\backslash \mathfrak{R}_{h-1}}{\mbox{argmax}}\big( cnts[\mathfrak{r}]\big)$
    \techrep{    
    \STATE $S(\mathfrak{r}_{h}) = cnts[\mathfrak{r}_{h}]/(N+N')$
    \STATE $e = Z(1-\frac{\alpha}{2})\sqrt{\frac{S(\mathfrak{r}_{h})(1-S(\mathfrak{r}_{h}))}{N+N'}}$
    }
    \submit{
    \STATE $S(\mathfrak{r}_{h}) = \frac{cnts[\mathfrak{r}_{h}]}{N+N'}$;
     $e = Z(1-\frac{\alpha}{2})\sqrt{\frac{S(\mathfrak{r}_{h})(1-S(\mathfrak{r}_{h}))}{N+N'}}$
    }
    \STATE {\bf return} $\mathfrak{r}_{h},S(\mathfrak{r}_{h}),e$
\end{algorithmic}
\end{algorithm}

Algorithm~\ref{alg:gnr1} uses a hash data structure that contains the aggregates of the rankings it has observed so far.
Upon calling the algorithm, it first draws $N$ sample functions from the region of interest $\ar$.
For each sample function, the algorithm finds the corresponding ranking and checks if it has previously been discovered. If not, it adds the ranking to the hash and sets its count as 1; otherwise, it increments the count of the ranking.
If the number of discovered rankings is at most $h$, the algorithm fails to find a new ranking and returns null.
The algorithm then chooses the ranking that does not belong to top-$h$ and has the maximum count. It computes the stability and confidence error of the ranking and returns it.

Considering a budget of $N$ samples while finding the ranking for each sample, the running time of \getnextr is $O(N\times n\log n)$.

\techrep{
The \getnextr with fixed budget selects the rankings from the ones it has observed. While finding the stable rankings that a relatively large portion of function space generate them is very likely, the rankings that are generated by a small portion of the region is unlikely. First, we would like to remind here that the objective is finding the stable rankings, not the rankings that are very unlikely to generate.
Still, if the user is not satisfied with the discovered rankings, she may continue calling the operator to collect more samples and increase the chance of discovering the rare rankings.
Next, we explain the other version of \getnextr that guarantees the confidence error. We will also explain the expected number of samples that are required for the discovery of the rankings with a specific stability.
}

\techrep{
\subsection{Fixed confidence error}
While the \getnextr with fixed budget consumes constant number of samples at every iteration, it does not provide a bound on the confidence error.
Here, we fix the error while accepting a non-deterministic number of samples per operation.
The algorithm is similar to the previous one and works based the central limit theorem.
However, instead of sampling a fixed number of functions, it continues the sampling until the confidence error for its estimate of the next stable ranking is not more than the required error $e$.
In high-level, the algorithm first finds its estimate of the next stable ranking and computes its confidence error, if it is less than $e$ returns the ranking, otherwise continue the sampling.
Algorithm~\ref{alg:gnr2} shows the pseudocode of this operation.
\begin{algorithm}[!h]
\caption{\getnextr\\
{\bf Input:} database $\mathcal{D}$, region of interest $\ar$, top-$h$ stable rankings $\mathfrak{R}_h$, hash of previous aggregates $counts$, total number of previous samples $N'$, confidence level $\alpha$, confidence error $e$\\
{\bf Output:} $(h+1)$-th stable ranking and its stability
}

\begin{algorithmic}[1]
\label{alg:gnr2}
	\STATE $N = 0$
    \STATE $\mathfrak{r}_{h+1} = \underset{\mathfrak{r}\in ~ counts.keys\backslash \mathfrak{R}_h}{\mbox{argmax}}\big( counts[\mathfrak{r}]\big)$
	\WHILE{{\bf true}}
        \IF{$\mathfrak{r}_{h+1}$ is not null}
    		\STATE $S(\mathfrak{r}_{h+1}) = \frac{counts[\mathfrak{r}_{h+1}]}{N+N'}$
    		\STATE $e' = Z(1-\frac{\alpha}{2})\sqrt{\frac{S(\mathfrak{r}_{h+1})(1-S(\mathfrak{r}_{h+1}))}{N+N'}}$
        	\IF{$e'\leq e$}
        		\STATE $N'+=N$
            	\STATE {\bf return} $\mathfrak{r}_{h+1}, S(\mathfrak{r}_{h+1}), N$
        	\ENDIF
        \ENDIF
        \STATE $w =$ Sample$\ar(\ar)$; $N+=1$
        \STATE $\mathfrak{r} = \mathcal{r}_{f_w}(\mathcal{D})$
        \IF{$\mathfrak{r}$ is in $counts.keys$} 
        	\STATE $counts[\mathfrak{r}]+=1$ 
        \ELSE
        	\STATE $counts[\mathfrak{r}]=1$
        \ENDIF
        \IF{$\mathfrak{r}\notin \mathfrak{R}_h$ {\bf and} $(\mathfrak{r}_{h+1}$ is null {\bf or} $counts[\mathfrak{r}]>counts[\mathfrak{r}_{h+1}])$}
        	\STATE $\mathfrak{r}_{h+1} = \mathfrak{r}$
        \ENDIF
    \ENDWHILE
\end{algorithmic}
\end{algorithm}

At every iteration, Algorithm~\ref{alg:gnr2} looks into the set of rankings it has observed so far.
Therefore, it first needs to observe a ranking at least once to be able to estimate its stability.
Theorem~\ref{th:observationcost} provides the expected cost for observing a ranking $\mathfrak{r}$ with the stability of $S(\mathfrak{r})$.
Following this theorem, \getnextr needs to spends the expected cost of $1/S(\mathfrak{r}_{h+1})$ for discovering $\mathfrak{r}_{h+1}$. 
\begin{theorem}\label{th:observationcost}
Given a ranking $\mathfrak{r}\in\mathfrak{R}$ with the stability of $S(\mathfrak{r})$, the expected number of samples required for observing it is $\frac{1}{S(\mathfrak{r})}$ with the variance of $\frac{1-S(\mathfrak{r})}{S(\mathfrak{r})^2}$.
\end{theorem}
\begin{proof}
Consider the probability distribution provided in Equation~\ref{eq:functiondist} and the random Bernoulli variable $x_\mathfrak{r}$ that is 1 if $\mathcal{r}_{f(\Theta)}(\mathcal{D}) = \mathfrak{r}$ and 0 otherwise.
We want to find the likelihood of the event $x_\mathfrak{r}$ with the success probability of $S(\mathfrak{r})$
happening first at exactly the i-th trial.
This is modeled by the Geometric distribution with the following probability mass function:
$$
p(i;S(\mathfrak{r}))=(1-S(\mathfrak{r}))^{x-1}S(\mathfrak{r})
$$
The mean and variance of such distribution are $\mu=\frac{1}{S(\mathfrak{r})}$ and $\sigma = \frac{1-S(\mathfrak{r})}{S(\mathfrak{r})^2}$.
These provide the expected number of samples, as well as the variance, required for observing $\mathfrak{r}$ for the first time.
\end{proof}

Theorem~\ref{th:observationcost} shows the reverse relation between the stability of a ranking and the cost for its discovery, which indicates that at the beginning that the rankings are stable, the cost of discovering them is low; on the other hand, the cost for discovering the unstable rankings that rarely happen is high.

Algorithm~\ref{alg:gnr2} guarantees the confidence error $e$.
Using the Equation~\ref{eq:costerror}, the expected number of samples for such guarantee for $\mathfrak{r}_{h+1}$ is:
\begin{align}\label{eq:expectedsample4e}
N_{\mathfrak{r}_{h+1}} = S(\mathfrak{r}_{h+1})(1-S(\mathfrak{r}_{h+1}))\big(\frac{Z(1-\frac{\alpha}{2})}{e}\big)^2
\end{align}
\begin{comment}
Using Theorem~\ref{th:observationcost} and Equation~\ref{eq:expectedsample4e} the expected complexity of \getnextr with fixed confidence error is as provided in Theorem~\ref{th:gnr2complexity}.

\begin{theorem}\label{th:gnr2complexity}
Let $X$ be $\max\big(1/S(\mathfrak{r}_{h+1}),
        S(\mathfrak{r}_{h+1})(1-S(\mathfrak{r}_{h+1}))\\(\frac{Z(1-\frac{\alpha}{2})}{e})^2 \big)$. The expected complexity of Algorithm~\ref{alg:gnr2} is\footnote{The expected cost for the top-$k$ version is $O(\max\big(c, (C_s + n) (N'- X) \big)$.}
        $$O(\max\big(c, (C_s + n\log(n)) (N'- X) \big)$$
where $C_s$ is the cost of sampling a function from the region of interest $\ar$.
\end{theorem}
\abol{In the analyses I need to figure out and add the expected size of $counts$.}
\begin{proof}
TBD!
\end{proof}

Intuitively speaking, the first few calls of \getnextr using Algorithm~\ref{alg:gnr2}
quickly discover stable rankings, but need to take enough samples to guarantee the required error; after that, most of the cost is for discovering the less stable rankings as it probably already has collected enough samples to guarantee the  confidence error for the discovered rankings.
\end{comment}
}

%\submit{\vspace{-1mm}}
\subsubsection{\revision Stable top-k items}\label{subsed:partialranking}
When $n$ is large, instead of the complete ranking, the user may be interested in the top-$k$ items.
The top-$k$ items may either be treated as a set or a ranked list.
A company that considers its top-$k$ candidates for the on-site interview is an example of the top-$k$ set model, whereas for a student that wants to apply for the top colleges, the ranking between the top-$k$ colleges is important.

% {\revision
% A challenge for finding the top-$k$ sets is that their regions are not necessarily convex. That is because multiple rankings can have equal top-$k$ sets. Let $\mathfrak{R}_\mathfrak{p}$ be the \techrep{set of }rankings in $\mathfrak{R}$ that have the same top-$k$ set $\mathfrak{p}$. Every function in the region of a ranking $\mathfrak{r}\in\mathfrak{R}_\mathfrak{p}$ produces the same top-$k$ set. Hence, the region of $\mathfrak{p}$ is $\cup_{\forall\mathfrak{r}\in\mathfrak{R}_\mathfrak{p}}R(\mathfrak{r})$ and is not necessarily convex.
% }

Unfortunately the MD algorithm \getnexttb is not applicable here, as different ranking regions may share the same top-$k$ items.
Therefore, the algorithm cannot focus only on a single region, while delaying the others.
Fortunately, the randomized algorithm \getnextr can be used for partial rankings. Instead of maintaining the counts of complete rankings, it counts the occurrences of partial rankings. %, and all subsequent results hold.
In \S~\ref{sec:exp}, we will show that \getnextr is both effective and efficient for top-$k$ items.

%\submit{\vspace{-2mm}}
\section{Unbiased Function Sampling}\label{sec:sampler}
\techrep{
\begin{figure*}[t] 
    \begin{minipage}[t]{0.24\linewidth}
        \centering
        \includegraphics[width = \textwidth]{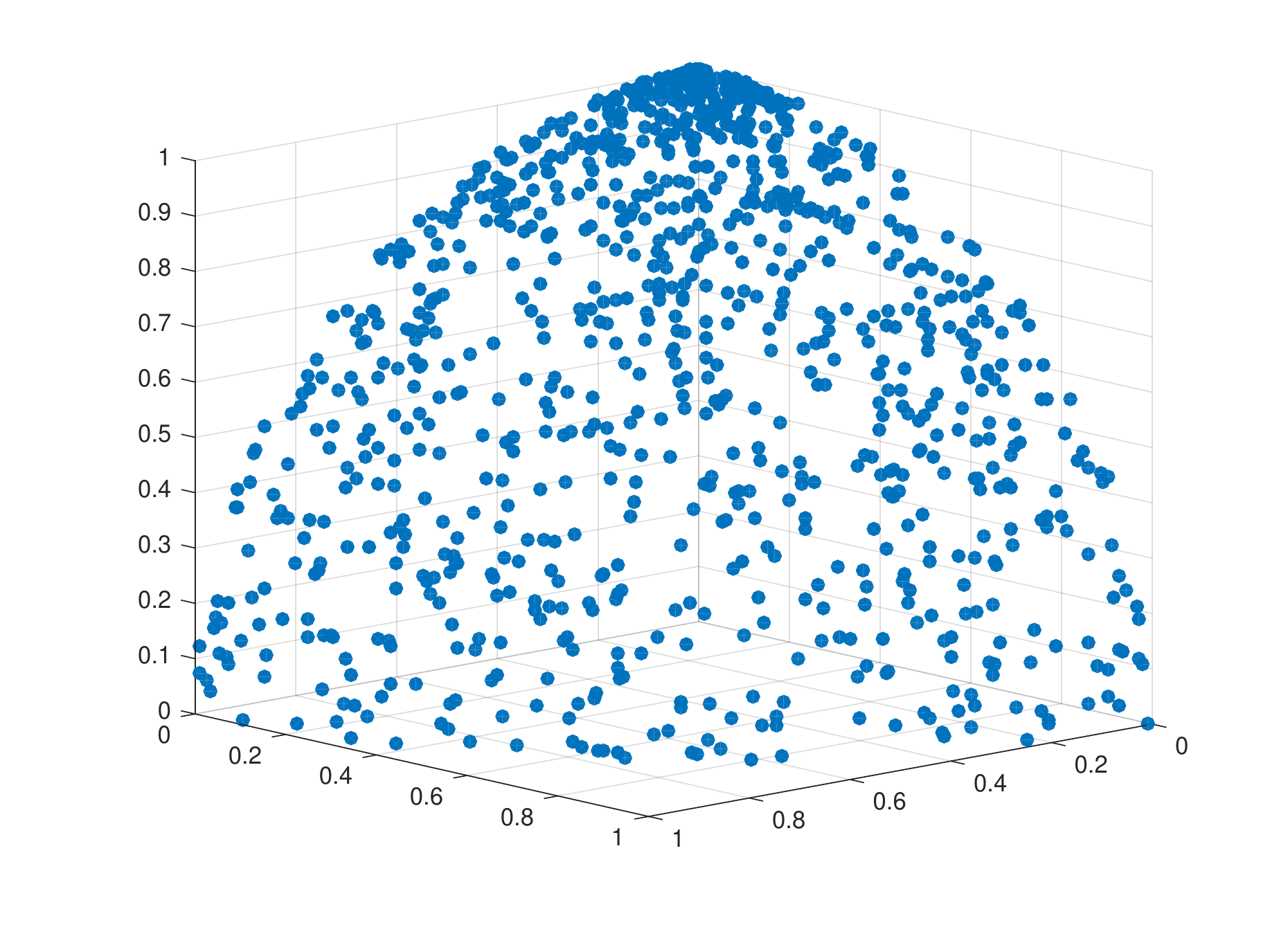}
        \caption{Distribution of 1000 random functions in $\mathbb{R}^3$, generated by taking uniform samples of the angles}
        \label{fig:randf1}
    \end{minipage}
    \hspace{1mm}
    \begin{minipage}[t]{0.24\linewidth}
        \centering
        \includegraphics[width = \textwidth]{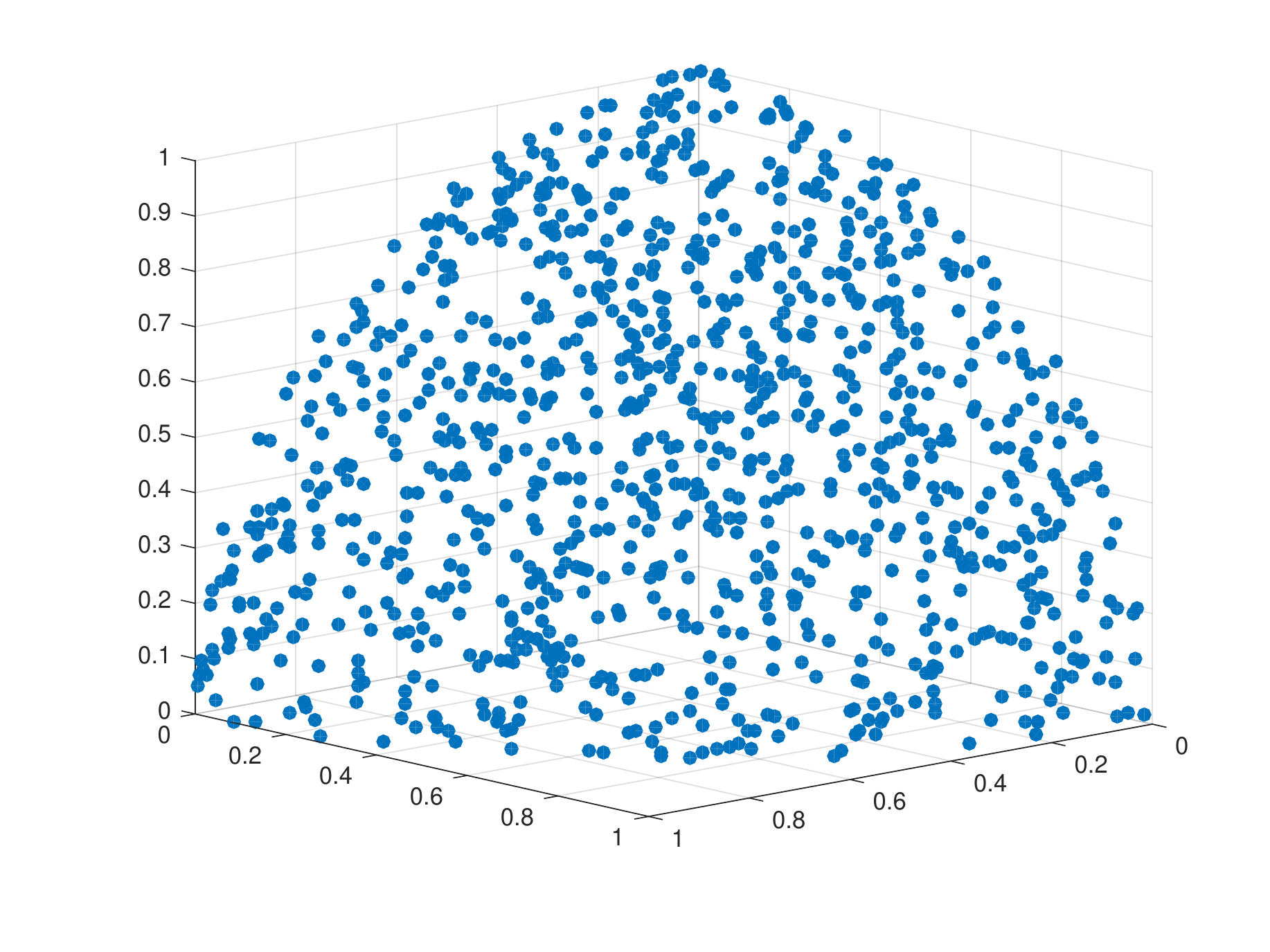}
        \caption{Illustration of 1000 random uniform functions taken in $\mathbb{R}^3$, using Algorithm~\ref{alg:sampu}}
        \label{fig:randf2}
    \end{minipage}
    \hspace{1mm}
    \begin{minipage}[t]{0.24\linewidth}
        \centering
        \includegraphics[width = 0.85\textwidth]{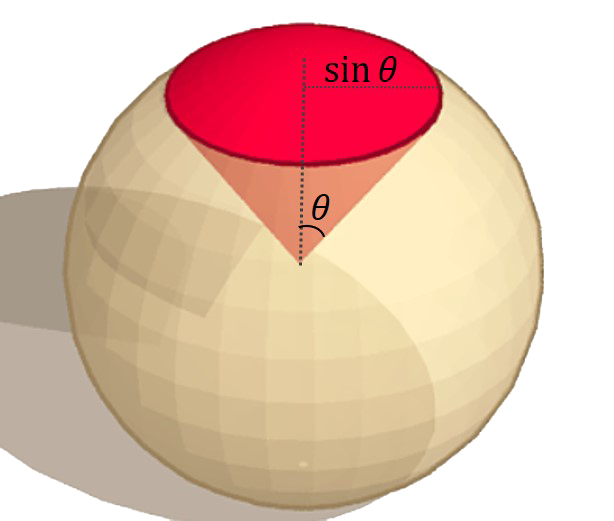}
        \caption{Modeling $\ar$ as a unit $d$-spherical cap around the $d$-th axis}
        \label{fig:roi}
    \end{minipage}
    \hspace{1mm}
    \begin{minipage}[t]{0.24\linewidth}
        \centering
        \includegraphics[width = 1.03\textwidth]{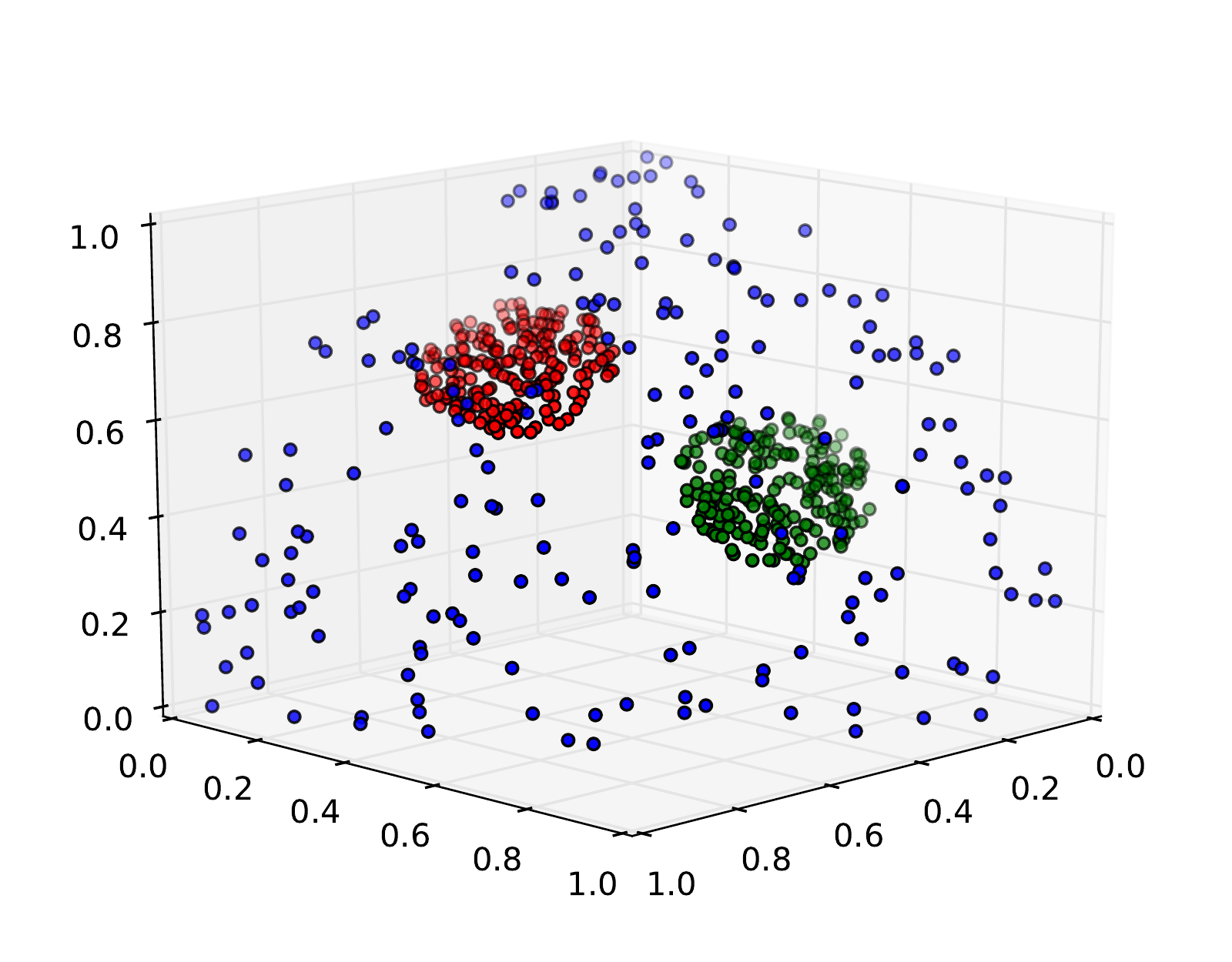}
        \caption{Samples generated using (i) blue: Algorithm~\ref{alg:sampu}, (ii) green: Algorithm~\ref{alg:sampui} and \techrep{Algorithm~\ref{alg:riemann}}\submit{numeric inverse CDF}, (iii) red: Algorithm~\ref{alg:sampui} and Equation~\ref{eq:cdf3d}}
        \label{fig:roi2}
    \end{minipage}
\end{figure*}
}

A uniform sampler from the function space is a key component for devising Monte-Carlo methods, both for estimating the stability of rankings and for designing randomized algorithms for the problem.
% While every ranking is generated by continuous ranges of functions, every function $f$ generates only one ranking of items. Moreover, the larger the volume of a ranking region is (i.e. the more stable it is), the higher is the chance of choosing a random function from it.
% Therefore, {\it uniform sampling of the function space provides the sampling of rankings based on their stability distribution}.
% This suggests using such samples from the function space to design a Monte-Carlo method.
% But first, we need a sampler that generates functions uniformly at random from a region of interest $\ar$.
In the following, we first discuss  sampling from the complete function space and then propose an efficient sampler for $\ar$.

\subsection{Sampling from the function space}\label{subsec:sampleu}
In this subsection we discuss how to generate unbiased samples from the complete function space.
Since every function is represented as a vector of $d-1$ angles, each in range $[0^\circ,\pi/2^\circ]$, one way of generating random functions is by generating angle vectors uniformly at random.
This, however, does not provide uniform random functions sampled from the function space, except for 2D.

\techrep{
To demonstrate this in 3D, we generated 1000
two dimensional random vector with values in range $[0,\pi/2]$. Then, for every vector $\Theta$, we plotted the point with the polar cardinalities $(1,\Theta)$ in Figure~\ref{fig:randf1}. These represent the end points for the unit vectors with angle $\Theta$.
Looking at the figure, it is clear that the distribution is not uniform, as the density of the end points reduces moving from the top of the figure to the bottom.
Also, looking at Figure~\ref{fig:randf1}, one may notice that the set of points in the first quadrant of the unit $d$-sphere represent the function space $\mathcal{U}$.
}
\submit{
As mentioned in \S~\ref{sec:pre}, the set of points in the first quadrant of the unit $d$-sphere represent the function space $\mathcal{U}$.
}
This is because of the one-to-one mapping between the points on the surface of the unit $d$-sphere and the unit origin-starting rays, each representing a function $f$.
Hence, the problem of choosing functions uniformly at random from $\mathcal{U}$ is equivalent to choosing random points from the surface of a $d$-sphere.
As also suggested in~\cite{rrr}, we adopt a method for uniform sampling of the points on the surface of the unit $d$-sphere~\cite{muller1959note, marsaglia1972choosing}.
Rather than sampling the angles, this method samples the weights using the {\em normal distribution}, and normalizes them. 
This method works because the normal distribution function has a constant probability on the surfaces of $d$-spheres with common centers~\cite{marsaglia1972choosing,cramer2016mathematical}.\submit{
Therefore, in order to generate a random function in $\mathcal{U}$, we set each weight as $w_i = |\mathcal{N}(0,1)|$, where $\mathcal{N}(0,1)$ draws a sample from the standard normal distribution.
}
\techrep{
Algorithm~\ref{alg:sampu}, adopt this method to generate random functions from $\mathcal{U}$. 
}

\techrep{
\vspace{3mm}
\begin{algorithm}[!h]
\caption{{\bf Sample$\mathcal{U}$} }
\begin{algorithmic}[1]
\label{alg:sampu}
    \FOR {$i=1$ to $d$} 
        \STATE $w_i = |\mathcal{N}(0,1)|$ \scriptsize{\tt // $\mathcal{N}(0,1)$ generates a random number from the standard normal distribution}
    \ENDFOR
    \STATE {\bf return} $w$
\end{algorithmic}
\end{algorithm}
}

\techrep{
In order to verify the uniformity of the output of Algorithm~\ref{alg:sampu}, we used it to draw 1000 sample functions. Similar to Figure~\ref{fig:randf1}, we plotted the endpoints of the unit vectors of sampled functions in Figure~\ref{fig:randf2} -- which shows the points are distributed uniformly.
}

\subsection{Sampling from the region of interest}\label{subsec:sampleui}
Drawing unbiased samples from a region of interest $\ar$ is critical for finding stable regions.
Given an unbiased sampler for the function space $\mathcal{U}$, an acceptance-rejection method~\cite{lucidl1989random} can be used for drawing samples from $\ar$.
The idea is simple: (i) draw a sample from the function space; (ii) test if the drawn sample is in the region of interest and, if so, accept it; otherwise reject the sample and try again.
Testing if the drawn sample is in the region of interest can be done by:
(a) computing its angle distance from the reference $\rho$ and comparing it with the reference angle $\theta$, if $\ar$ is specified by a ray and angle, or by
(b) checking if the sampled point satisfies the constraint, if $\ar$ is specified by a set of constraints.

The efficiency of this method, however, depends on the acceptance probability $p$, defined by the volume ratio of $\ar$ to $\mathcal{U}$.
The expected number of trials for drawing a sample for such probability is $1/p$.
Hence, it is efficient if the volume of $\ar$ is not small.

Therefore,
in the following, we alternatively propose an inverse CDF (cumulative distribution function) method~\cite{devroye1986sample} for generating random uniform functions from a region of interest. This method is preferred over the acceptance-rejection method when $\ar$ is small. It generates functions with the maximum angular distance of $\theta$ from the reference ray $\rho$.
For a region of interest specified by a set of constraints, the bounding sphere~\cite{fischer2003fast} for the base of its $d$-cone identifies the ray and angle distance that include $\ar$. For such regions of interest, the inverse CDF method enables an acceptance-rejection method with higher acceptance rate, leading to better performance.

Consider $\ar$ as the set of functions with the maximum angle $\theta$ around the ray of some function $f$.
We model $\ar$ by the surface unit $d$-spherical cap with angle $\theta$ around the d-th axis in $\mathbb{R}^d$ (Figure~\ref{fig:roi}).
This is similar to the mapping of $\mathcal{U}$ to the surface of the unit hyperspherical, and is due to the one-to-one mapping between the rays in $\ar$ and the points on the surface unit hyperspherical cap.
We use a transformation that maps the ray of $f$ to the d-D axis. 
After drawing a function we  transform it around the ray of $f$. 

For an angle $\theta$, the $d$-dimension orthogonal plane $x_d=\cos \theta$ partitions the cap from the rest of the $d$-sphere.
Hence, the intersection of the set of the ($d$-th axis orthogonal) planes $\cos \theta\leq x_d\leq 1$ with the $d$-sphere define the cap.

The intersection of each such plane with the $d$-sphere is a $(d-1)$-sphere. For example, in Figure~\ref{fig:roi} the intersection of a plane that is orthogonal to the z-axis with the unit sphere is a circle (2-sphere).
An unbiased sampler should sample points from the surface of such $(d-1)$-spheres proportionally to their areas.
In the following, we show how this can be done.

The surface area of a $\delta$-sphere with the radius $r$ is~\cite{li2011concise}:
\begin{align}\label{eq:area1}
A_\delta (r) = \frac{2\pi^{\delta/2}}{\Gamma(\delta/2)}r^{\delta-1}
\end{align}
$\Gamma$, in the above equation, is the gamma function.

Using this equation, the area of the unit $d$-spherical cap can be stated as the integral over the surface areas of the $(d-1)$-spheres, 
defined by
the intersection of the planes $\cos \theta\leq x_d\leq 1$ with the $d$-sphere, as follows~\cite{li2011concise}:
\techrep{
\begin{align}\label{eq:area2}
\nonumber
A_d^{cap} (1) &= \int_0^\theta A_{d-1}\sin\phi d\phi \\
              &= \frac{2\pi^{d/2}}{\Gamma(d/2)} \int_0^\theta \sin^{d-2}(\phi)d\phi
\end{align}
}
\submit{
\begin{align}\label{eq:area2}
A_d^{cap} (1) &= \int_0^\theta A_{d-1}\sin\phi d\phi
              = \frac{2\pi^{d/2}}{\Gamma(d/2)} \int_0^\theta \sin^{d-2}(\phi)d\phi
\end{align}
}

\submit{
\begin{figure}[t!]
\centering
\subcaptionbox{\label{fig:roi}}{\includegraphics[scale=0.35]{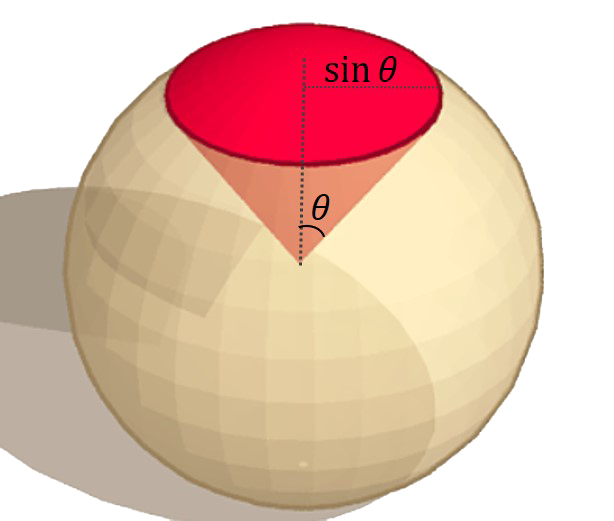}}
\hfill 
\subcaptionbox{\label{fig:roi2}}
{\includegraphics[scale=0.27]{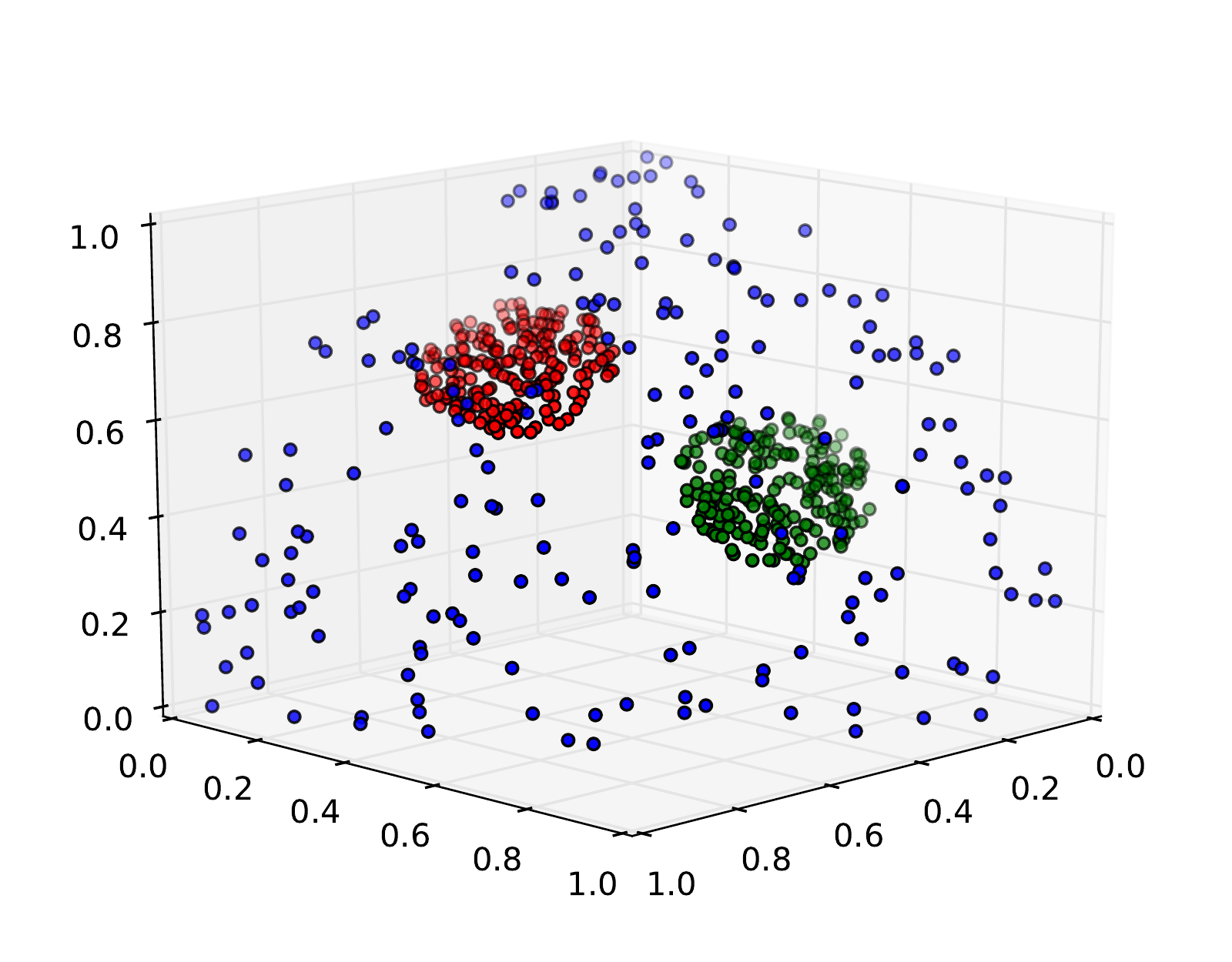}
}
\vspace{-3mm}
\caption{
\label{fig:ROI}
Sampling from $\ar$. a) $\ar$ as a unit $d$-spherical cap around $d$-th axis. b) Samples generated using  {\revision(blue -- scattered over the space)}: \techrep{Algorithm~\ref{alg:sampu}}\submit{\S~\ref{subsec:sampleu}}, {\revision(green -- right cluster)}: Algorithm~\ref{alg:sampui} and \techrep{Algorithm~\ref{alg:riemann}}\submit{numeric inverse CDF}, {\revision(red -- left cluster)}: Algorithm~\ref{alg:sampui} and Equation~\ref{eq:cdf3d}.
}
\vspace{-5mm}
\end{figure}
}

Therefore, considering the random angle $0 \leq x\leq \theta$, the cumulative density function (cdf) for $x$ is given by:

\begin{align}\label{eq:cdf1}
F(x) &= \frac{\int_0^x \sin^{d-2}(\phi)d\phi}{\int_0^\theta \sin^{d-2}(\phi)d\phi}
\end{align}

For a specific value of $d$, one can solve Equation~\ref{eq:cdf1}, find the inverse of $F$ and use it for sampling.
For instance, for $d=3$:
\begin{align}\label{eq:cdf3d}
F(x) = \frac{1-\cos x}{1-\cos\theta}
\Rightarrow F^{-1}(x) = \arccos \big(1-(1-\cos\theta)x\big)
\end{align}

\submit{For a general $d$, we use numerical methods for finding the inverse CDF. Details can be found in the technical report~\cite{techrep}.}
\techrep{
For a general $d$, we can use the representation of $\int_0^\theta \sin^{d-2}(\phi)d\phi$ in the form of beta function and regularized incomplete beta function~\cite{li2011concise} and rewrite Equation~\ref{eq:cdf1} as\footnote{\small $I_z(\alpha, \beta)$ is the regularized incomplete beta function.}:
\begin{align}\label{eq:cdf}
F(x) &= \frac{I_{\sin^2(x)}\big(\frac{d-1}{2}, \frac{1}{2}\big)}{I_{\sin^2(\theta)}\big(\frac{d-1}{2}, \frac{1}{2}\big)}
\end{align}
However, since numeric methods are applied for finding the inverse of the regularized incomplete beta function~\cite{cran1977remark}, we consider a numeric solution for Equation~\ref{eq:cdf1}.
Consider a regular partition of the interval $[0,\theta]$. The integral $\int_0^\theta \sin^{d-2}(\phi)d\phi$ can be computed as the sequence of Riemann sums over the partitions of the interval.
We apply this for computing both the denominator and the nominator of Equation~\ref{eq:cdf1}.
Given the partition of the interval, we start from the angle 0, and for each partition, $x'$ compute the value of $\int_0^{x'} \sin^{d-2}(\phi)d\phi$ as the aggregate over the previous summations and store it in a sorted list.
As a result, in addition to the value of the denominator, we have the value of $F$ for each of the partitions. We will later apply binary search on this list, in order to find the angle $x$ that has the area $F(x)$.
Algorithm~\ref{alg:riemann} shows the pseudocode of the function {\it RiemannSums} that computes the denominator and returns the list of partial integrals divided by the denominator.
In addition to the angle $\theta$, the function takes the number of partitions as the input.
}

\techrep{
\begin{algorithm}[!h]
\caption{{\bf RiemannSums}\\
{\bf Input:} The angle $\theta$ and number of partitions $\gamma$
}
\begin{algorithmic}[1]
\label{alg:riemann}
    \STATE $\epsilon = \theta/\gamma$
    \STATE $L=[0]$; $A=0$; $\alpha= \epsilon$
    \FOR{$i = 1$ to $\gamma$}
        \STATE $A= A+\sin^{d-2}(\alpha)$
        \STATE $L.append(A)$
        \STATE $\alpha = \alpha + \epsilon$
    \ENDFOR
    \STATE {\bf for} $i=1$ to $\gamma$ {\bf do} $L[i] = L[i]/A$
    \STATE {\bf return} $L$
\end{algorithmic}
\end{algorithm}
}% End of TechRep

Algorithm~\ref{alg:sampui} shows the pseudocode of the inverse CDF sampler.
{\revision
For example, consider the example in $\mathbb{R}^3$, where the objective is to generate random numbers around the ray $(\pi/6,\pi/4)$ with angle $\theta=\pi/20$.
The algorithm starts by drawing a random uniform number in range $[0,1]$. 
Let such a random number be 0.13.
}
It takes the list $L$ (computed using the function RiemannSums) as the input and draws a random function from $\ar$.
To do so, it first draws a random uniform number $y$ in the range [0,1].
Next, it applies a binary search on the list of partial integrals to find the range to which $y$ belongs. Considering a fine granularity of the partitions, we assume that the areas of all $(d-1)$-spheres inside each partition are equal.
The algorithm, therefore, returns a random angle (drawn uniformly at random) from the selected partition.
{\revision
Obviously, instead, the algorithm can use the equation of the inverse function.
Continuing with our example, while using Equation~\ref{eq:cdf3d}, the corresponding $y$ value for 0.13 is $\pi/55.5$.
}

\begin{algorithm}[!h]
\caption{{\bf Sample $\ar$}\\
{\bf Input:} The ray $\rho$, angle $\theta$\techrep{, number of partitions $\gamma$, and $L$}
}
\begin{algorithmic}[1]
\label{alg:sampui}
    \STATE $y = U[0,1]$ {\scriptsize \tt // draw a uniform sample in range [0,1]}
     \submit{
     \STATE $x = F^{-1}(y)$
     }
    \techrep{
    \STATE $i =$ {\bf binarySearch}$(y,L)$
    \STATE $\hat{\epsilon} = U[0,\epsilon]$
    \STATE $x = i\epsilon +\hat{\epsilon}$
    }
    \STATE {\bf for} $i=1$ to $d-1$ {\bf do} $\hat{w}_i = \mathcal{N}(0,1)$
    \STATE $\langle \theta_1,\cdots,\theta_{d-2} \rangle = $ the angles in polar representation of $\hat{w}$
    \STATE $w = $ {\bf toCartesian}$(1, \langle \theta_1,\cdots , \theta_{d-2},x \rangle)$
    \STATE {\bf return Rotate}($w$, $\rho$)
\end{algorithmic}
\end{algorithm}

Recall that the angle $x$ specifies the intersection of a plane with the $d$-spherical cap, which is a $(d-1)$-sphere.
Hence, after finding the angle $x$, we need to sample from the surface of a $(d-1)$-sphere, uniformly at random.
{\revision 
For our example in $\mathbb{R}^3$, the intersection is a circle (2-sphere) and, therefore, we need to sample from the surface of the circle.
}
Also, recall from \S~\ref{subsec:sampleu} that the normalized set of $d-1$ random numbers drawn from the normal distribution provide a random sample point on the surface of $(d-1)$-sphere.
The algorithm Sample$\ar$ uses this for generating such a random point.
It uses the angle combination of the drawn random point from the surface of a $(d-1)$-sphere and combines them with the angle $x$ (with the $d$-th axis).
{\revision
In our example in $\mathbb{R}^3$, let the sampled point on the circle have the angle $0.8\pi$. Hence, the angle combination is $\langle 0.8\pi, \pi/55.5\rangle$.
}
After this step, the point specified by the polar coordinates $(1, \langle \theta_1,\cdots , \theta_{d-2},x \rangle)$ is the random uniform point from the surface of $d$-spherical cap around the $d$-th axis.
As the final step, the algorithm needs to rotate the coordinate system such that the center of the cap (currently on $d$-th axis) falls on the ray $\rho$.
We rely on the existence of the function {\it Rotate} for this, presented in \submit{the technical report~\cite{techrep}}\techrep{ Appendix~\ref{app:rotation}}.
\submit{
Figure~\ref{fig:roi2} shows three cases of 200 samples in $\mathbb{R}^3$ drawn from {\revision(blue -- scattered over the space)} $\mathcal{U}$ using \S~\ref{subsec:sampleu} and {\revision(green and red -- right and left clusters)} $\ar$ around the rays $(\pi/3,\pi/3)$ and $(\pi/6,\pi/4)$ with angle $\theta=\pi/20$.
}

\techrep{
The final discussion is when to apply the acceptance-rejection method and when the inverse CDF.
For a fix $d$ that the inverse CDF is calculated (e.g. Equation~\ref{eq:cdf3d}), the sampling cost is constant and, therefore, it is always preferred over the acceptance/reject method.
However, in general where the numeric method is applied for the inverse CDF method (Algorithm~\ref{subsec:sampleu}), the cost $O(\log |L|)$. Specifically, in order to make Algorithm~\ref{alg:sampui} run in $O(\log n)$, one can consider $L=O(n)$.
On the other hand, the (expected) cost of the acceptance/reject method depends on the ration of the volume of $\mathcal{U}$ to $\ar$.
Based on Equations~\ref{eq:area1} and~\ref{eq:area2}, $1/\int_0^\theta \sin^{d-2}(\phi)d\phi$ shows this ratio.
Therefore, if $\log(|L|) > (1/\int_0^\theta \sin^{d-2}(\phi)d\phi)$, the acceptance/reject method is expected to outperform the inverse CDF method; otherwise, the inverse CDF method is preferred.

Figure~\ref{fig:roi2} shows three cases of 200 samples in $\mathbb{R}^3$ where (i) the blue points are sampled from function space using Algorithm~\ref{alg:sampu}, (ii) green points are generated around the ray $(\pi/3,\pi/3)$ with angle $\pi/20$ using Algorithm~\ref{alg:sampui} and Algorithm~\ref{alg:riemann}, and (iii) red points are generated around the ray $(\pi/6,\pi/4)$ with angle $\pi/20$ using Algorithm~\ref{alg:sampui} while using Equation~\ref{eq:cdf3d} for the inverse CDF.
}

\techrep{
\subsection{The stability oracle}\label{subsec:SOracle}
In \S~\ref{sec:md}, we relied on the existence of a stability oracle, that given a region, in the form of the intersection a set of halfspaces, computes its stability.
Here, we explain the details of that oracle.
The unbiased sampling from $\ar$ enables the Monte-Carlo methods~\cite{montecarlo} for estimating the stability of the ranking regions.
We will also use Monte-Carlo methods for devising randomized algorithms in \S~\ref{sec:randomized}.

Based on Definition~\ref{def:1}, the stability of a region $R$ is equal the ratio of its volume to the one of the region of interest $\ar$.
Therefore, using the Monte-Carlo estimation, given a set of samples $\mathcal{S}$, drawn uniformly at random from $\ar$, the stability of $R$ can be estimated as the ratio of samples fall inside it.

Consider a region $R$ as the intersection of a set of halfspaces. Recall from \S~\ref{sec:md} that every halfspace can be considered as an inequality in the forms of $\sum h[i]x_i<0$ for a negative halfspace $h^-$ and as $\sum h[i]x_i>0$ for a positive one.
Checking if a sample function with the weight vector $w$ falls inside a region $R$ can be done by computing the values of $\sum h[i]w_i$ for the halfspaces defining $R$.
The function falls inside $R$, if it satisfies all of the constraints enforces by the halfspaces defining it.

Considering $|R|$ as the number of half-spaces defining the region $R$, the stability oracle is in $O(|R|~|\mathcal{S}|)$.

Algorithm~\ref{alg:Soracle} shows the pseudocode of the stability oracle.
\begin{algorithm}[!h]
\caption{{\bf S} {\scriptsize \tt // the stability oracle}\\
{\bf Input:} A set $R$ of halfspaces defining a region and a set $\mathcal{S}$ of unbiased samples taken from $\ar$\\
{\bf Output:} The stability of $R$ in $\ar$
}
\begin{algorithmic}[1]
\label{alg:Soracle}
    \STATE count$=0$
    \FOR {sample $w$ in $\mathcal{S}$}
        \STATE flag = true
        \FOR{halfspace $(h,$sign$)$ in $R$}
            \STATE sum = $\sum_{i=1}^d h[i]w_i$
            \IF{(sign is + and sum$<0$) or (sign is - and sum$>0$) }
                \STATE flag = false; {\bf break}
            \ENDIF
        \ENDFOR
        \STATE{\bf if} flag=true {\bf then} count $=$ count $+1$
    \ENDFOR
    \STATE {\bf return} $\frac{\mbox{count}}{|\mathcal{S}|}$
\end{algorithmic}
\end{algorithm}
}%end of techrep

\techrep{
\subsection{Partitioning samples for an arrangement}\label{subsec:par}
The MD algorithm \getnexttb in \S~\ref{subsec:thbased} uses the function\\ {\em PassThrough} to check if a hyperplane $h$ intersects with a region $R$.
One can use linear programming and develop the function.
Instead, here we propose a method that uses the set of samples $\mathcal{S}$ that are used for stability estimation for this purpose. The idea is that $h$ intersects $R$, if there exists two points $w$ and $w'$ in $R$ such that $w$ belongs to $h^-$ and $w'$ is in $h^+$, i.e. $\sum_{i=0}^d h[i]w_i<0$ and $\sum_{i=0}^d h[i]w'_i>0$. Since the sample points are used for estimating the stability of the regions, while we are interested in finding the stable regions, $\mathcal{S}$ can also be used for checking if $h$ intersects $R$.
To do so, we partition the samples $\mathcal{S}$ between the current regions of the arrangement, and for a hyperplane $h$, we check the existence of two points in the given region $R$ that fall in two sides of $h$.

Interestingly, despite the complexity of the arrangement and independent from the number of dimensions, a one-dimensional array can be used for partitioning the points.
Recall that the arrangement construction is iterative, and that every time a hyperplane $h$ is added to a region $R$, the algorithm partitions it into two sides: the intersection of $R$ with $h^-$ and its intersection with $h^+$.

Consider the set $\mathcal{S}$ of samples, drawn from $\ar$. Adding the first hyperplane $H[1]$, partitions $\mathcal{S}$ to $\mathcal{S}\cap H[1]^-$ and $\mathcal{S}\cap H[1]^+$.
We can apply the famous {\em partition} algorithm, used in quick-sort~\cite{clrs}, and partition the samples in two sets such that all the samples in the range $[\mathcal{S}[1],\mathcal{S}[i]]$ belong to $R_1=\{H[1]^-\}$ and all the ones with indices $i+1$ to $|\mathcal{S}|$ belong to$R_2=\{H[1]^+\}$.
We use the fields $sb$ and $se$ in the Region data structure (Figure~\ref{fig:region}) to identify the ranges of samples that fall into the region.
Interestingly, for every region $R$,
applying the partition algorithm on the points between $R.sb$ and $R.se$ for a hyperplane $h$ both checks if $h$ intersects with $R$ and adds it to $R$ in the case of intersection. If the output index $i$ from the partition algorithm is $R.sb-1$ or $R.se$, $h$ does not intersects $R$, otherwise the samples are already partitioned to $R_l= R\cap \{H[1]^-\}$ and $R_r=R\cap \{H[1]^+\}$, while $(R_l.sb,R_l.se)=(R.sb,i)$ and $(R_r.sb,R_r    .se)=(i+1,R.se)$.

Partitioning the samples also enables the constant time computation of the stability estimation as:
%\submit{$S(R) = (R.se-R.sb+1)/|\mathcal{S}|$.}
%\techrep{
$$S(R) = \frac{R.se-R.sb+1}{|\mathcal{S}|}$$
%}
}

\begin{figure*}[ht]
	\begin{minipage}[t]{0.33\linewidth}
        	\centering\includegraphics[width=\textwidth]{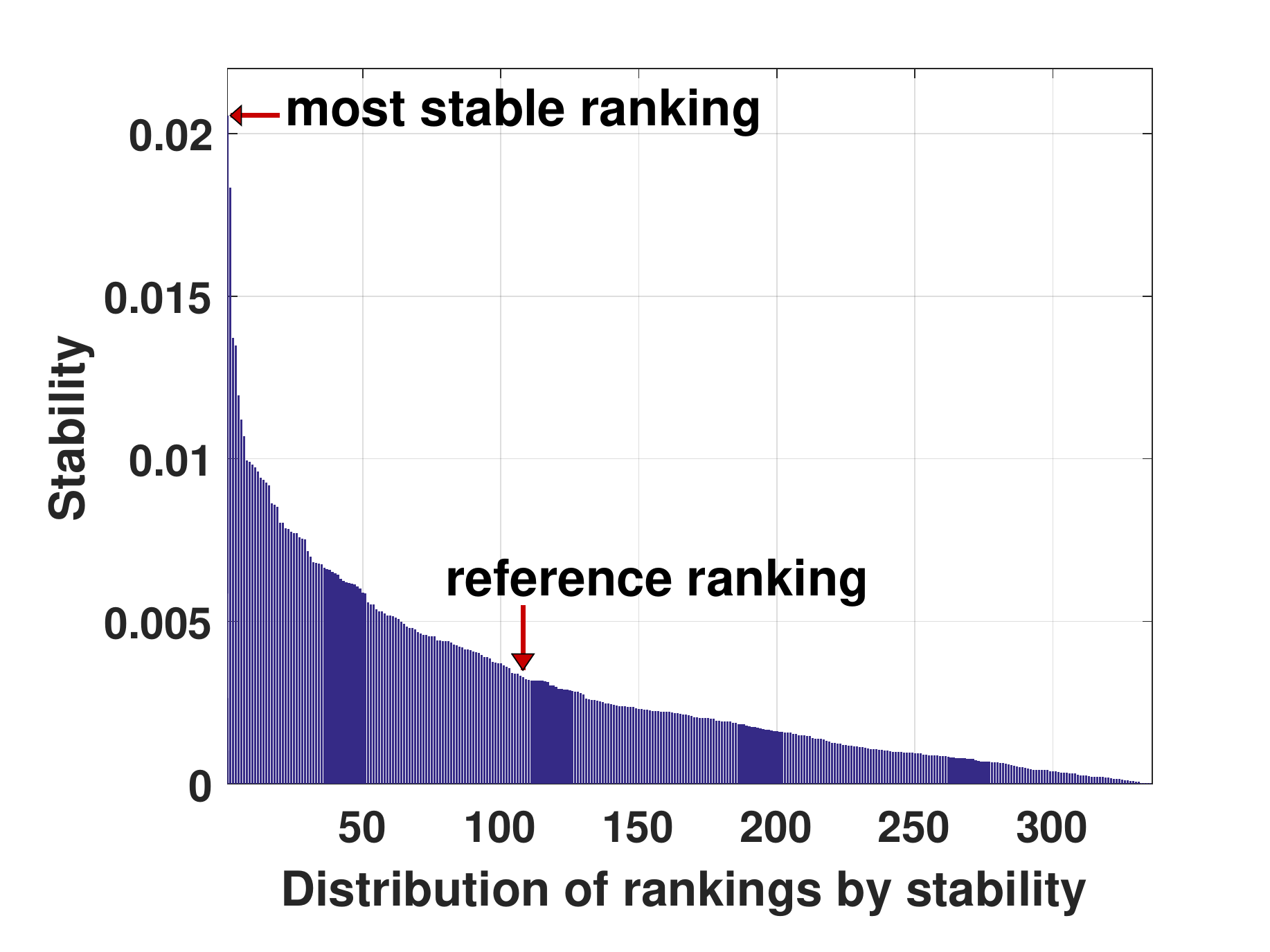}
    		\caption{CSMetrics: overall distribution of rankings by stability.}
    		\label{fig:cms1}
    \end{minipage}
    \hspace{1mm}
    \begin{minipage}[t]{0.33\linewidth}
        \centering
    		\includegraphics[width=\textwidth]{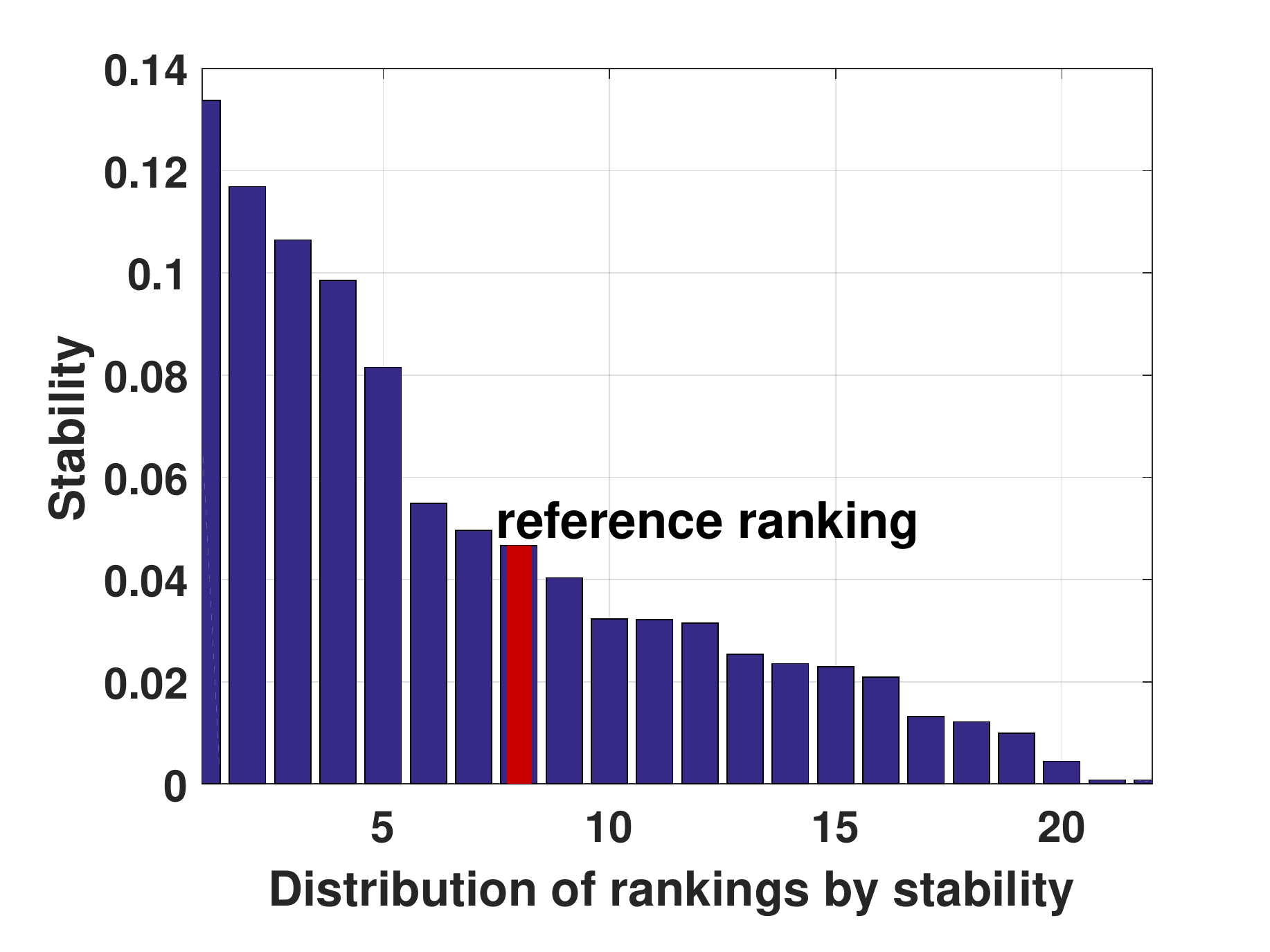}
    		\caption{CSMetrics: stability around reference vector $\langle 0.3,0.7 \rangle$ with 0.998 cosine sim.}
    		\label{fig:cms2}
    \end{minipage}
    \hspace{1mm}
	\begin{minipage}[t]{0.33\linewidth}
        \centering
        \includegraphics[width=\textwidth]{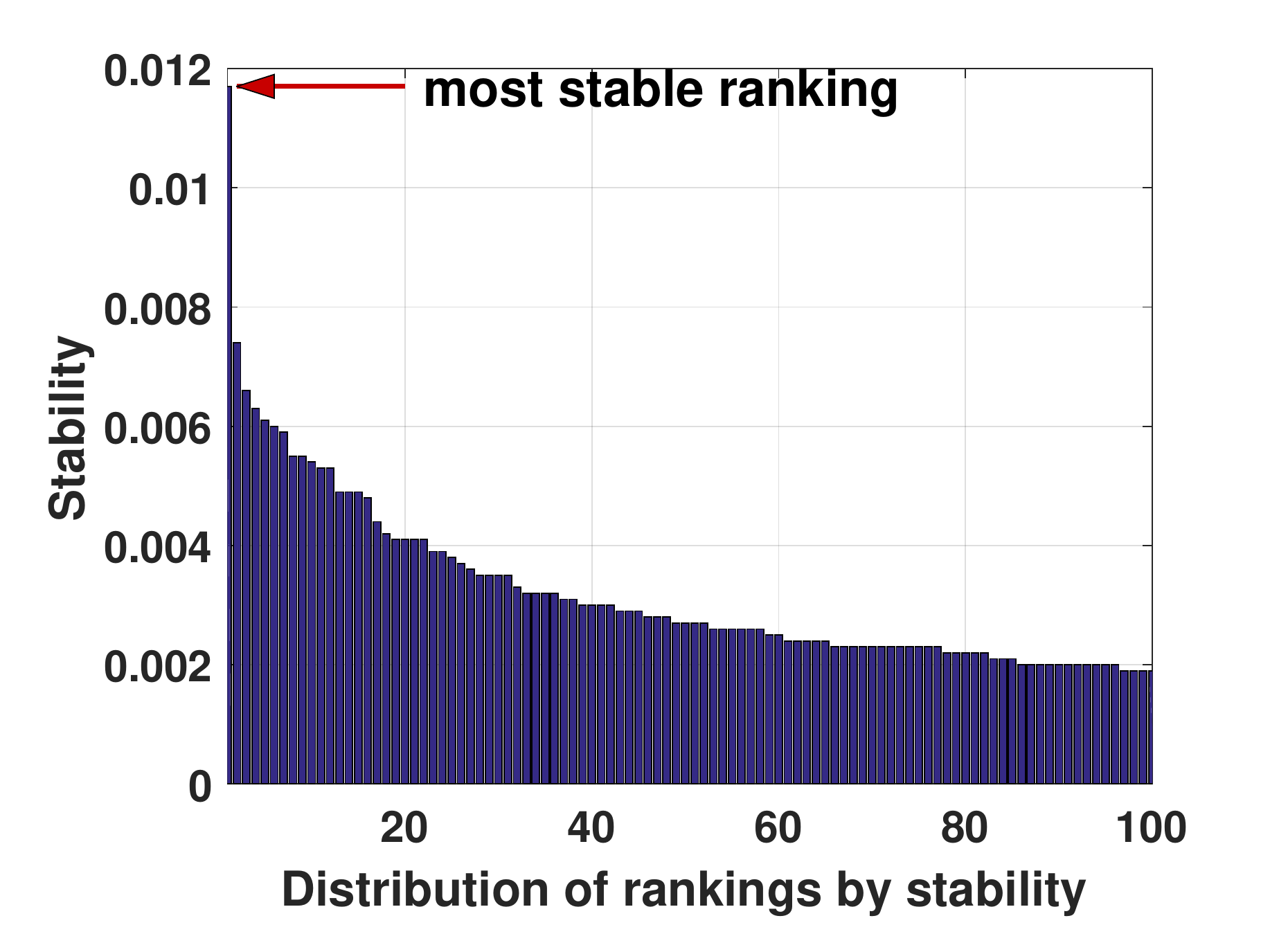}
        \caption{FIFA: stability around reference vector $\langle 1, 0.5, 0.3, 0.2 \rangle$ with 0.999 cosine sim.}
        \label{fig:fifa1}
    \end{minipage}  
\submit{\vspace{-5mm}}
\end{figure*}

%\submit{\vspace{-2mm}}
\section{Experiments}\label{sec:exp}
Here we validate our stability measure and evaluate the efficiency of our algorithms on three real datasets used for ranking.  In particular, we study the stability of two of our datasets in \S~\ref{subsec:expValidation}, showing that the proposed reference rankings are not stable. In \S~\ref{subsec:expperformance}, we study the running times of our algorithms, including stability verification, the \getnext problems in 2D and MD, as well as the randomized algorithm and top-$k$ items.

%\submit{\vspace{-1mm}}
\subsection{Experimental setup} \label{sec:exp:setup}

\stitle{Hardware and platform.}
The experiments were conducted using a 3.8 GHz Intel Xeon processor, 64 GB memory, running Linux.  The algorithms were implemented using Python 2.7.

\stitle{Datasets} 
{\revision
We use four real datasets CSMetrics ($d=2$), FIFA ($d=4$), Blue Nile ($d=5$), and Department of Transportation ($d=3$), as well as a set of three synthetic datasets described below.
}
% Our three datasets are CSMetrics ($d=2$), FIFA ($d=4$), and Blue Nile ($d=5$), as described below:

{\it CSMetrics~\cite{csmetrics}:}
CSMetrics ranks computer science research institutions based on publication metrics.  
%Citation count is considered as the measure of impact. However, since it takes several years that the newer publications to accumulate citations, the (geometric) mean of citations per full research paper in each venue over the last ten years is considered as the predicted count for the newer papers.
For each institution, a combination of {\tt measured} ($M$) citations and an attribute intended to capture future citations, called {\tt predicted} ($P$), is used for ranking, according the score function: $(M)^\alpha(P)^{1-\alpha}$, for parameter $\alpha$.  This score function is not linear, but under a transformation of the data in which $x_1=log(M)$ and $x_2=log(P)$ we can write an equivalent score function linearly as: $\alpha x_1 + (1-\alpha)x_2$. 
%Note that this is the linear combination over the log values of attributes as: 
%\begin{align}\label{eq:csm}
%\alpha\log(M)+(1-\alpha)\log(P)
%\end{align}
The CSMetrics website uses $\alpha=.3$ as the default value, but allows other values to be selected. We use $\alpha=.3$ %to generate a reference ranking. We 
and restrict our attention to the top-100 institutions according to this ranking.
% The system still selects the preselected value of $\alpha=0.3$ which offers an ``original ranking''. Since the user does not usually have the sense to choose a meaningful value for alpha, it is natural to accept the ranking as is.
%Still, as highlighted on the website, the ranking objective is to help ``prospective students, parents, faculty, and University administrators'' to {\em make decisions}.
%Therefore, we collected its set of top $100$ research institutes as one of our datasets for the experiments.

{\it FIFA Rankings~\cite{fifa1}:} 
The FIFA World Ranking of men's national football teams is based on measures of recent performance. Specifically, the score of a team $t$ depends on team performance values for $A_1$ (current year), $A_2$ (past year), $A_3$ (two years ago), and $A_4$ (three years ago).  The given score function, from which the reference ranking is derived, is: $t[1] + 0.5 t[2] + 0.3 t[3] + 0.2 t[4]$.
FIFA relies on these rankings for modeling the progress and abilities of the national-A soccer teams~\cite{fifa} and to seed important competitions in different tournaments, including the 2018 FIFA World Cup.
We consider the top 100 teams in our experiments.

%is a ranking system for men's national teams, provided by the F\'ed\'eration Internationale de Football Association (FIFA).
%The ranking is point-based, based on the performance of the teams in the last four years. 

%FIFA relies on these rankings for modeling the progress and abilities of the national soccer teams, as ``a reliable measure for comparing national A-teams''~\cite{fifa}. Despite the trust of FIFA to these rankings, there have been many critics that question their validity. Ranking Brazil as 22 in 2014, the U.S. as 4 in 2006, and Belgium as number one in Nov. 2015 are some examples that got widely criticized. Still, these rankings are used as part of the calculation, or the entire grounds to seed competitions in different tournaments, including the Russia 2018 FIFA World Cup's final draw.

{\it Blue Nile~\cite{bluenile}:} 
Blue Nile is the largest online diamond retailer in the world.  We collected its catalog that contained 116,300 diamonds at the time of our experiments.
We consider the scalar attributes {\tt Price}, {\tt Carat}, {\tt Depth}, {\tt LengthWidthRatio}, and {\tt Table} for ranking. For all attributes, except {\tt Price}, higher values are preferred.
We normalize each value $v$ of a higher-preferred attribute $A$ as $(v-\min(A))/(\max(A)-\min(A))$; for a lower-preferred attribute $A$, we use $(\max(A) - v)/(\max(A)-\min(A))$.

{\revision {\it Department of Transportation (DoT)~\cite{DoT}:}}
{\revision The flight on-time dataset is published by the US Department of Transportation.  
We collected a set of 1,322,023 records, for the flights conducted by the 14 US carriers in the last three months of 2017. We consider the attributes {\tt air-time}, {\tt taxi-in} and {\tt taxi-out} for ranking.
}

{\revision {\it Synthetic Data:}}
{\revision In order to study the effect of the correlation between the attributes, using the code provided by~\cite{skylineoperator}, we generated three synthetic datasets (independent, correlated, anti-correlated), containing 10,000 items and three scoring attributes in range $[0,1]$.
}

\submit{\vspace{-1mm}}
\subsection{Stability investigation of real datasets}\label{subsec:expValidation}
For the two datasets which provide a reference ranking (CSMetrics and FIFA) we assess these rankings below\techrep{, using our stability measure and our algorithmic tools}. 

\stitle{CSMetrics:}
%We start the validation by studying the computer science research institutes rankings by CSMetrics. The CSMetrics both provides a default ranking and allows ``designing'' the ranking by choosing the ranking parameter $\alpha$.
%As a result, it can be viewed from the angles of both consumer and producer of rankings. Recall that this is a two dimensional ranking done over the {\tt measured} ($M$) and {\tt predicted} ($P$) citations, using Equation~\ref{eq:csm}.
%Besides verifying the stability of the original ranking, we consider two scenarios for producing the stable rankings: (i) considering the complete function space for ranking, and (ii) limiting the scope of design to 0.998 cosine similarity ($\theta=\pi/50$) around the original weight vector ($\alpha=0.3$).
%We use the algorithms \svtd and \getnexttd for stability verification and the \getnext operator, respectively.
\submit{Two attributes are used for ranking here}
\techrep{In this dataset, items are scored according to two attributes} (i.e. $d=2$). We can therefore use the \getnext operator repeatedly to enumerate all feasible rankings and their stability values.  While an upper bound on the number of rankings for $n=100$ and $d=2$ is around $5,000$, the actual number of feasible rankings for this dataset is $336$.  Figure~\ref{fig:cms1} shows the distribution of rank stability across all rankings, showing a few rankings with high stability, after which stability rapidly drops.  
%The most stable ranking is $1.5$ times more stable than the third most stable ranking.  
The reference ranking is highlighted in Figure~\ref{fig:cms1}; using \svtd, we calculated the stability of the reference ranking to be $0.0032$. Notably, the reference ranking did not appear in the top-$100$ stable rankings (it is the $108^{th}$ most stable ranking).

%\gmref{Also, considering the number of rankings, the stability of the original ranking is near the stability of if the rankings were of uniformly distributed.}{This requires more explanation; I was tempted to cut it.  Is it very important?}

%The $\alpha$ value for the function in the middle of the region of the most stable ranking is $0.608$.

Maximizing stability would cause a number of changes compared with the reference ranking.  For example, Cornell University is not in the top-$10$ universities in the reference ranking, but replaces the University of Toronto in the top-$10$ in the most stable ranking. One of the bigger changes in rank position is Northeastern University which improves from $40$ in the reference ranking to $35$ in the most stable ranking.

We also study stability for an acceptable region close to the reference ranking.  We choose 0.998 cosine similarity ($\theta=\pi/50$) around the weight vector of the reference score function. There are $22$ feasible rankings in this acceptable region; their stability distribution is shown in Figure~\ref{fig:cms2}. Even in this narrow region of interest, the reference ranking is far below the maximum stability. 

\stitle{FIFA Rankings:}
Next, we evaluate the higher-dimensional FIFA rankings that are used for important decisions such as seeding different tournaments, including the 2018 FIFA World Cup.
We focus on an acceptable region defined by 0.999 cosine similarity ($\theta=\pi/100$) around the reference weights used by FIFA, i.e. $w=\langle 1, 0.5, 0.3, 0.2 \rangle$.
Using the MD algorithm \getnexttb, we conducted 100 operations to get the distribution of the top-100 stable rankings around the reference weight vector. We considered 10,000 samples drawn using Algorithm~\ref{alg:sampui} for estimations. Figure~\ref{fig:fifa1} shows the distribution of stable rankings.  
First, since $d=4$, there are many feasible rankings, even in such a narrow region of interest, with a significant drop in stability after the most stable rankings, as was the case for CSMetrics. 

Perhaps the most interesting observation is that {\em the reference ranking did not appear in the top-100 stable rankings} (as a result it is not highlighted in Figure~\ref{fig:fifa1}).  While FIFA advertises this ranking as ``a reliable measure for comparing the national A-teams''~\cite{fifa}, 
our finding questions FIFA's trust in such an unstable ranking for making important decisions such as seeding the world cup.
%the lack of stability of the ranking suggests that it may be arbitrary, at best, or intentionally engineered for a specific outcome, at worst.
To highlight an example, while Tunisia holds a higher rank than Mexico in the reference ranking, Mexico is ranked higher in the most stable ranking. % Interestingly, both these teams are present in the 2018 world cup.
This supports the many critics that have questioned the validity of FIFA rankings in the recent past. Examples of controversial rankings include Brazil at 22 in 2014, the U.S. at 4 in 2006, and Belgium at 1 in 2015.

%FIFA relies on these rankings for modeling the progress and abilities of the national soccer teams, as ``a reliable measure for comparing national A-teams''~\cite{fifa}. Despite the trust of FIFA to these rankings, there have been many critics that question their validity. Ranking Brazil as 22 in 2014, the U.S. as 4 in 2006, and Belgium as number one in Nov. 2015 are some examples that got widely criticized. Still, these rankings are used as part of the calculation, or the entire grounds to seed competitions in different tournaments, including the Russia 2018 FIFA World Cup's final draw.

%Next, we study the FIFA rankings which are considered as 

\begin{figure*}[ht]
    \begin{minipage}[t]{0.23\linewidth}
        \centering
        \includegraphics[scale=0.24]{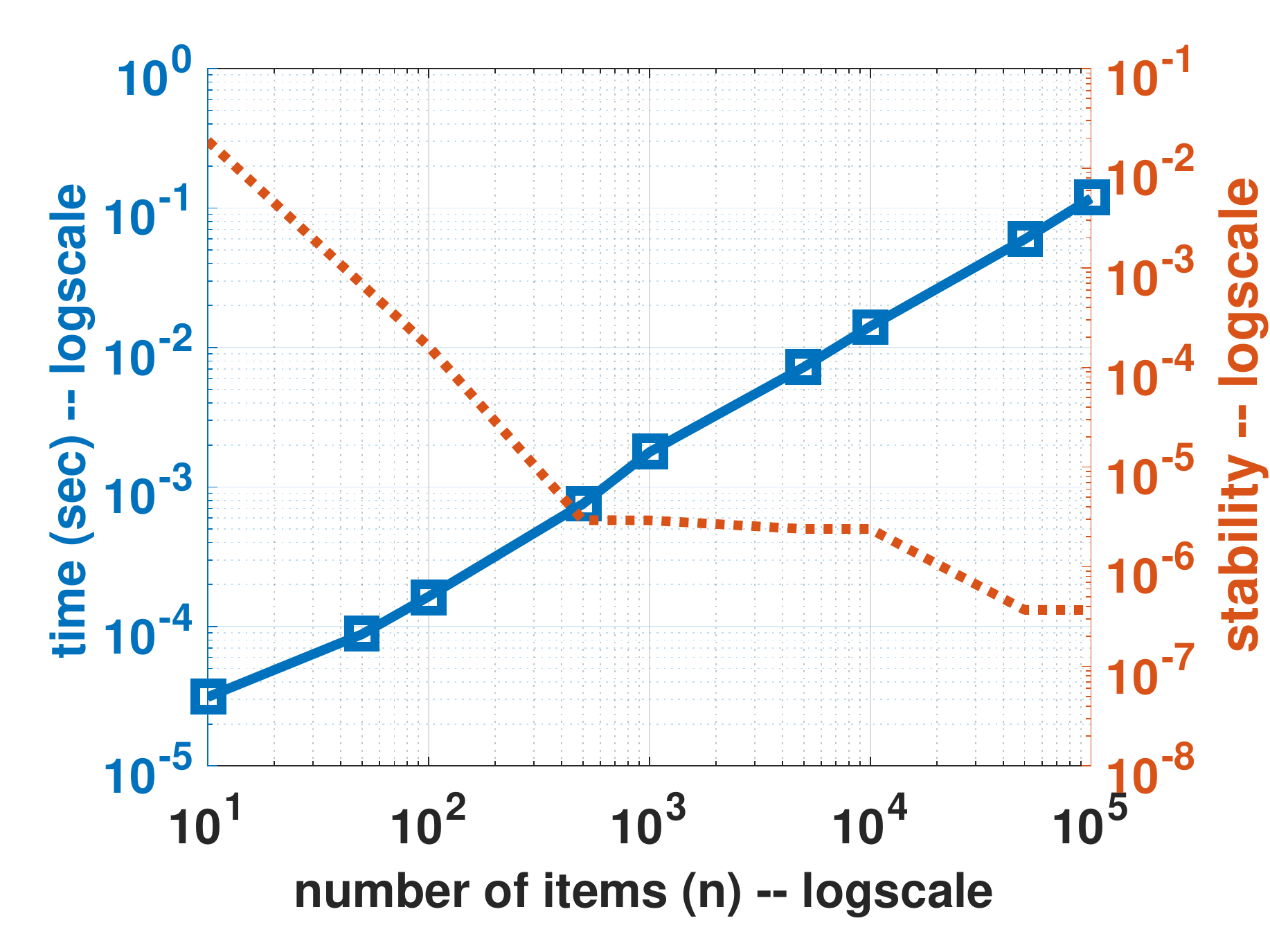}
        \vspace{-8mm}\caption{2D: Stability verification, Impact of dataset size ($n$)}
        \label{fig:bn2dsv}
    \end{minipage}
    \hspace{1mm}
    \begin{minipage}[t]{0.23\linewidth}
        \centering
        \includegraphics[scale=0.24]{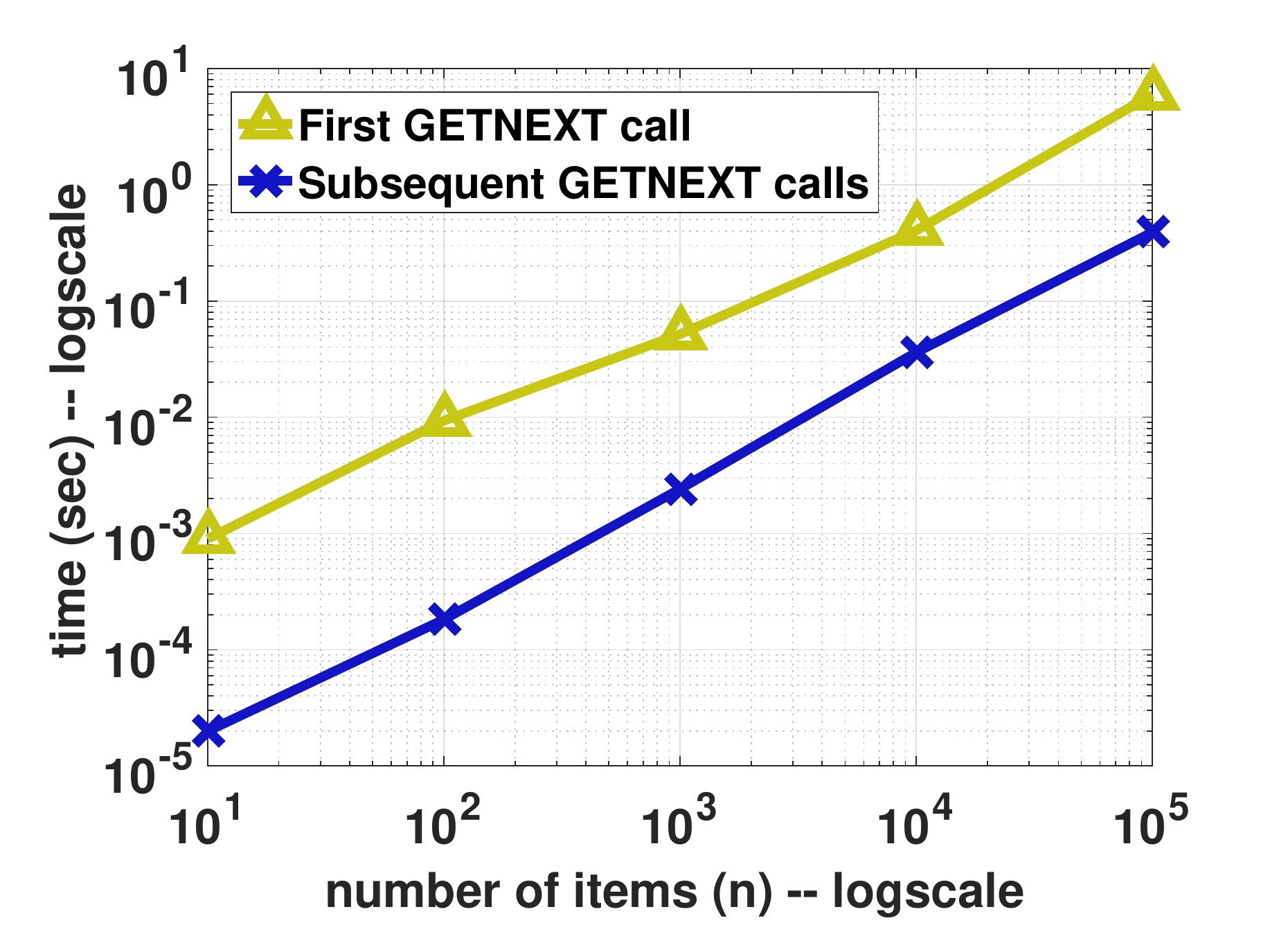}
        \vspace{-8mm}\caption{2D: \getnext, Impact of dataset size ($n$)}
        \label{fig:BN2D2t}
    \end{minipage}
    \hspace{1mm}
    \begin{minipage}[t]{0.23\linewidth}
        \centering
        \includegraphics[scale=0.24]{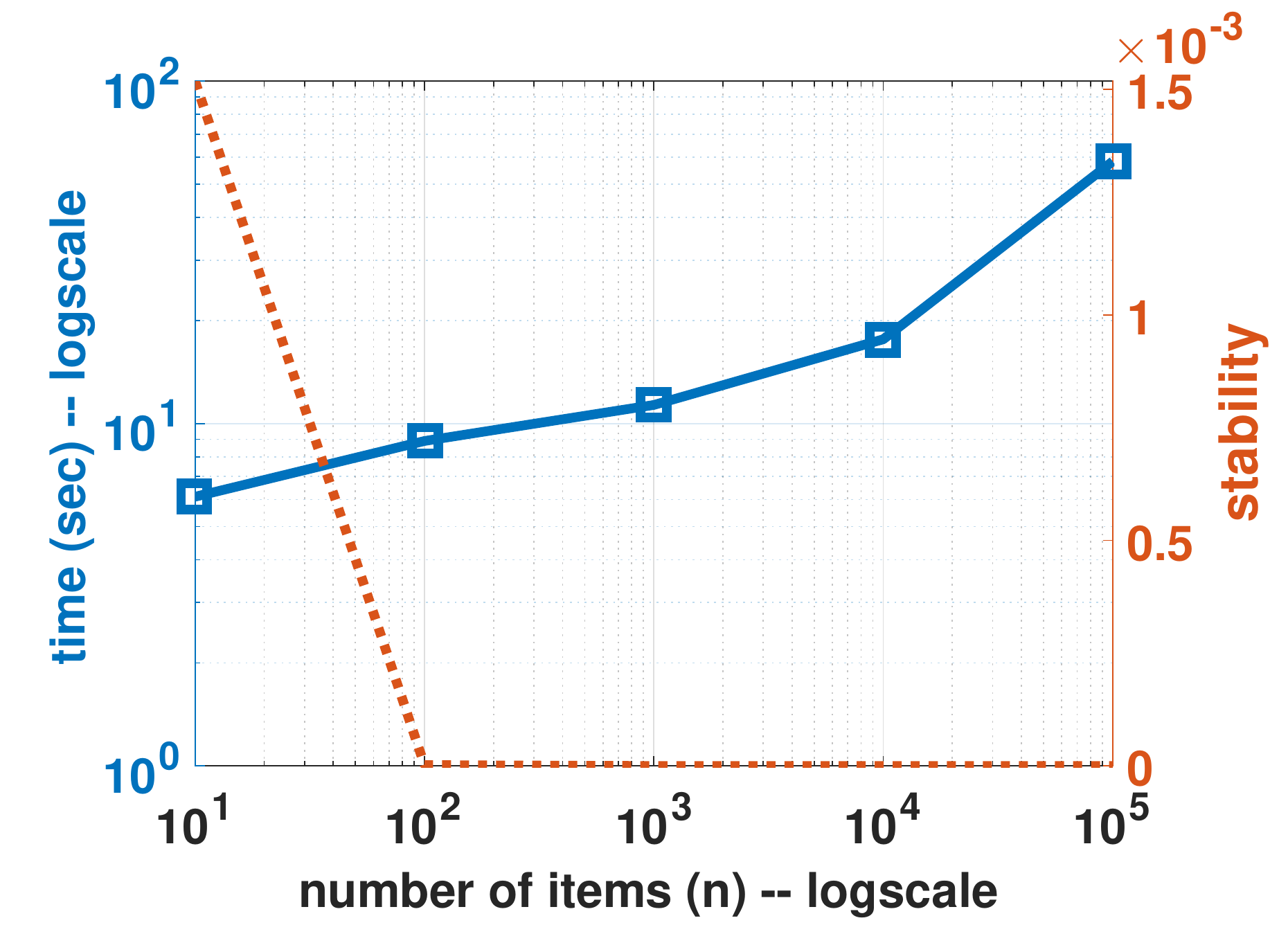}
        \vspace{-8mm}\caption{MD: Stability verification, Impact of dataset size}
        \label{fig:bnmdsvn}
    \end{minipage}
    \hspace{1mm}
%     \begin{minipage}[t]{0.23\linewidth}
%         \centering
% %        \includegraphics[scale=0.24]{plots/BNMD1}
%         \vspace{-8mm}\caption{MD: Stability verification, Impact of number of attributes ($d$)\abol{*}}
%         \label{fig:bnmdsvs}
%     \end{minipage}
	\begin{minipage}[t]{0.23\linewidth}
        \centering
        \includegraphics[scale=0.24]{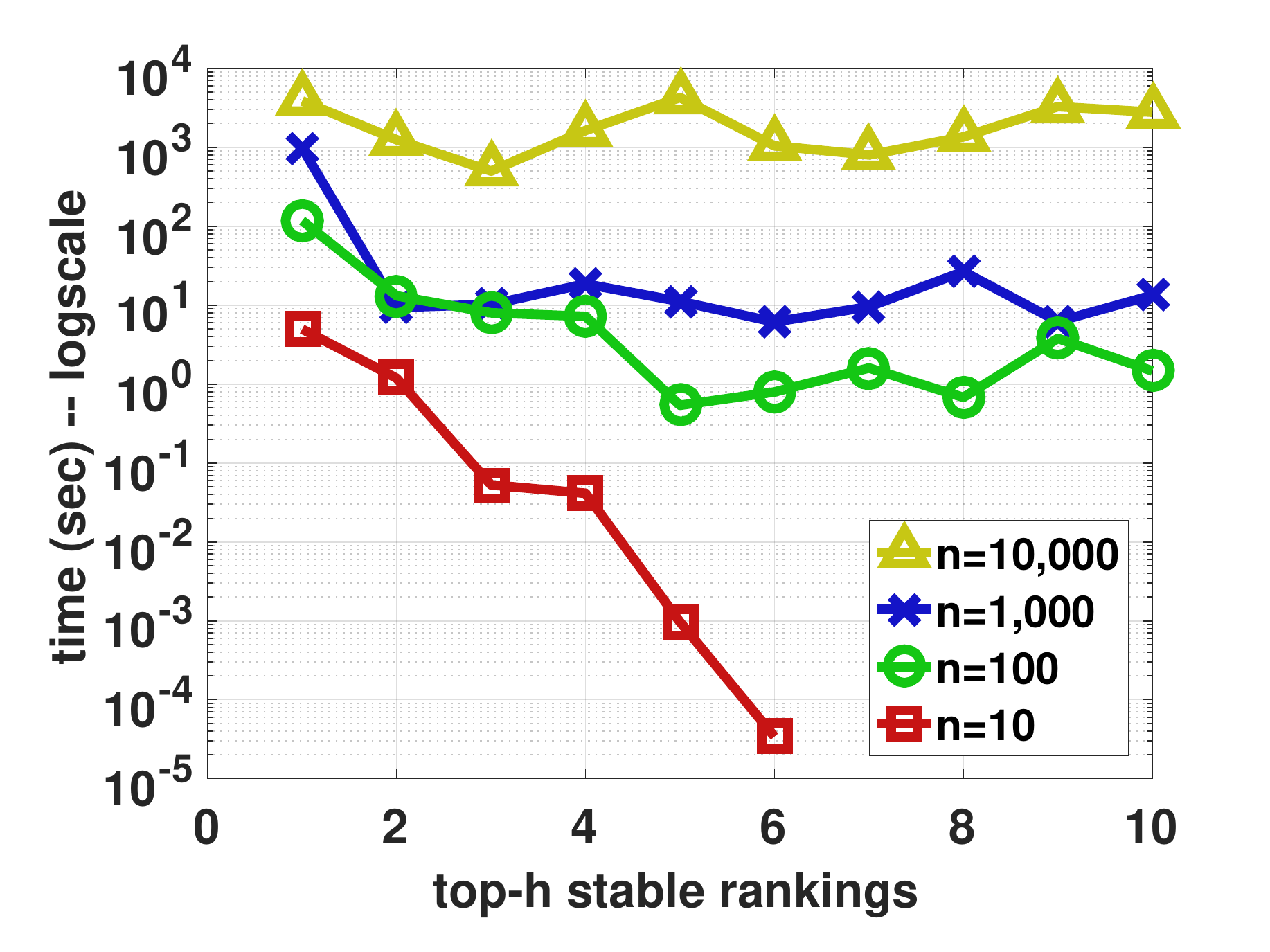}
        \vspace{-8mm}\caption{MD: stable rankings, Impact of dataset size ($n$)}
        \label{fig:BNMDVNt}
    \end{minipage}
    \submit{\vspace{-3mm}}
\end{figure*}

\begin{figure*}[ht]
    \begin{minipage}[t]{0.23\linewidth}
        \centering
        \includegraphics[scale=0.24]{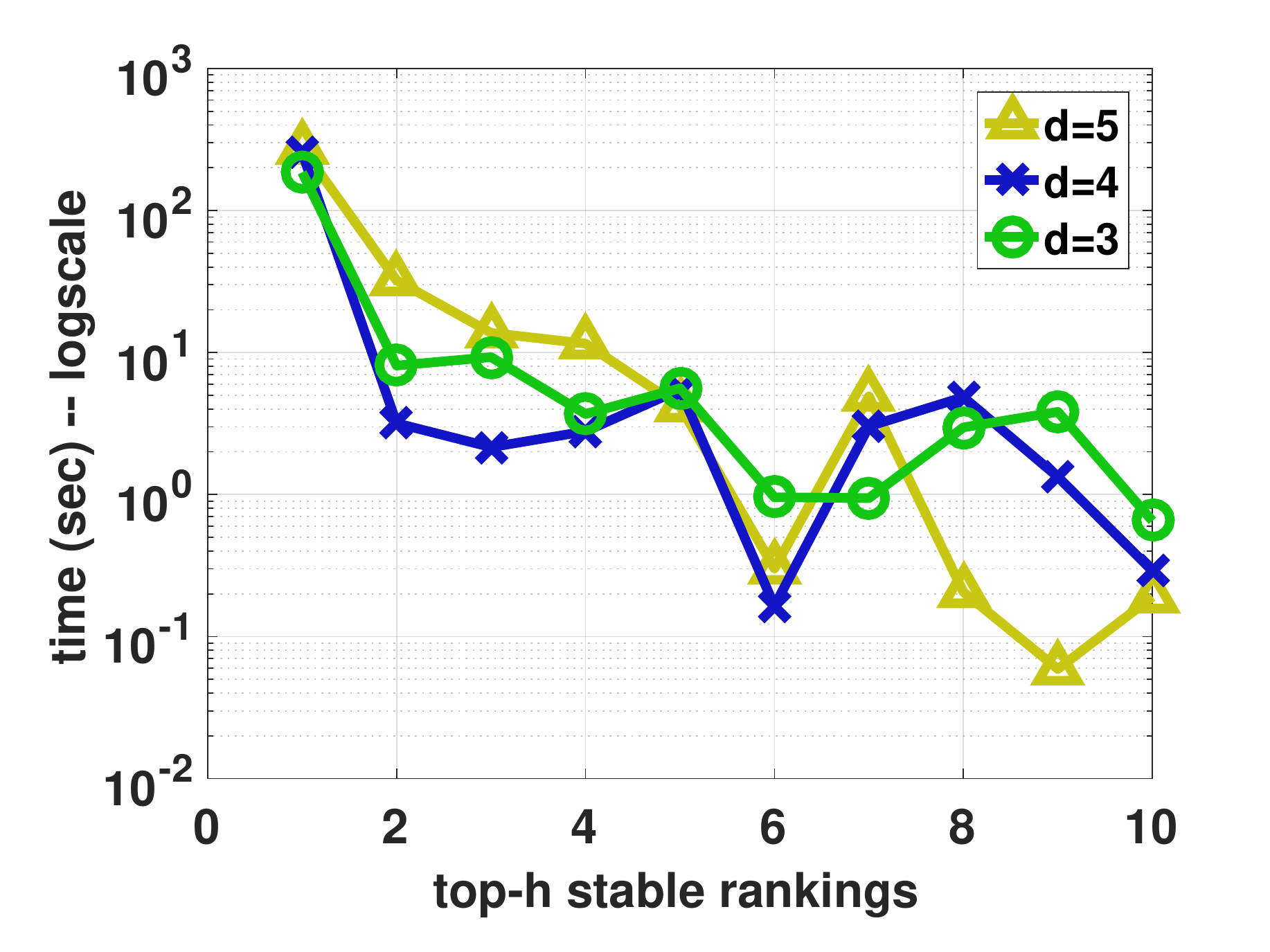}
        \vspace{-8mm}\caption{MD: stable rankings, impact of number of attributes ($d$).}
        \label{fig:BNMDVdt}
    \end{minipage}
    \hspace{1mm}
	\begin{minipage}[t]{0.23\linewidth}
        \centering
        \includegraphics[scale=0.24]{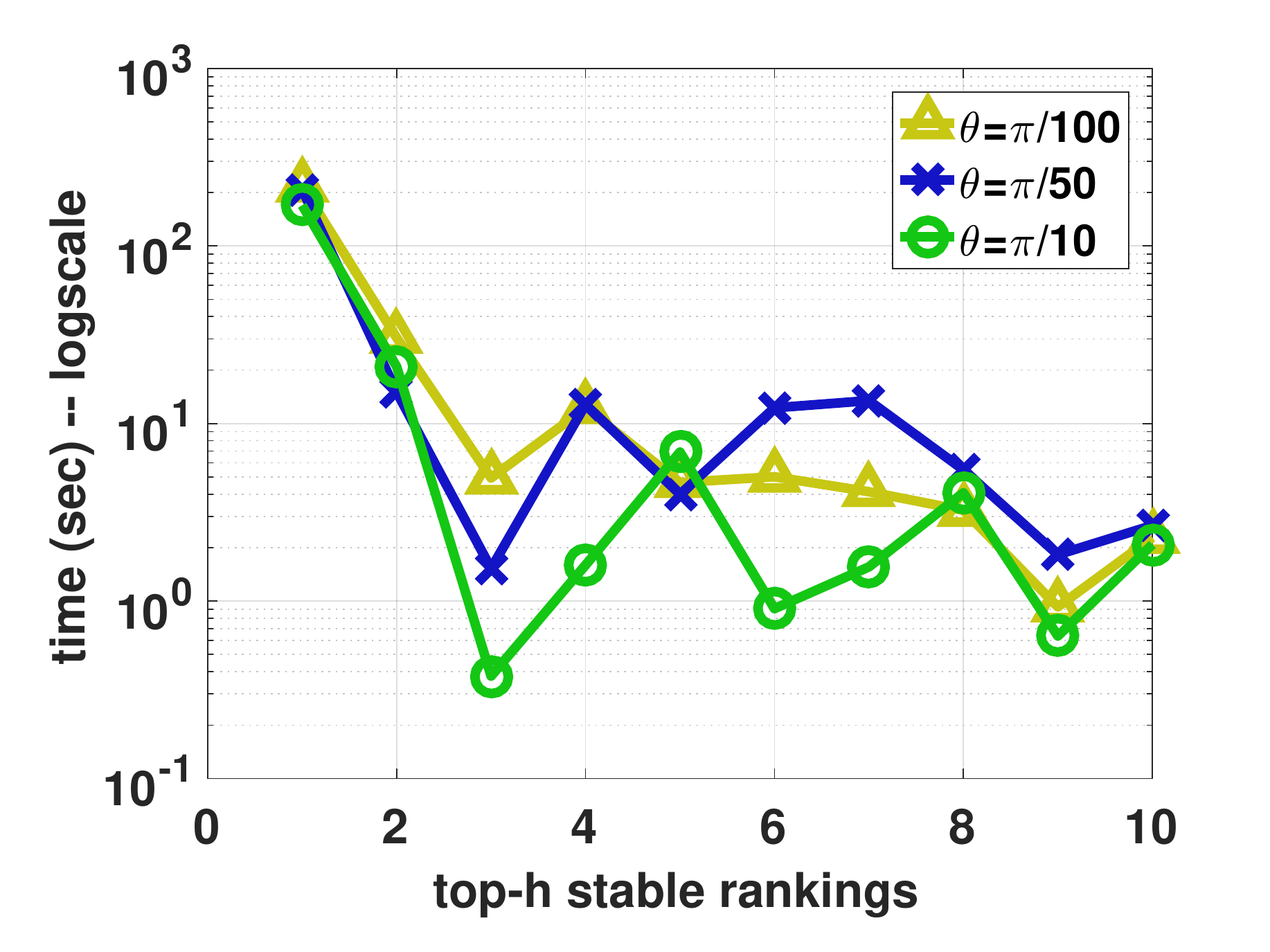}
        \vspace{-8mm}\caption{MD: stable rankings, impact of width of region of interest ($\theta$).}
        \label{fig:BNMDVtheta_t}
    \end{minipage}   
    \hspace{1mm}  
    \begin{minipage}[t]{0.23\linewidth}
        \centering
        \includegraphics[scale=0.24]{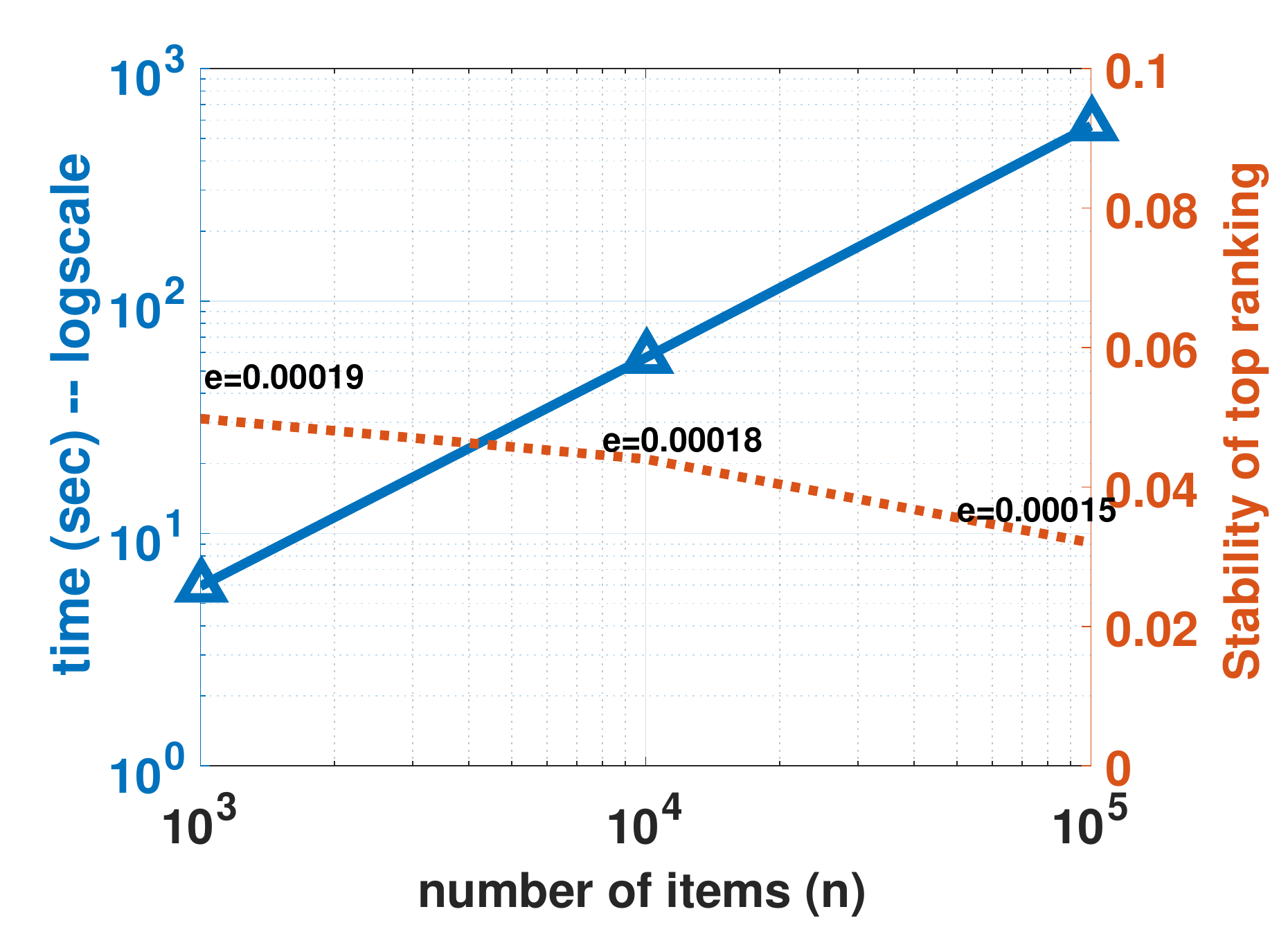}
        \vspace{-8mm}\caption{\getnextr: stable top-$k$ items, impact of dataset size ($n$).}
        \label{fig:BNRandn}
    \end{minipage}
    \hspace{1mm}
    \begin{minipage}[t]{0.23\linewidth}
        \centering
        \includegraphics[scale=0.24]{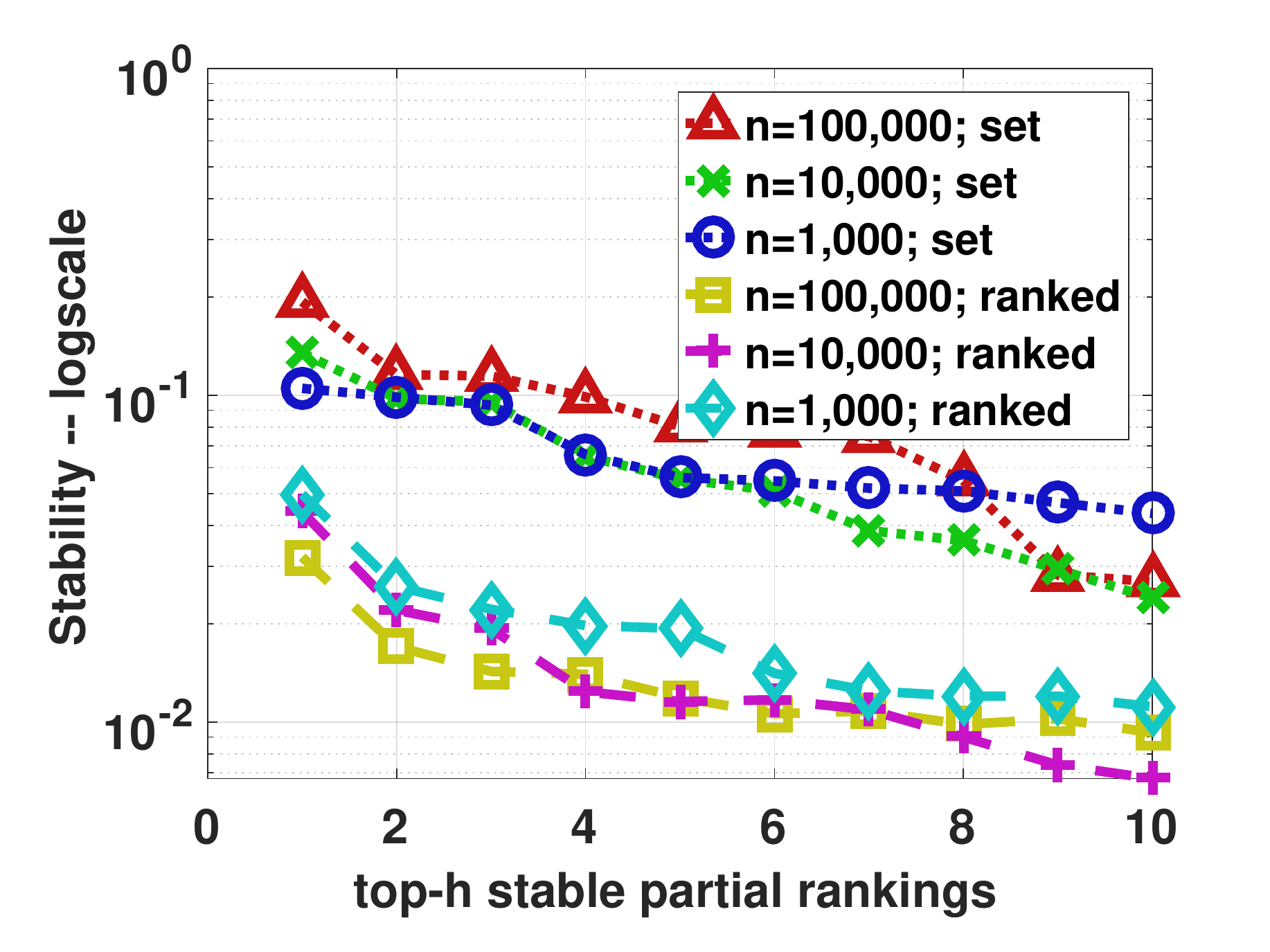}
        \vspace{-8mm}\caption{\getnextr: stable top-$k$ items, impact of dataset size ($n$) on stability.}
        \label{fig:BNRandns}
    \end{minipage}
\vspace{-5mm}
\end{figure*}
\begin{figure*}[ht]
    {\revision
    \begin{minipage}[t]{0.23\linewidth}
        \centering
        \includegraphics[scale=0.24]{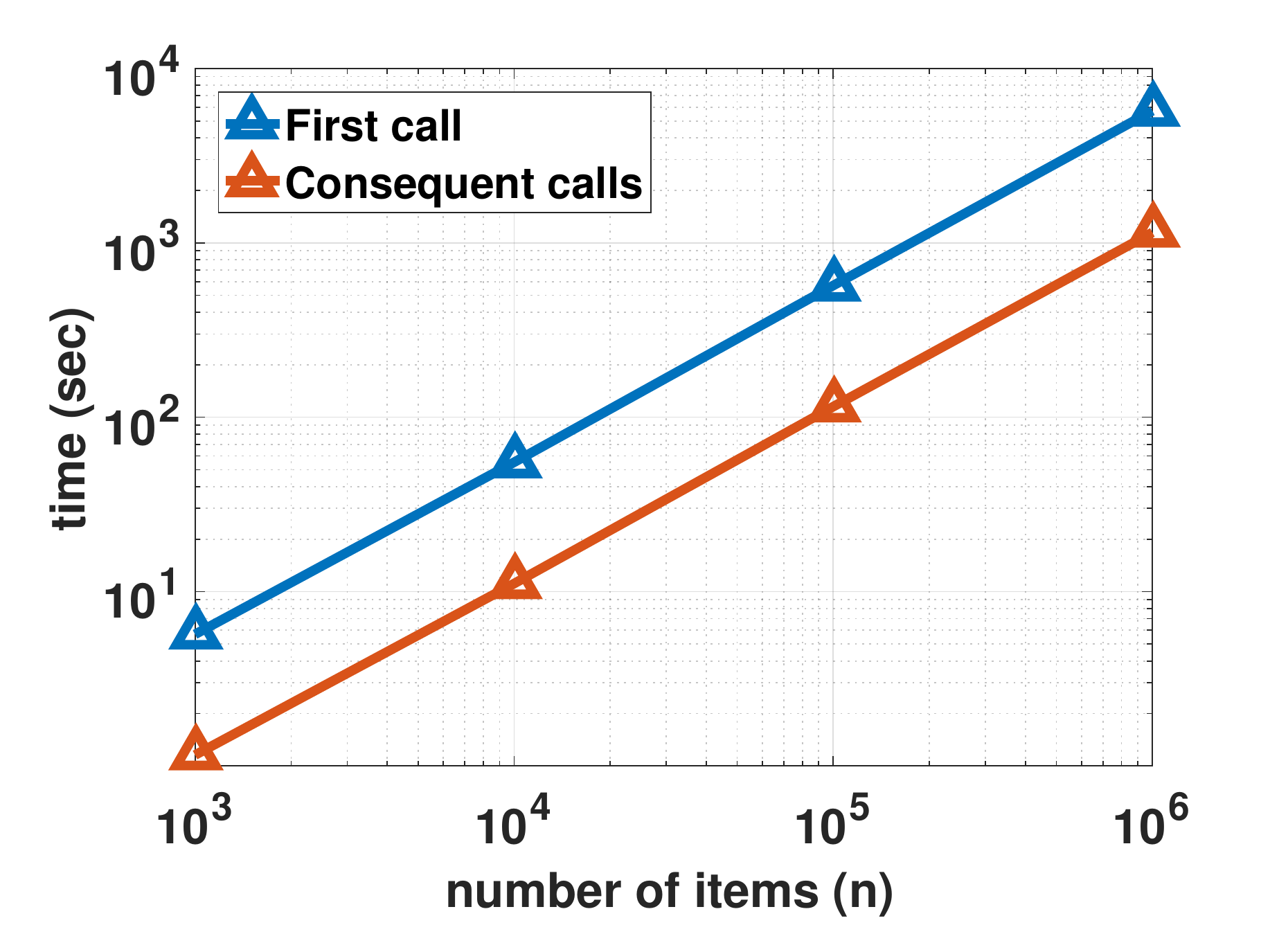}
        \vspace{-8mm}\caption{\revision DoT, \getnextr: stable top-$k$ items, Impact of dataset size ($n$)}
        \label{fig:DoT}
    \end{minipage}
    }
    \hspace{1mm}
    \begin{minipage}[t]{0.23\linewidth}
        \centering
        \includegraphics[scale=0.24]{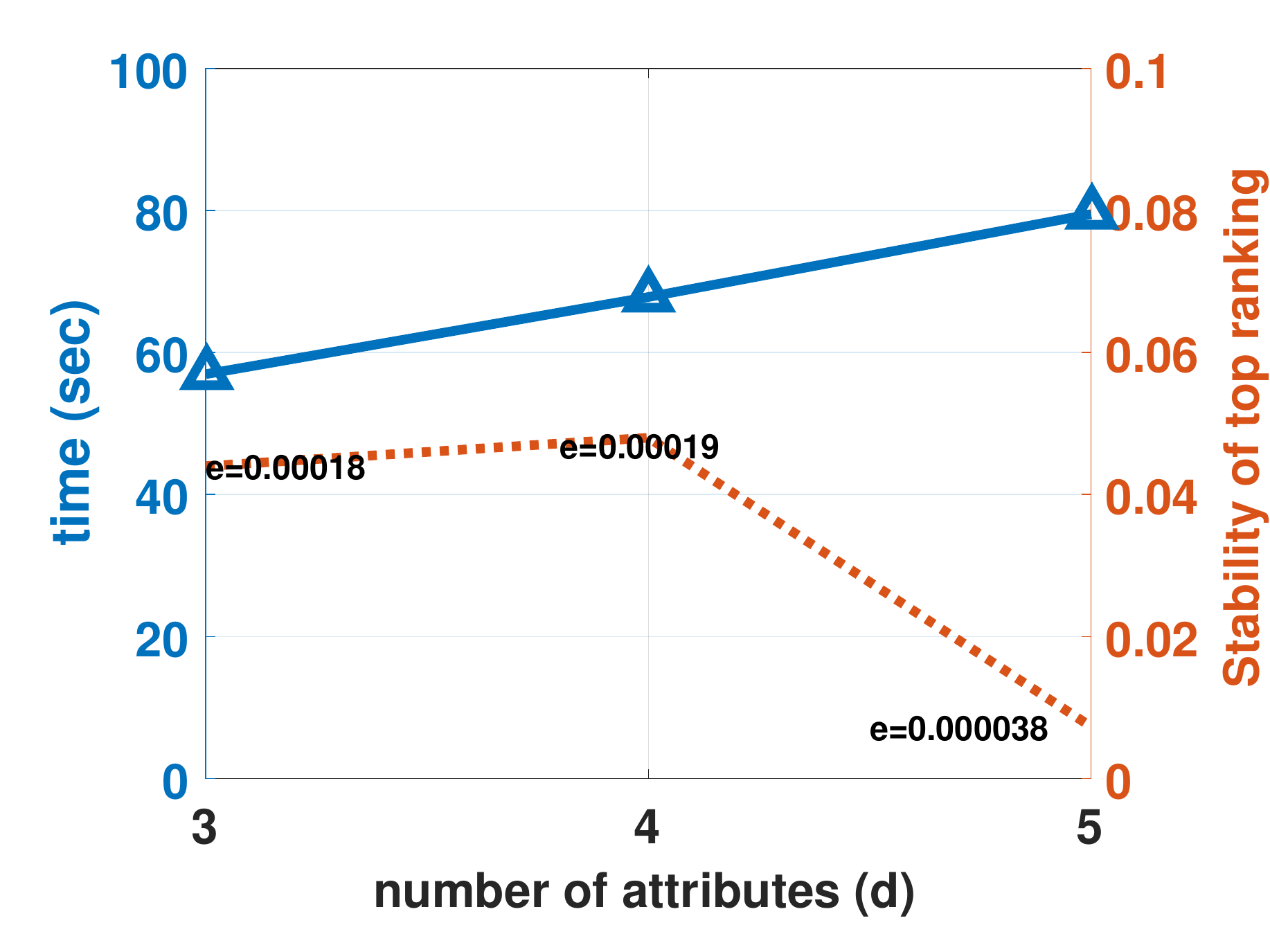}
        \vspace{-8mm}\caption{\getnextr: stable top-$k$ items, Impact of number of attributes ($d$)}
        \label{fig:BNRanddtTrue}
    \end{minipage}
    \hspace{1mm}
    \begin{minipage}[t]{0.23\linewidth}
        \centering
        \includegraphics[scale=0.24]{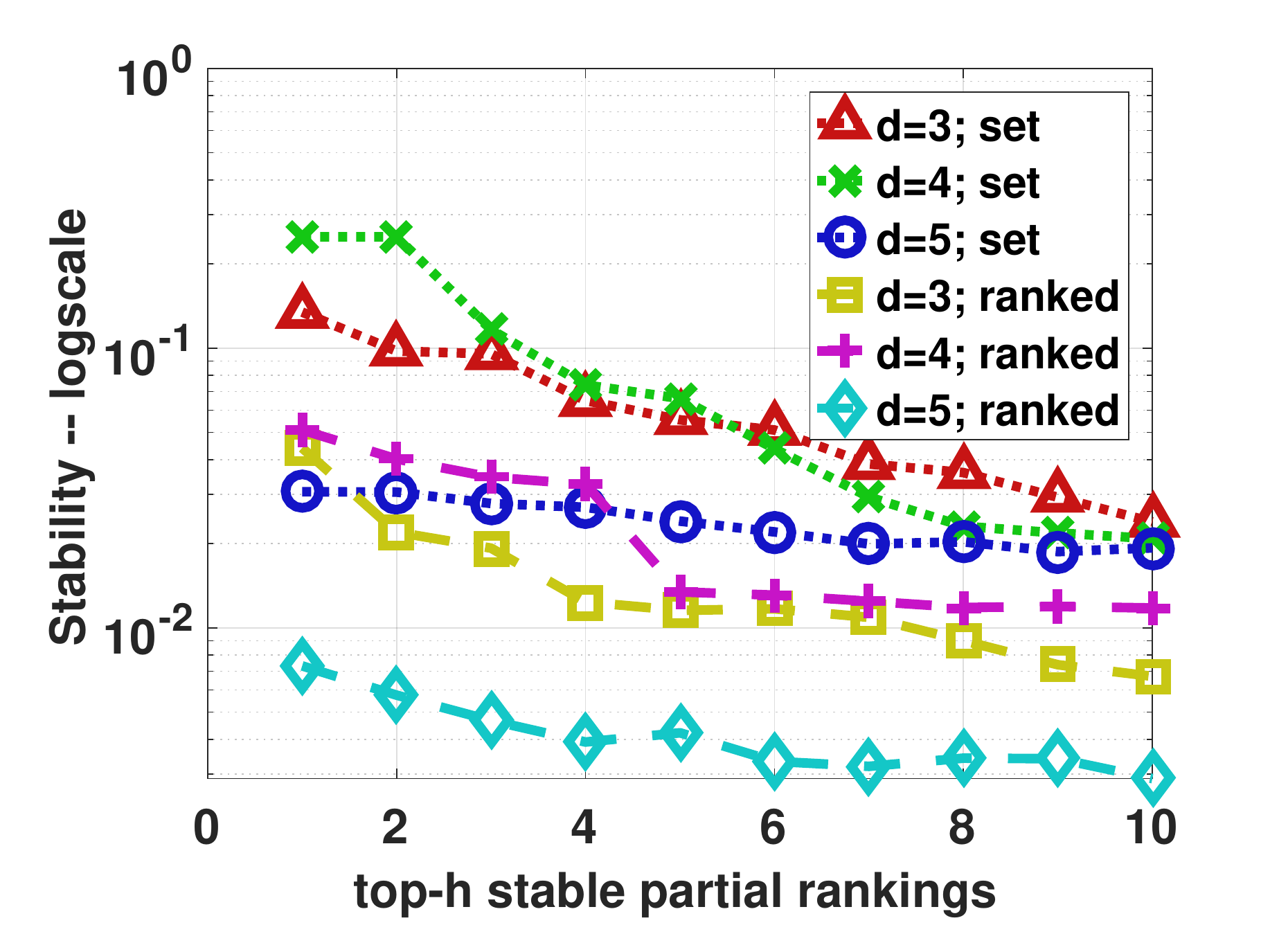}
        \vspace{-8mm}\caption{\getnextr: stable top-$k$ items, Impact of number of attributes ($d$)}
        \label{fig:BNRandds}
    \end{minipage}
    \hspace{1mm}
    \begin{minipage}[t]{0.23\linewidth}
        \centering
        \includegraphics[scale=0.24]{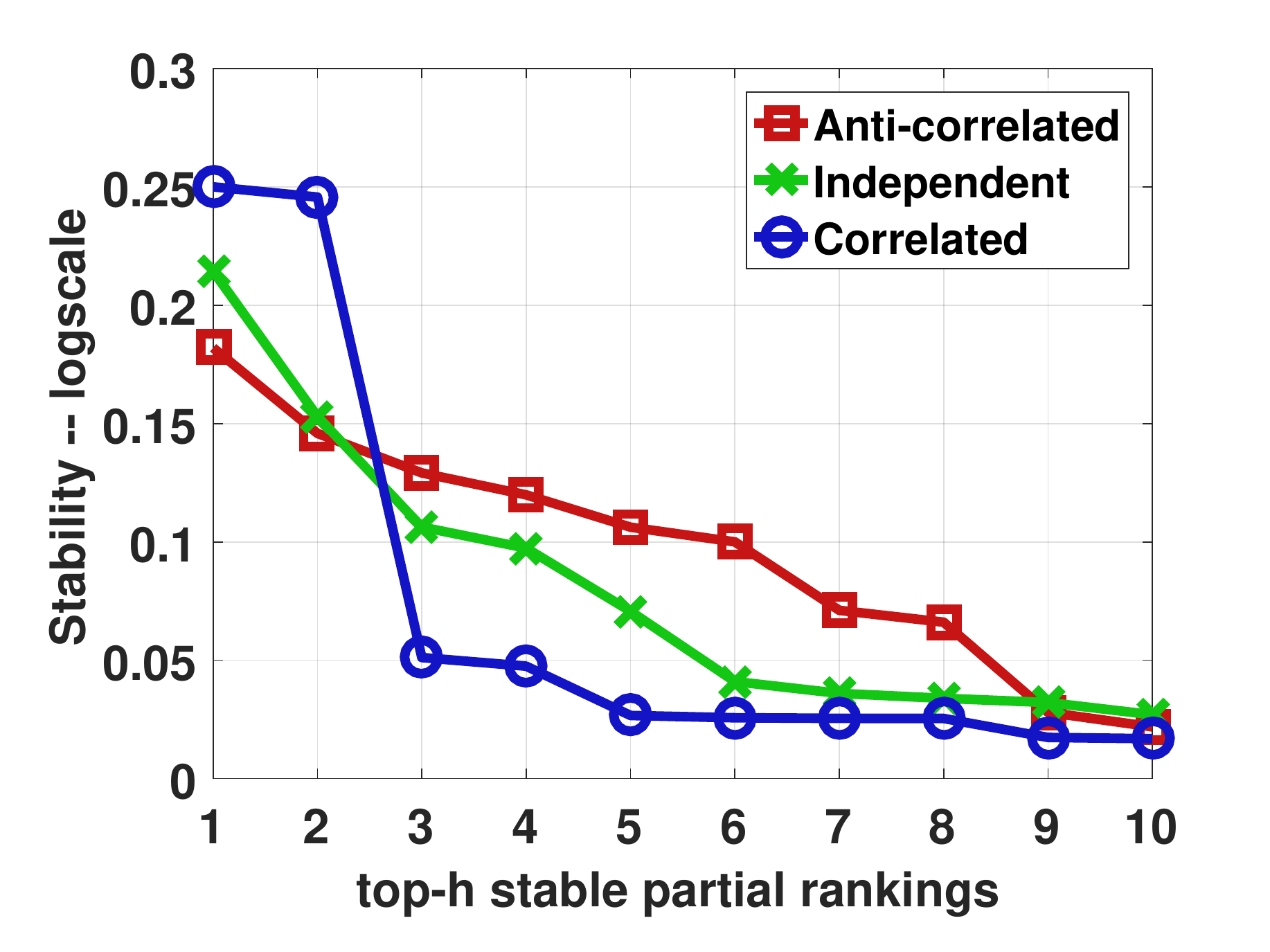}
        \vspace{-8mm}\caption{\revision Synthetic data, \getnextr: stable top-$k$ items, Impact of correlation}
        \label{fig:Cor}
    \end{minipage}
\submit{\vspace{-5mm}}
\end{figure*}

%\clearpage
%\submit{\vspace{-2mm}}
\subsection{Algorithm performance}\label{subsec:expperformance}
To evaluate the efficiency of our algorithms, we use the {\em Blue Nile} dataset, which consists of $116,300$ items over $5$ ranking attributes.  To vary the number of items, we take random samples; to vary the number of dimensions to $d=k<5$ we project the first $k$ attributes.  We equally weight the attributes as the default function. %\gm{Would be nice to have a summary sentence... ``Our findings show that our algorithms are uniformly efficient, even for.. ''}
%After validating our proposal, next we evaluate the performance of the proposed algorithms. To be able to test the performance under different settings, we chose our collection of Blue Nile (BN) catalog that includes $116,300$ diamonds over $5$ ranking attributes. In the experiments, we consider the ranking that equally weights all attribute as the default ranking.

\stitle{2D:}
First we study the impact of $n$, the database size, on the efficiency of \svtd to compute the stability of the default ranking (i.e. $w=\langle 1,1 \rangle$. We vary $n$ from $100$ to $100,000$, measuring both time and the stability of the default ranking (Figure~\ref{fig:bn2dsv}).  As stated in \S~\ref{sec:2d} computing the stability of a ranking in 2D is in $O(n)$. 
We find that the running time increased linearly and was only $0.12$ seconds for the largest data set.  We observe that the stability quickly drops from the order of $10^{-2}$ for $n=100$ to less than $10^{-6}$ for $n=100$K. 
This is because the number of ordering exchanges increase by $n$, leading to many small regions and low stability measures.
% This is because the larger the dataset is the more The number of ordering exchanges between items grows quickly as $n$ grows, leading to many small regions and low stability measures.

Next, we study the performance of \getnexttd under different database sizes. The first \getnext call needs to perform ray sweeping to construct the heap of ranking regions, while subsequent calls just remove the next most stable ranking from the heap. Therefore, in Figure~\ref{fig:BN2D2t}, we separate the first call from subsequent calls. As expected, as the number of items increases, the number of ordering exchanges increases and therefore, the efficiency of the operator drops. Also, subsequent \getnext calls are significantly faster than the first. Still, even for the largest setting (i.e., $n=100$K), the first call to the operator took less than $10$ seconds.

\vspace{3mm}
\stitle{MD:}
Next we study the performance of the stability verification algorithm, \sv, and the MD algorithm, \getnexttb, for producing the stable rankings. We vary the number of items ($n$), number of attributes ($d$), and width of the region of interest ($\theta$). 

First we evaluate stability verification (Figure~\ref{fig:bnmdsvn}). Choosing the default weight vector $w=\langle 1,1,1 \rangle$, while setting $d=3$, we initiate the stability oracle with a set of $1$M samples drawn from the entire function space $\mathcal{U}$ and vary the number of items from $100$ to $10$K.
The stability verification algorithm in MD needs to iterate over the sample points, counting those falling inside the ranking region, described as a set of $O(n)$ constraints. This took less than a minute for $n=10$K.
On the other hand, the stability of the default ranking immediately drops to near zero, even for $100$ items. Compared to 2D, this is due to the increase in the complexity of the function space. 
% Not clear:
%In 2D, the exchanges happen in a single function, whereas in MD, those are the hyperplanes that can create a complex arrangement. 

Next, we evaluate the performance of the \getnexttb operator under varying $n$, $d$, and $\theta$ (the width of $\ar$). The default values are $n=100$, $d=3$ and $\theta=\pi/100$. We use a set of $100$K samples from the region of interest for the measurements.
Figures~\ref{fig:BNMDVNt},~\ref{fig:BNMDVdt}, and~\ref{fig:BNMDVtheta_t} 
show the performance of the \getnexttb algorithm for the top 10 stable rankings for
varying (i) number of items from $n=10$ to $n=10$K,
(ii) number of attributes from $d=3$ to $d=5$, and (iii) width of the region from $\theta=\pi/10$ to $\pi/100$, respectively.

Overall, the running time decreases for subsequent calls. This is because the algorithm initially builds part of the arrangement before it finds the most stable ranking. 
%As the arrangement gets more complete, less time is required for finishing it for the next stable region. 
% cut for space:
%However, this is not a certain rule as many times the algorithm may spend a large time on building the regions that end up not being the output, which reduces its performance. This is confirmed in the zig-zags in Figures~\ref{fig:BNMDVNt},~\ref{fig:BNMDVdt}, and~\ref{fig:BNMDVtheta_t}. 
As shown in Figure~\ref{fig:BNMDVNt}, every \getnext call took up to several thousand seconds for the large setting of $10$K items. That is because of the complexity of the arrangement of $O(n^2)$ ordering exchanges which makes even the focus on the most stable region inefficient. In such complex situations, all the regions are very small and unstable, as too many ordering exchanges pass through a narrow region. 
Nevertheless, in a large setting, it is more reasonable to consider the top-$k$ items rather than the complete list. Our proposal for such settings is the randomized operator.

The next observation is that the running times are similar for different values of $d$  and $\theta$.  While the complexity of the space changes for the $O(n^2)$ ordering exchanges, the search is still done using a fixed set of samples and, using the Partition algorithm, only the subset of points falling into a region are used for constructing the arrangement in it. As a result, the lines in Figures~\ref{fig:BNMDVdt} and~\ref{fig:BNMDVtheta_t} show similar behaviors for different settings.

\stitle{Randomized algorithm:}
As the complexity of arrangement increases, \getnexttb becomes less efficient. On the other hand, when the number of items is large, users may be more interested in top-$k$ items: that is they may focus on the top of the ranked list. In \S~\ref{sec:randomized}, we proposed a Monte-Carlo-based randomized algorithm to handle these cases.
As the last set of experiments, we evaluate the performance of the randomized algorithm under different settings. We look at two models of top-$k$ items, (i) ranked top-$k$ items and (ii) top-$k$ sets. In (i) the user is interested in the orderings among the top-$k$ items, whereas in (ii) the user's interest is in the top-$k$ sets in the ranking lists. We consider a budget of $5,000$ samples (from the region of interest) for the first \getnextr call and $1,000$ for subsequent calls. The default values are number of items $n=10,000$, number of attributes $d=3$, the width of the region of interest $\theta=\pi/50$, and $k=10$.

Figures~\ref{fig:BNRandn} and~\ref{fig:BNRanddtTrue} show the running time of the first \getnextr call and the stability of the most stable ranking for varying the number of items from $1$K to $100$K, and the number of attributes from $3$ to $5$, while considering the  ranked top-$k$ items (the running times are similar for top-$k$ sets). The plots verify the scalability of the randomized algorithm for large settings, as it took a few minutes for $100$K items while the running times for $d=3$, $4$, and $5$ are similar.
Looking at right y-axis in Figure~\ref{fig:BNRandn}, despite the increase in the number of items from $1$K to $100$K, the stability of the most stable ranked top-$k$ did not noticeably decrease. This confirms the feasibility of considering the top-$k$ items for the large settings.

{\revision
Also, to evaluate the scalablity of our proposal for a very large setting, we use the DoT dataset and set the budget to 5K samples for the first \getnextr call and 1K for subsequent calls.
Similar to the previous experiment, we set $d$ to 3, the width of the region of interest to $\theta=\pi/50$, and consider top-$k$ sets for $k=10$, while varying the number of items up to one million.
Figure~\ref{fig:DoT} shows the performance of the algorithm for each setting.
As expected the run-time linearly increases with the number of items, while it takes on the order of an hour for the largest setting.
Note that the number of samples plays an important role in the performance of the algorithm: the higher the sampling budget, the more accurate the results, and the run-time is also higher.
This can be confirmed by comparing the lines for the first call (5000 samples) versus the consequent calls (1000 samples) of the primitive.
}

Figures~\ref{fig:BNRandns} and~\ref{fig:BNRandds} show the 
stability of the top-10 stable rankings for both ranked top-$k$ items and the top-$k$ sets. In both figures, the top-$k$ sets are more stable than the top-$k$ rankings. The reason is that the top-$k$ sets do not consider the ordering between the items, and thus the variety of possible outcomes is reduced compared to top-$k$ rankings.  An observation in Figure~\ref{fig:BNRandns} is the similarity of the stability distributions for different numbers of items, which, again, confirms the feasibility of considering the top-$k$ items for large settings.
In Figure~\ref{fig:BNRandds}, as expected, the number of attributes have a negative correlation with the stability of the top-$k$ items.

{\revision
\stitle{The effect of attribute correlation}
%\gm{can we get away with just the grey below and delete the orange?}
Finally, we study the effect of attribute correlation on the stability of the rankings. To do so, we use the synthetic datasets (independent, correlated, anti-correlated), each containing 10K items and $d=3$ attributes. Using a budget of 5000 samples for evaluation, we set the width of the region of interest to $\theta=\pi/50$, and $k$ to 10.  Figure~\ref{fig:Cor} shows the stability of the most stable top-$k$ sets.
We find that strong attribute correlation leads to a greater skew in the distribution of stable regions: the most stable regions have higher stability.  This is illustrated in Figure~\ref{fig:Cor} where we see that the correlated dataset results in the greatest maximum stability but also has the steepest slope as we descend from the most-stable to the 10th-most-stable top-$k$ set.  Accordingly, the independent dataset has a slightly lower stability most-stable region with a reduced slope, and the anti-correlated dataset displays the least skew in the stabilities.  This is expected since, in a dataset with highly correlated attributes, we are more likely to find items in dominance relationships with one another (i.e. the attributes of item X are greater than those of Y in all, or nearly all, dimensions). In that case, those items are almost always ranked in one way, reducing the number of feasible rankings, and resulting in a large number of relatively unstable rankings and a few highly stable rankings.
}

\submit{\vspace{-4mm}}
\section{Related Work}\label{sec:related}

{\revision Given a dataset with multiple attributes, ordering the items and choosing a subset to support decision making is challenging.  This has motivated a rich body of work on ranking~\cite{geerts2004relational,chaudhuri2009keyword,agrawal2006context,asudeh2016query,agarwal2000efficient}, top-$k$~\cite{fagin2003,ihab}, and skyline queries~\cite{skylineoperator,nanongkai2010,asudeh2017,DBLP:conf/icde/StoyanovichMR10}.  Broadly speaking, ranking and top-$k$ are employed when a user's preference in the form of a scoring function is available, while skyline queries are used when only the scoring attributes are known, but the scoring function is left unspecified.   To the best of our knowledge, no existing work considers a {\em range of acceptable scoring functions}, and discovers {\em stable rankings} within that range. In our work, the region of interest can be as narrow as a single scoring function, or as wide as the entire space of scoring functions.

The work on ranking and top-$k$ includes managing datasets with uncertainty and noise with respect to item existence or their attribute values~\cite{chaudhuri2004probabilistic, li2009unified,li2010ranking}, and using human computation to fill in missing information~\cite{pareto}.  While the work on probabilistic rankings considers uncertainty in the data, in our work we focus on uncertainty in the scoring function that reflects a user's preferences.}
\eat{
%
%Ranking has also been considered in spatial databases~\cite{hjaltason1995ranking} and in Web databases~\cite{asudeh2016query,qr2}.
%\julia{Is the Web databases work relevant?  I moved spatial databases lower, want to remove this sentence, OK?}
%}
%
%The work for finding outstanding items of a database with multiple attributes has two main directions: (i) when there exists a specific ranking function reflecting the user's preference, and (ii) in the absence of such preferences. Our proposal stands in the middle of these two as it is not limited to a specific function for reflecting user's preference, yet the preference is specified as a region of interest. The span of the region of interest can be as narrow as a single function or as wide as the whole function space.
%
%The first direction includes ranking and its varieties~\cite{geerts2004relational, li2005ranksql, chaudhuri2009keyword, agrawal2006context,agarwal2000efficient}.
%Ranking in spatial databases~\cite{hjaltason1995ranking} and web databases~\cite{asudeh2016query,qr2} are a few of such varieties.
%
}
There has been extensive effort on efficient processing of top-$k$ queries~\cite{ihab}: threshold-based algorithms~\cite{fagin2003} consider parsing presorted lists along each attribute, view-based approaches\cite{PREFER, das2006views} utilize presorted lists that are built on various angles of the function space, and indexing-based methods~\cite{chang2000onion} create layers of extreme points for efficient processing of queries. Ranking has also been considered in spatial databases~\cite{hjaltason1995ranking}.
%Ranking on probabilistic databases and noisy environments has also been the focus of works such as~\cite{li2010ranking, chaudhuri2004probabilistic, li2009unified}. The assumption in the probabilistic databases is uncertainty on the existence of items or in the attribute values. While probabilistic databases consider the uncertainty on the data we consider it on the function reflecting the user's preference.
\techrep{Recent work also considers ranking under fairness and diversity constraints~\cite{fairranking,CelisSV17,DBLP:conf/edbt/StoyanovichYJ18,ZehlikeB0HMB17}. For instance, given a weight vector for ranking,~\cite{fairranking} looks for a similar vector that provides fairness in its output.}

In the absence of a scoring function, the effort is on finding the set of potentially high-scoring representatives such as the  skyline~\cite{farhad, skylineoperator, asudeh2016discovering}, also known as the pareto-optimal set~\cite{pareto} --- the set of non-dominated items. %Since the size of the skyine can be large~\cite{asudeh2017}, works such as~\cite{asudeh2017,rrr,nanongkai2010,chester2014} look for small subsets of the skyline that minimize a measure of regret for using these sets for finding the outstanding items. 
{\revision Since the number of skyline points can be large~\cite{asudeh2017}, works such as~\cite{asudeh2017,rrr,nanongkai2010,chester2014,lin2007selecting,su2010top} look for smaller representative subsets. For example, \cite{lin2007selecting} finds a subset of $k$ skyline points that dominate the maximum number of points, while \cite{su2010top} picks the top-$k$ combinatorial skyline based on an importance ordering of the attributes. Also, extensive recent work~\cite{nanongkai2010,chester2014} aims to find a small subset of the skyline that minimizes some notion of regret.
A key difference between the stable top-$k$ set and these proposals is that a top-$k$ set is not necessarily a subset of the skyline.  %For example, the goal of the work on $k$-regret~\cite{chester2014} is to return $k$ points from the skyline that minimize, over all scoring functions, the score difference between the $k$-th selected point and the $k$-th best skyline point.  The key difference between our proposed methods and skyline-based methods is that we do not necessarily
}

In this paper, we used notions such as half-space, duality, and arrangement from combinatorial geometry that are explained in detail in~\cite{edelsbrunner,de2008computational}. Arrangement of hyperplanes, its complexity, construction, and applications are studied in \cite{orlik2013arrangements,grunbaum2003arrangements, schechtman1991arrangements, agarwal2000arrangements, edelsbrunner}.
{\revision
Geometric aspects of top-$k$ queries are presented in a recent tutorial~\cite{mouratidis2017geometric}.
}

%In the absence of user preferences, the effort in (ii) is on finding representatives such as skyline, or a subset of it. Majority of work in this direction is on efficient processing of skyline queries ~\cite{farhad, skylineoperator, asudeh2016discovering, pareto}. Since the size of such representatives can be large~\cite{asudeh2017}, works such as~\cite{asudeh2017,rrr,nanongkai2010,chester2014,lin2007selecting,su2010top} look for small subsets of data that minimize a measure of regret for using these sets for finding the outstanding items. For example, \cite{lin2007selecting} finds a set of $k$ skyline points that dominate the maximum number of items in the dataset and \cite{su2010top} defines the skyline notion on the combinations of the items and picks the top-$k$ combinatorial skyline based on a given ordering of the importance of the attributes. Among the set of work in this direction, there has been an extensive recent attention on regret-minimizing representatives~\cite{nanongkai2010}. For a given linear function, the regret of a subset, compared to the complete skyline is the difference between the score of the top item in the subset, compared to the one for actual top item. The notion of regret has been generalized by \cite{chester2014} to $k$-regret, by computing the regret of the top item of the subset versus the actual $k$-th top item.
%\submit{\vspace{-3mm}}
\section{Final Remarks}\label{sec:conclusion}
In this paper, we studied the problem of obtaining stable rankings for databases with multiple attributes when the rankings are produced through a goodness score for each item as a weighted sum of its attribute values.
A stable ranking is more meaningful than one susceptible to small changes in scoring weights, and hence engenders greater trust. 
We developed a framework that gives consumers the facility to assess the stability of a ranking and  enables producers to discover stable rankings.
We devised an unbiased function sampler that enables Monte-Carlo methods. We designed a randomized algorithm for the problem that works both for the complete ranking of items, as well as the top-$k$ partial rankings.
The experiments on three real datasets demonstrated the validity of our proposal.
{\revision
Our current definition of stability considers two rankings to be different if they differ in one pair of items.
An alternative is to allow minor changes in the ranking.
Similarly, we note that a weight vector is a single point in a stable region. It would be nice, for some applications, to characterize the boundaries of the stable region. We will consider these in future work.
}

% \section{Acknowledgments}
% \julia{I suggest to say: ``This work was supported in part by NSF Grants No. 1741022, 1741254 and 1741047.''  We can just make this a title note, rather than a separate section, to save space and make this acknowledgement more prominent.}
% %The work of Abolfazl Asudeh and H. V. Jagadish was supported by XXXX, Gerome Miklau was supported by NSF Grant No. 1741254, and Julia Stoyanovich's research was supported in part by NSF Grant No. 1741047.
\submit{\vfill\null}
\bibliographystyle{unsrt}
\bibliography{ref}
\techrep{
\appendix
\section{Coordinate system rotation}~\label{app:rotation}
In order to generate random functions in $\ar$, Algorithm~\ref{alg:sampui} models the region of interest as a $d$-spherical cap around the $d$-th axis.
Therefore, after picking a random vector $w$, it needs to {\em rotate} the space such that the $d$-th axis falls on the input vector $\rho$. This moves the drawn sample to the region of interest.
Such rotation is done via a so-called ``transformation matrix''~\cite{baker2012matrix}. 
Having such a $d\times d$ rotation matrix $M$, the result of rotation on a vector $w$ is a vector $w'$, generated as $w' = Mw$.
For example in $\mathbb{R}^2$, the following matrix rotates the coordinate system counterclockwise to an angle of $\theta$:
 \[
   M=
  \left[ {\begin{array}{cc}
   \cos \theta & -\sin\theta \\
   \sin\theta & \cos\theta \\
  \end{array} } \right]
\]

We use this matrix for deriving the rotation matrix we are looking for.
The idea is that instead of applying the the rotation at once, we can do the rotation on axes separately. For example, for $\mathbb{R}^3$, we first can fix the z-axis and do the rotation on the x-y plane and then fix the y-axis and do the rotation on the x-z plane.

The $d$ by $d$ matrix $M_i$, specified in Equation~\ref{eq:mi}, rotates the coordinate system on the $x_1$-$x_{i+1}$ plane counterclockwise to an angle of $\rho_i$. All the values in $M$ except the diameter, $M[1,i+1]$, and $M[i+1,1]$ are zero. Also, all the values on the diameter, except $M[1,1$ and $M[i+1,i+1]$ are one.

\begin{align}\label{eq:mi}
M_i = 
\begin{blockarray}{ccccccc}
1 & 2& \cdots & i+1&\cdots& d \\
\begin{block}{(cccccc)c}
  \cos\rho_i & 0 & \cdots & -\sin\rho_i &\cdots & 0& 1 \\
  0 & 1 & \cdots & 0 &\cdots & 0 & 2 \\
  \vdots & \vdots & \ddots & \vdots &\vdots & \vdots & \vdots \\
  \sin\rho_i & 0 & \cdots & \cos\rho_i &\cdots & 0 & i+1 \\
  \vdots & \vdots & \vdots & \vdots &\ddots & \cdots & \vdots \\
  0 & 0 & \cdots & 0 &\cdots & 1 & d \\
\end{block}
\end{blockarray}
 \end{align}

The remaining point is that in order to rotate the $d$-th axis to $\rho$, on all $x_1$-$x_{i+1}$ planes except the last one the rotations are counterclockwise, while for the $x_1$-$x_d$ plane the rotation is clockwise. We change $\rho_{d-1}$ to $(\pi/2 - \rho_{d-1})$ to make the last rotation also counterclockwise.
Algorithm~\ref{alg:rotate} shows the pseudocode of the function {\it Rotate}, used in Algorithm~\ref{alg:sampui}. Algorithm~\ref{alg:rotate} is in $O(d^3)$, which is negligible for small values of $d$.

\vspace{3mm}
\begin{algorithm}[!h]
\caption{{\bf Rotate} \\
		 {\bf Input:} vector $w$ and ray $\rho$ (in form of $d-1$ angles)\\
		 {\bf Output:} vector $w'$
		}
\begin{algorithmic}[1]
\label{alg:rotate}
	\STATE $w' = w$
    \STATE $\rho_{d-1} = \pi/2 - \rho_{d-1}$
    \FOR{$i=d-1$ down to $1$}
    	\STATE $M_i = $ Equation~\ref{eq:mi}
    	\STATE $w' = M_i\times w'$
    \ENDFOR
    \STATE {\bf return} $w'$
\end{algorithmic}
\end{algorithm}
}
\end{document}